\documentclass{article}

    \PassOptionsToPackage{numbers, compress}{natbib}
 \usepackage[preprint]{neurips_2026}

\usepackage[utf8]{inputenc} 
\usepackage[T1]{fontenc}    
\usepackage{hyperref}       
\usepackage{url}            
\usepackage{booktabs}       
\usepackage{amsfonts}       
\usepackage{nicefrac}       
\usepackage{microtype}      
\usepackage{xcolor}         
\usepackage[utf8]{inputenc} 
\usepackage[T1]{fontenc}    
\usepackage{hyperref}       
\usepackage{url}            
\usepackage{booktabs}       
\usepackage{amsfonts}       
\usepackage{nicefrac}       
\usepackage{microtype}      
\usepackage{xcolor}         
\usepackage{amsthm}
\newcounter{theorem}

\newtheorem{corollary}[theorem]{Corollary}

\newtheorem{thm}[theorem]{Theorem}

\usepackage{amsmath}

\usepackage{subcaption}

\usepackage{float}

\usepackage{titletoc}
\usepackage{tocloft}
\usepackage{minitoc}
\usepackage{wrapfig}
\usepackage[percent]{overpic}

\title{The Bayesian Origin of the Probability Weighting Function in Human Representation of Probabilities}

\author{
  Xin Tong\textsuperscript{1}\thanks{Contact: \texttt{xtong@lst.uni-saarland.de}} \quad
  Thi Thu Uyen Hoang\textsuperscript{1}\quad  
  Xue-Xin Wei\textsuperscript{2}\quad  
  Michael Hahn\textsuperscript{1}\thanks{Contact: \texttt{mhahn@lst.uni-saarland.de}} \\
  \textsuperscript{1}Saarland University \quad
  \textsuperscript{2}The University of Texas at Austin
}

\begin{document}

\maketitle

\begin{abstract}
Humans systematically misrepresent probability in a stereotyped inverse-S pattern. It has been documented for decades, but its origin remains unexplained. We propose a Bayesian encoding–decoding account in which probabilities are represented by noisy internal signals and decoded by Bayes-risk minimization. For bounded probability stimuli, we show that distortion decomposes into boundary regression, likelihood repulsion, and prior attraction, yielding a key prediction: the classic inverse-S-shaped weighting pattern implies a U-shaped allocation of encoding precision with greater sensitivity near 0 and 1. Across judgment of relative frequency, lottery pricing, and risky choice, this U-shape is recovered from data without imposing any functional form on the encoding, and our framework outperforms deterministic weighting functions, bounded log-odds models, uniform-encoding Bayesian accounts, and matched efficient-coding models on held-out data. In a new dot probability estimation experiment with bimodal stimulus statistics, the recovered prior tracks the new distribution while the recovered encoding remains U-shaped. Together, these results identify the inverse-S-shaped probability weighting function as the joint product of a stable U-shaped encoding and a flexible prior, integrated by optimal Bayesian decoding.\footnote{Code: \url{https://anonymous.4open.science/r/probability-distortion/}}
\end{abstract}

\section{Introduction}
Human's representation of probability is distorted. A 10\% chance of an event may feel higher than 10\%, while a 90\% chance may feel lower than 90\%. Prospect Theory \citep{kai1979prospect} formalized this overweighting of small probabilities and underweighting of large ones as an inverse-S-shaped probability weighting function (Figure~\ref{fig:intro}A). Probability weighting function represents one of the most robust empirical findings in behavioral decision research \citep{ruggeri2020replicating}, and is central for explaining anomalies such as the Allais paradox \citep{allais1953comportement}.


A fundamental question remains: what is the origin of the probability weighting function? Classical parametric forms \citep{prelec1998probability,ZM2012} describe its shape but do not explain it. A growing body of work points to a more fundamental answer: that the distortion reflects how the brain encodes probability, rather than how decisions are computed from it. Yet the literature offers strikingly different proposals for what that representation looks like. Some accounts attribute the distortion to a log-odds transformation of noisy estimates \citep{zhang2020,khaw2021cognitive}. Others locate it in regression away from response boundaries, with encoding itself unbiased \citep{fennell2012uncertainty,bedi2025probability}. Still others derive it from efficient coding, in which the encoding precisely matches the prior distribution of probabilities in the environment \citep{frydman2023source}. Each proposal implies a different picture of how the brain represents probability, and each has been validated on a different task. Whether they are competing explanations, complementary mechanisms, or special cases of a common principle has not been resolved.

In this paper, we show that these three proposals are not competitors but special cases of a single Bayesian encoding–decoding framework on a bounded domain.
Our main theoretical result decomposes distortion into boundary regression, likelihood repulsion, and prior attraction. This decomposition yields a key prediction: the classic inverse-S weighting pattern implies a U-shaped allocation of encoding precision, with sensitivity peaking near 0 and 1. We confirm this prediction empirically across three behavioral paradigms(judgment of relative frequency, lottery pricing, and risky choice), recovering the U-shape from data without imposing any functional form on the encoding.
Across paradigms, the data favor our framework's predictions over those of any single previous account, including deterministic weighting, single-anchor log-odds models, uniform-encoding Bayesian models, and matched efficient-coding accounts.
Beyond the encoding, the framework's second ingredient, the prior, is tested in a new dot probability estimation experiment with bimodal stimulus statistics. The recovered prior tracks the stimulus distribution, while the recovered encoding remains U-shaped. Together, these findings characterize probability distortion as the joint product of a stable U-shaped encoding and a flexible prior that tracks the environment.

\begin{figure}[htbp]
    \begin{center}    
    \includegraphics[trim=0 245 0 0, clip, width=1.0\linewidth]{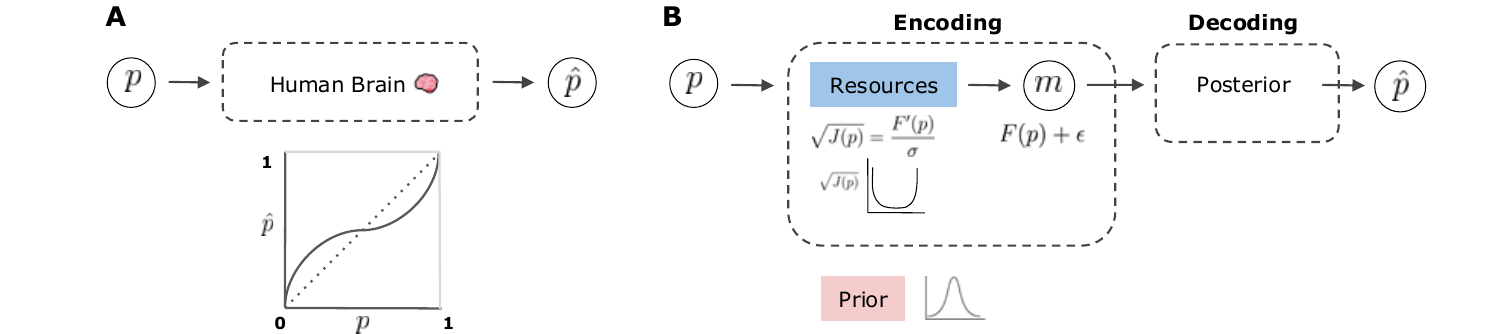}
    \end{center}
\caption{Distorted probability perception and our Bayesian model that explains such perceptual distortion.
(A) Human decisions systematically distort probabilities, producing the inverse S-shaped probability weighting function central to Prospect Theory.
(B) We propose that the inverse-S weighting pattern in (A) emerges from optimal decoding of a U-shaped allocation of encoding resources, recovered empirically from behavior. The Bayesian encoding–decoding framework: probabilities $p$ are encoded noisily, combined with a prior, and decoded into perceived probabilities $\hat{p}$.
}
\label{fig:intro}
\end{figure}

\section{Background and Relevant Work}
\label{sec:related-work}
Existing accounts of probability distortion fall into two broad families. The first formalizes the inverse-S pattern through a deterministic parametric function, $w(p)$. The second treats it as a consequence of noisy internal encoding. Within each family, models differ substantially in their commitments and in the data they have been validated against. We review both, with particular attention to the Bounded Log-Odds (BLO) model, which is the most quantitatively successful account to date and the principal benchmark in our empirical comparisons.

A body of work in economics and psychology models probability distortion with parametric forms of the probability weighting function. Classical examples include the inverse-S shaped forms in Prospect Theory \citep{tversky1992advances}, the Prelec function \citep{prelec1998probability}, and the Linear-in-Log-Odds (LILO) model \citep{ZM2012}. As an illustration, the LILO model assumes that under the log-odds transform, the perceived probability $\widehat{p}$ is linear in the true probability $p$. These models successfully capture behavioral regularities such as overweighting of small probabilities and underweighting of large ones, but they are primarily descriptive. The question of where this shape comes from has motivated a second line of work.

A second line of work focuses on the idea that probabilities are encoded imprecisely in the brain, with distortion arising as a systematic consequence of noisy internal representations. Within this view, accounts diverge on where the noise enters and what structure the encoding has. Regression-based accounts locate the bias at the response stage, attributing distortion to regression away from the scale boundaries while leaving encoding itself unbiased \citep{fennell2012uncertainty,bedi2025probability}. Log-odds based accounts locate the bias at the encoding stage, assuming that internal representations are approximately linear in log-odds space \citep{khaw2021cognitive,Khaw2022StateDependent}. Efficient-coding accounts further constrain the encoding to match the prior distribution of probabilities in the environment \citep{frydman2023source}, while related work has examined optimal inference under noisy encoding more broadly \citep{juechems2021optimal,enke2023cognitive}. These models have only been evaluated against limited behavioral data, and their differing predictions have not been compared on common ground, either theoretically or empirically.

Among encoding-based accounts, the Bounded Log-Odds (BLO) model \citep{zhang2020} stands out as the most quantitatively developed. It has been validated across multiple behavioral paradigms and consistently outperforms fixed probability weighting functions, making it the natural benchmark against which any new account must be measured. Because our empirical comparisons (Section~\ref{sec:empirical-validation}) treat BLO as the primary reference point, we review its mechanism in detail here. BLO assumes that probability $p$ is first mapped to log-odds and truncated to a bounded interval $[\Delta^-,\Delta^+]$ (Eq.~\ref{eq:blo-Lambda}). The clipped value is linearly mapped to an encoding $\Lambda(p)$ on an interval $[-\Psi,\Psi]$ (Eq.~\ref{eq:blo-second-map}), then combined with an anchor $\Lambda_0$ via a weight that depends on a free parameter $\kappa > 0$ (Eq.~\ref{eq:hat-Lambda}). This encoding is subject to Gaussian noise and decoded by applying the inverse log-odds function (Eq.~\ref{eq:blo-final-trasnform}).
\begin{align}
\lambda(p) &= \log \tfrac{p}{1-p}, \qquad
\Gamma(\lambda) = \min{(\max(\lambda, \Delta^-), \Delta^+}) \label{eq:blo-Lambda}\\
\Lambda(p) &= \tfrac{\Psi}{(\Delta_+ - \Delta_-)/2} \left(\Gamma(\lambda(p)) - \tfrac{\Delta_-+\Delta_+}{2}\right) \label{eq:blo-second-map}\\
\hat{\Lambda}_\omega(p) &= \omega_p \cdot \Lambda(p) + (1-\omega_p)\cdot\Lambda_0, \qquad \omega_p = \tfrac{1}{1+\kappa V(p)},\qquad V(p)\propto p(1-p), \label{eq:hat-Lambda}\\
\hat{\pi}(p) &= \lambda^{-1}(\hat{\Lambda}_\omega(p) + \epsilon_\lambda), \quad \epsilon_\lambda \sim \mathcal{N}(0,\sigma_\lambda^2). \label{eq:blo-final-trasnform}
\end{align}

\section{A general Bayesian framework for perceived probability}\label{sec:framework}

\subsection{Encoding–decoding framework}
Across the encoding-based accounts reviewed in Section~\ref{sec:related-work}, models differ in their choice of encoding function and prior, but share a common structure: probabilities are encoded into a noisy internal signal, and a perceived probability is decoded from this signal. We make this shared structure explicit by formalizing it as a single Bayesian encoding–decoding framework. Within this framework, the existing accounts, including BLO, regression-based, log-odds-based, and efficient-coding models, correspond to particular choices of two ingredients: the encoding function $F$ and the prior $P_{\text{prior}}$. This unification will allow us in Section~\ref{sec:framework-decomposition} to derive analytical predictions that distinguish the existing accounts and to identify a key empirical signature (a U-shaped allocation of encoding precision) that any successful account must satisfy.

Formally, during the encoding phase, the stimulus $p$ is encoded into a noisy internal signal $m$, modeled as a one-dimensional transformation with additive Gaussian noise: \begin{equation} m = F(p) + \epsilon, \quad\text{where } \epsilon \sim \mathcal{N}(0, \sigma^2) \end{equation} where $F : [0,1] \rightarrow \overline{\mathbb{R}}$ is strictly monotone increasing and smooth. We interpret $m$ as an abstract representation of the neural code for probability. The slope of $F$ at $p$ determines how distinguishable nearby probability values are after encoding: a steeper $F$ stretches a small interval around $p$ into a wider range of $m$ values, making them more separable despite the noise. This is captured by the Fisher Information $\mathcal{J}(p) = (F'(p))^2/\sigma^2$. Following \cite{hahn2024unifying}, we will refer to $\sqrt{\mathcal{J}(p)}$ as the \textbf{encoding resources} allocated to $p$.

During the decoding phase, given $m$ and a prior distribution $P_{\text{prior}}(p)$, Bayesian inference yields a posterior $P(p \mid m)$. We assume the decision-maker derives a point estimate $\hat{p}(m)$ as the posterior mean, which minimizes the Bayes risk for the mean-squared loss function. 

In this encoding-decoding framework, the model's behavior is therefore governed by two ingredients: the encoding function $F$ and the prior $P_{\text{prior}}$, both of which we will infer from behavioral data. Of these, the shape of the encoding resources $\sqrt{\mathcal{J}(p)}$ is the central quantity in our analysis, and we will show that the observed inverse-S weighting pattern places strong constraints on its shape.
In Section~\ref{sec:framework-decomposition} we make this correspondence explicit: each existing account is instantiated as a specific choice of $F$ and $P_{\text{prior}}$, and we derive the qualitatively different predictions they make. This sets up the cross-task empirical comparison in Section~\ref{sec:empirical-validation}.

\subsection{Predictions from the bias decomposition}\label{sec:framework-decomposition} 

We analyze the framework's predictions through the bias of the Bayesian estimate, i.e.,  the average deviation of the estimate $\hat{p}$ from the true probability across trials:
\begin{equation}
\operatorname{Bias}(p) := \mathbb{E}[\hat{p} \mid p] - p
\end{equation}

Prior work on the bias of Bayesian perceptual models has largely focused on unbounded or circular stimulus spaces, or on stimuli far from a boundary \citep{hahn2024unifying, wei2017lawful, prat2021bias, morais2018power, stocker2006noise}. Probability distortion, by contrast, takes place on a bounded interval $[0,1]$, and the boundaries are not incidental: they are central to longstanding accounts of probability weighting \citep{fennell2012uncertainty, bedi2025probability}.
We therefore derive a bias decomposition for the bounded interval, extending the unbounded and one-boundary results of \cite{hahn2024unifying} to the two-boundary setting, with the boundary regression term now capturing effects at both endpoints.(Theorem~\ref{thm:bayesian-bias}; proof in Appendix~\ref{app:proofs-bayesian})
\begin{thm}\label{thm:bayesian-bias}
There are functions $A_{1,\sigma}, A_{2,\sigma}, A_{3,\sigma} : [0,1] \rightarrow \mathbb{R}_+$ that are bounded, positive, and satisfy
    \begin{equation}\label{eq:boundary-effect-vanishes-distance}
         A_{\dots, \sigma}(p) \leq \phi\left(\frac{\sigma}{\min\{|F(p)-F(0)|, |F(p)-F(1)|\}}\right)
    \end{equation}
    for some nondecreasing function $\phi : \mathbb{R} \rightarrow \mathbb{R}$ with $\lim_{x\downarrow 0}\phi(x) = 0$,
    across all $p \in (0,1); \sigma > 0$.
Then, at any $p \in (0,1)$, the Bayesian model has the bias, as $\sigma \rightarrow 0$:
    \begin{align*}\label{eq:bayesian-bias-app}
  \operatorname{Bias}(p) =      \underbrace{A_{1,\sigma}(p) \cdot  \frac{\operatorname{sign}(0.5-p)}{\sqrt{\mathcal{J}(p)}}}_{\text{Regression from Boundary}} +        \underbrace{(1-A_{2,\sigma}(p)) \cdot \frac{\operatorname{d}}{\operatorname{d}p} \left(\frac{1}{\mathcal{J}(p)}\right)}_{\text{Likelihood Repulsion}} \\
+ 
        \underbrace{(1-A_{3,\sigma}(p)) \cdot \frac{1}{\mathcal{J}(p)} \cdot \frac{\operatorname{d}}{\operatorname{d}p} \left(\log P_{{prior}}(p)\right)}_{\text{Prior Attraction}} + \mathcal{O}(\sigma^4)
    \end{align*} 
\end{thm}

The decomposition isolates three distinct mechanisms of distortion that have appeared piecemeal across previous accounts but have not been separated within a single framework:  \textbf{boundary regression} pushes estimates inward, \textbf{likelihood repulsion} pushes them away from regions where encoding resources are concentrated, and \textbf{prior attraction} pulls them toward regions of high prior density. 
Under this decomposition, the regression-based, log-odds-based, and efficient-coding accounts are not competing explanations but special cases of a single structure, each built on a different subset of the three terms. Identifying them as components of a common decomposition allows us to derive predictions that distinguish them empirically, which we develop in the rest of this section and test in Section~\ref{sec:empirical-validation}.

The factors $A_{1,\sigma}, A_{2,\sigma}, A_{3,\sigma}$ all vanish when the noise $\sigma$ is small relative to the distance from $F(p)$ to $F(0)$ and $F(1)$ in the encoding space (Eq.~\ref{eq:boundary-effect-vanishes-distance}). Their role, however, differs across the three terms: far from the boundary, regression vanishes ($A_{1,\sigma} \to 0$), while repulsion and attraction reach full strength ($1 - A_{2,\sigma}, 1 - A_{3,\sigma} \to 1$). Near the boundary, the situation reverses: regression dominates while repulsion and attraction are suppressed.

 We examine each in turn, showing how it constrains the encoding $F$, the prior $P_{\text{prior}}$, or both, and how the resulting predictions can adjudicate among existing accounts.

\begin{figure}
  \centering
  \includegraphics[trim=0 12 0 0, clip, width=0.7\textwidth]{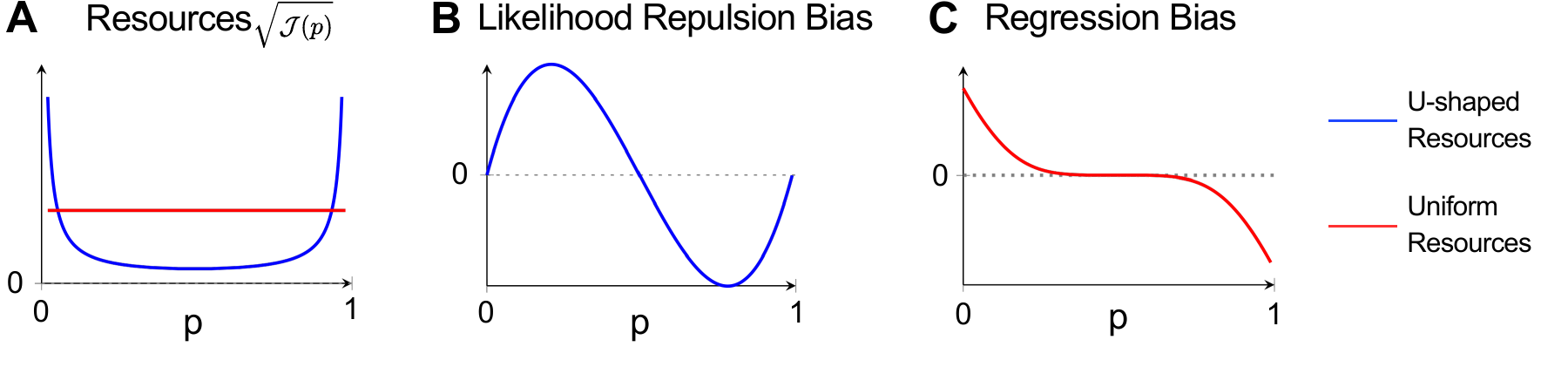}
  \caption{Impact of the encoding (resource allocation) on bias. (A) A U-shaped (blue) and uniform (red) encoding allocation resources $\sqrt{\mathcal{J}(p)}$. (B) The U-shaped encoding generates an S-shaped Likelihood Repulsion bias, $\propto \frac{\operatorname{d}}{\operatorname{d}p} (1/\mathcal{J}(p))$. (C) The uniform encoding generates a Regression bias  $\propto  A_{1,\sigma}(p) \cdot \frac{\operatorname{sign}(0.5-p)}{\sqrt{\mathcal{J}(p)}}$, maximized at 0,1. Both biases generate overestimation of small and underestimation of large probabilities, but with distinct shapes.}
  \label{fig:conceptual}
\end{figure}

\paragraph{Boundary regression.}
The first term in Theorem~\ref{thm:bayesian-bias} is the basis of regression-based accounts: \cite{fennell2012uncertainty} and \cite{bedi2025probability} attribute the inverse-S pattern primarily to boundary regression, treating encoding itself as uniform or otherwise unbiased. Under this view, distortion is a consequence of the bounded response scale rather than of how probabilities are encoded.

This account has a built-in constraint. Boundary regression is strongest near $0$ and $1$, whereas the inverse-S weighting function vanishes there \citep{prelec1998probability,tversky1992advances}. The two are compatible only if $\mathcal{J}(0), \mathcal{J}(1) \to \infty$, which is biologically implausible. Accounts that rely primarily on boundary regression therefore predict substantial distortion at the endpoints, in contradiction with the inverse-S pattern. We test this theoretical prediction empirically in Section~\ref{sec:empirical-validation}: uniform-encoding accounts should fit the data poorly near the endpoints.

\paragraph{Likelihood repulsion.}
The second term in Theorem~\ref{thm:bayesian-bias} pushes estimates away from stimuli with high Fisher Information, with sign and shape determined by the gradient of $1/\mathcal{J}(p)$.

Consider the first two terms in Theorem~\ref{thm:bayesian-bias}, boundary regression and likelihood repulsion, independent of the prior. As we showed in the previous paragraph, boundary regression alone cannot produce the inverse-S pattern, which vanishes at $p=0$ and $p=1$. When boundary regression is neutralized, the bias is dominated by the likelihood repulsion term $\propto \frac{d}{dp}(1/\mathcal{J}(p))$. This term is S-shaped (Figure~\ref{fig:conceptual}B) if and only if $\mathcal{J}(p)$ is U-shaped.

\begin{center}
\colorbox{gray!15}{%
  \parbox{0.98\linewidth}{\textbf{Prediction 1.} The classic inverse-S weighting pattern implies a U-shaped allocation of encoding resources, with $\sqrt{\mathcal{J}(p)}$ peaking near $0$ and $1$.}}
\end{center}

Prediction 1 is consistent with several existing models, although each model arrives at the U-shape resource allocation through a different assumption. Specifically, log-odds-based accounts \citep{zhang2020, khaw2021cognitive, Khaw2022StateDependent} commit to a fixed log-odds transform $F(p) = \log\frac{p}{1-p}$, which yields a U-shape with $\sqrt{\mathcal{J}(p)} \propto 1/p + 1/(1-p)$. U-shaped-FI accounts \citep{enke2023cognitive} allow more general encodings constrained only by U-shaped Fisher Information. Efficient-coding accounts \citep{frydman2023source} derive the U-shape by matching the encoding to a prior assumed itself to be U-shaped (so that $\sqrt{\mathcal{J}(p)} \propto P_{\text{prior}}(p)$, see Appendix~\ref{sec:frydman-jin-fi}). Within our framework, by contrast, U-shape is the necessary consequence of the inverse-S pattern, requiring no specific functional form on the encoding or prior. 
In Section~\ref{sec:empirical-validation}, we test both points: whether the U-shape can be recovered from data without imposing functional-form constraints, and whether the recovered shapes match the commitments of these alternative accounts.

As shown in Appendix~\ref{app:proofs-blo}, BLO in the limit of untruncated log-odds is equivalent at the level of bias to a Bayesian model with log-odds encoding and a particular unimodal prior, although it uses a sub-optimal decoder. 
This connection allows the empirical comparison in Section~\ref{sec:jrf} to assess the joint contribution of optimal decoding and flexible prior.

\paragraph{Prior attraction.}
The third term in Theorem~\ref{thm:bayesian-bias} pulls estimates toward regions of high prior density, with strength scaled by $1/\mathcal{J}(p)$. Unlike the previous two terms, prior attraction depends on $P_{\text{prior}}$ rather than on the encoding alone, and so its shape can change as $P_{\text{prior}}$ changes—for instance, as the decision-maker is exposed to a new stimulus distribution.

This term is shared, in some form, by all Bayesian accounts \citep{fennell2012uncertainty, bedi2025probability, frydman2023source}, but the accounts differ on how flexibly $P_{\text{prior}}$ is allowed to vary. BLO \citep{zhang2020} approximates prior attraction with a single anchor $\Lambda_0$ and does not accommodate general priors. Efficient-coding accounts \citep{frydman2023source} tie $P_{\text{prior}}$ to the encoding through $\sqrt{\mathcal{J}(p)} \propto P_{\text{prior}}(p)$; under this constraint, the posterior mean is biased \emph{away} from the prior mode rather than toward it \citep{wei2015bayesian}, so prior attraction is effectively absorbed into the encoding. Bayesian accounts with free priors \citep{fennell2012uncertainty, bedi2025probability} permit attraction toward arbitrary high-density regions.

These commitments make distinct predictions when the stimulus distribution changes. A flexible Bayesian prior should track the new distribution, producing a corresponding shift in the bias; a single anchor cannot reproduce such shifts; and a prior tied to the encoding can only adapt insofar as the encoding itself adapts, which is generally assumed to occur on slower timescales \citep{fritsche2020bayesian}.

\begin{center}
\colorbox{gray!15}{%
  \parbox{0.98\linewidth}{\textbf{Prediction 2.} Stimulus distribution changes  produce corresponding changes in the bias, mediated by the prior attraction term. Accounts with flexible Bayesian priors predict the largest such effects; accounts with single-anchor priors or prior–encoding matches predict weaker or no adaptation.}}
\end{center}

We test this prediction in Section~\ref{sec:empirical-validation-adaptation} using a dot-estimation experiment with bimodal stimulus statistics: participants are exposed to a distribution with two modes, and we ask whether their bias adapts in a manner consistent with each account's prior commitment.

\section{Empirical Validation of the Bayesian Framework}\label{sec:empirical-validation}

We test the framework's predictions across three behavioral paradigms. In each case, we instantiate the framework with different commitments about the encoding $F$ and prior $P_{\text{prior}}$, corresponding to the accounts surveyed in Section~\ref{sec:related-work}. The judgment of relative frequency (JRF) task tests Prediction 1: whether the inverse-S bias requires non-uniform encoding. Lottery pricing and risky choice test whether the same mechanism extends to economic decisions under risk. Finally, we collected data in a dot-estimation experiment with bimodal stimulus statistics to test Prediction 2: whether bias adapts to changes in the stimulus distribution.

\subsection{Validating Prediction 1 on Judgment of Relative Frequency (JRF) data}\label{sec:jrf}

We analyze data from \cite{ZM2012} and \cite{zhang2020}, where subjects ($N=86$) were asked to judge the percentage of dots of a target color in arrays of black and white dots on a gray background.\footnote{A complete description of the tasks and datasets used in this paper is available in Appendix~\ref{app:details-task-datasets}.}

We compare the parametric baseline \textbf{BLO} with three Bayesian variants, each instantiating a different commitment about the encoding while sharing a flexible non-parametric prior. \textbf{FreeP+UniformE} assumes uniform encoding $F(p)=p$, instantiating regression-based accounts \citep{fennell2012uncertainty,bedi2025probability}. \textbf{FreeP+BoundedLOE} uses a parametric bounded log-odds encoding, instantiating log-odds-based accounts \citep{khaw2021cognitive,Khaw2022StateDependent}. \textbf{FreeP+FreeE} places no functional form on the encoding, allowing the data to determine its shape. We additionally test the efficient-coding account that prior matches encoding via \textbf{FreeP+PriorMatchedE} \citep{frydman2023source}.

Parameters are optimized for each subject individually by maximizing trial-by-trial likelihood; Bayesian models are fit using the method of \cite{hahn2024unifying, Hahn2025Identifiability}, which represents non-parametric components on a 200-point grid over $[0,1]$ with smoothness regularization. We compare models by two metrics: summed $\Delta$AICc and summed held-out $\Delta$NLL (Appendix~\ref{app:details-model-fit-metrics}). We present the $\Delta$AICc results in the main paper. $\Delta$AICc penalizes parameter counts to ensure fair comparison across model variants. For freely fitted components, the smoothness regularization reduces effective degrees of freedom, but the effective parameter count is nontrivial to compute exactly. We therefore conservatively count all 200 grid points as free parameters (e.g., 403 for FreeE+FreeP), which provides an upper bound and overpenalizes variants with freely fitted components. Full method details and parameter counts are in Appendix~\ref{app:fitting-models}.

Model fit shows clear differences across variants (Figure~\ref{fig:jrf-bias-decomposition-resources-daicc}D). FreeP+FreeE achieves the best summed $\Delta$AICc, substantially outperforming both the parametric baseline BLO and other Bayesian variants with constrained encoding (UniformE, BoundedLOE, PriorMatchedE). The advantage of FreeE over UniformE in particular indicates that uniform encoding is insufficient: data require non-uniform encoding resources.

To understand these differences, we examine how each variant produces its bias (Figure~\ref{fig:jrf-bias-decomposition-resources-daicc}B, C). FreeP+UniformE produces a bias driven almost entirely by boundary regression, with zero likelihood repulsion (Panel C, bottom) — by construction, since uniform encoding has constant resources. The resulting fit is poor near the endpoints (Panel B, bottom), recovering the constraint from Section~\ref{sec:framework-decomposition} that uniform encoding cannot reproduce inverse-S weighting.
FreeP+FreeE shows a different decomposition (top of Panel B and C): the bias is dominated by S-shaped likelihood repulsion, with only minor boundary regression. Moreover, it recovers a U-shaped encoding with sharp peaks near 0 and 1 (Panel A, top). This pattern emerges without any functional-form assumption on $F$, aligning exactly with Prediction 1. FreeP+BoundedLOE, which constrains the encoding to a parametric U-shape, achieves comparable fit, suggesting that what the data require is a U-shaped resource allocation.

The previous analyses concern the encoding. The framework's second ingredient, the prior, leads to a separate signature in response variability. Across stimulus values, response variability shows a bimodal pattern with a dip near 0.5 (Appendix Figure~\ref{fig:jrf-bias-var}). Among the variants tested, this structure is captured only by FreeP, whose recovered prior develops a sharp peak around 0.5 because the local high prior density mathematically suppresses variance there. BLO produces an inverted-U variability profile by construction (Theorem~\ref{thm:blo-var} in Appendix~\ref{app:proofs-blo}) and cannot reproduce the bimodal dip. Its single-anchor prior provides no mechanism for fine-grained prior structure to shape variability.

The efficient-coding hypothesis predicts that the prior matches the encoding. Enforcing this constraint substantially reduces fit, and the recovered prior and encoding shapes differ in the freely fitted models (Appendix Figure~\ref{fig:jrf-prior-comparison} and ~\ref{fig:jrf-fi-comparison}), suggesting that the data require prior and encoding to be specified separately rather than tied together.

Together, these results show that probability distortion in JRF arises from optimal decoding of a U-shaped encoding combined with a flexible, encoding-independent prior. The recovered structure agrees with Prediction 1 even when no functional form is imposed. We next ask whether the same mechanism extends beyond perceptual judgment to economic decision under risk.

\begin{figure}
    \begin{center}
    \includegraphics[trim=0 0 0 0, clip, width=1.0\linewidth]{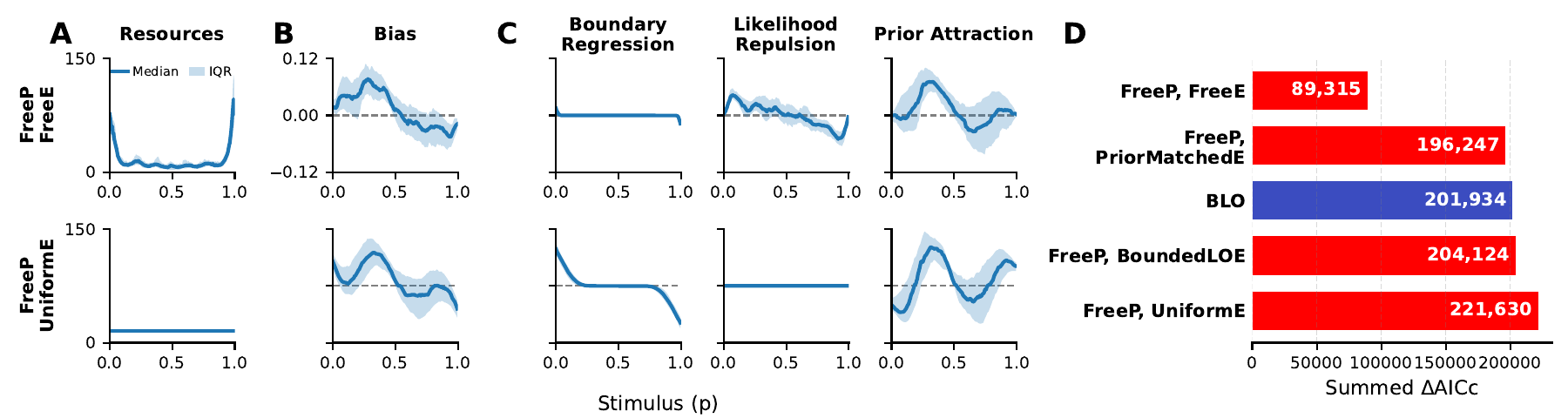}
    \end{center}
  \caption{Resources (A), bias (B) and bias decomposition (C) of FreeP FreeE and FreeP UniformE models for JRF data. IQR denotes interquartile range across subjects. In the best-fitting FreeP FreeE model, boundary regression plays only a small role, whereas Likelihood Repulsion consistently produces a bias away from 0 and 1. The method to compute empirical bias decomposition is described in Appendix~\ref{app:bias-decomposition-empirical}; results for additional model variants are in Appendix~\ref{app:cross-task-bias-decomposition}. (D) $\Delta$AICc summed across subjects (lower is better; non-parametric variants reported with upper-bound parameter count).}
  \label{fig:jrf-bias-decomposition-resources-daicc}
\end{figure}

\subsection{Testing the Generality on Decision-Making Under Risk}\label{sec:dmr}
Decision-making under risk has been the primary domain for probability weighting research. We test whether the U-shape recovery extends here using two paradigms: lottery pricing and risky choice. We analyze two datasets in which subjects evaluated two-outcome gambles $(x_1, p; x_2, 1-p)$ with real-valued payoffs, instantiated within Cumulative Prospect Theory (CPT) \citep{tversky1992advances}: 
\begin{equation}\label{eq:cpt}
\text{Utility} := v(x_1) \cdot w(p) + v(x_2) \cdot (1 - w(p)), 
\quad v(x) = |x|^\alpha,
\end{equation}
with a two-part value function for mixed gambles and separate weighting functions $w^+(p), w^-(p)$ for gains and losses (Appendix Eq.~\ref{eq:full-cpt}). Models differ only in the form of $w(p)$. Parametric accounts (BLO, Prelec, LILO) treat $w(p)$ as a deterministic transformation; in the Bayesian framework, $w(p)$ arises from stochastic encoding and optimal decoding, with FreeP+\{UniformE, BoundedLOE, FreeE\} instantiating the same family of variants tested in Section~\ref{sec:jrf}. Both tasks test \textbf{Prediction 1}: whether U-shaped encoding resources are recovered in economic decisions, as in perceptual judgment.

\paragraph{Pricing.}
We first analyze the pricing data from \cite{zhang2020}, in which subjects ($N=75$) chose between a two-outcome monetary gamble and a sure amount, with the sequence of choices adaptively converging to the Certainty Equivalent (CE) of the gamble. Because each CE reflects both the utility function and probability weighting, fitting CE responses directly does not isolate the contribution of $w(p)$. We therefore include an intermediate step: each reported CE is first inverted through the CPT utility function to obtain an implied probability weight, and the Bayesian encoding parameters are estimated from these trial-level weights; the utility exponent $\alpha$ and CE noise variance $\sigma^2_{\text{CE}}$ are then re-optimized to maximize the likelihood of the observed CE data (Appendix~\ref{app:details-bayesian-dmrpricing-models} for details). Final comparison across all models is based on the likelihood of the original CE responses.

The freely fitted encoding (FreeP+FreeE) recovers a U-shaped allocation of resources (Figure~\ref{fig:dmr-FI}A), confirming Prediction~1 in economic valuation. Model fit supports this conclusion: FreeP+FreeE substantially outperforms all alternatives (Figure~\ref{fig:dmr-FI}B), including the parametric baseline BLO and other Bayesian variants with constrained encoding. We further test the prediction in a choice task below, where responses depend on $w(p)$ directly rather than through the utility function.

\paragraph{Choice.}
To further test the generality of the framework, we analyzed a subset of the Choice Prediction Competition 2015 dataset \citep{erev2017anomalies}, restricted to two-outcome gambles with fully described probabilities and uncorrelated payoffs (Amb=0, Corr=0, LotNum=1; 187,150 trials across 153 subjects). Following CPT, probabilities of gains and losses are modeled with separate weighting functions $w^+(p)$ and $w^-(p)$. The Bayesian model implements these by fitting separate priors for gains and losses while sharing the same encoding, with $w^\pm(p)$ computed as the expectation of $\hat{p}$ under the encoding distribution. Choice probabilities are modeled with a logit rule on subjective utility differences with a temperature parameter $\tau$. As parametric benchmarks, we include the LILO function (Eq.~\ref{eq:lilo}) and the Prelec function \citep{prelec1998probability}; BLO is not applicable here, as its original specification does not extend to choice tasks.

The freely fitted encoding again recovers a U-shape (Figure~\ref{fig:dmr-FI}C). FreeP+FreeE outperforms all alternatives on fit (Figure~\ref{fig:dmr-FI}D), including the parametric baselines Prelec and LILO. Within the Bayesian framework, FreeE outperforms UniformE: constraining the encoding to be uniform produces a substantial drop in fit, providing direct evidence for Prediction~1 in the choice setting.

Across both pricing and choice paradigms, the recovered encoding is U-shaped, agreeing with Prediction~1. We next turn to the second ingredient of the framework, the prior, and test Prediction~2 in an adaptation experiment.

\begin{figure} 
\begin{center}
\includegraphics[trim=0 25 30 0, clip, width=1.0\linewidth]{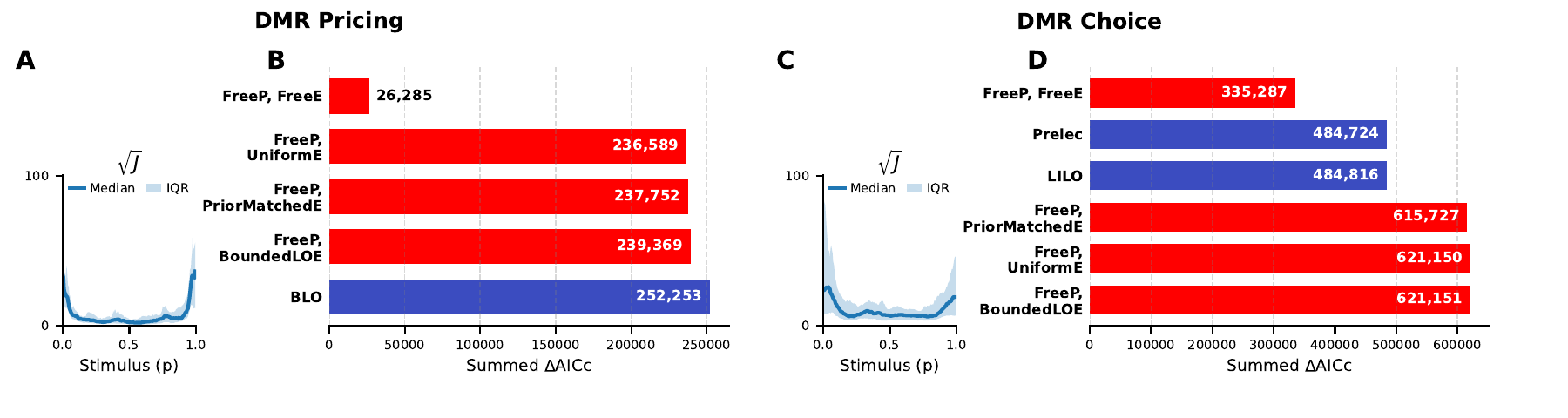}
\end{center}
\caption{
Results from the pricing (A, B) and choice tasks (C, D). (A, C) Bayesian models with non-uniform encoding recover U-shaped resources with peaks near 0 and 1, consistent with Prediction 1. (B, D) $\Delta$AICc summed across subjects (lower is better; non-parametric variants reported with upper-bound parameter count). In both paradigms, FreeP+FreeE achieves the best fit.
}
\label{fig:dmr-FI}\label{fig:choice}
\end{figure}

\subsection{Testing Prediction 2 via Adaptation to Stimulus Statistics}\label{sec:empirical-validation-adaptation}

We now examine the prior component of the framework. Section~\ref{sec:framework-decomposition} showed that the prior enters bias through a prior-attraction term that pulls estimates toward regions of high prior density. If the prior adapts to the stimulus distribution, then a bimodal distribution should induce a bimodal prior, producing attraction toward each of the two modes and a corresponding sign reversal of bias on either side of each mode (\textbf{Prediction 2}). The accounts reviewed in Section~\ref{sec:related-work} make distinct predictions here. The single-anchor account, BLO, can produce only one attraction point and cannot reproduce a two-mode pattern. Efficient-coding accounts with matched prior and encoding \citep{frydman2023source} predict the opposite sign of bias near the modes: when the encoding matches the prior, the posterior mean is biased \emph{away} from the prior mode rather than toward it \citep{wei2015bayesian}, yielding repulsion rather than attraction.

We tested these predictions in a new psychophysical experiment using the same dot-counting task as in Section~\ref{sec:jrf}, but with a bimodal stimulus distribution: dot proportions were sampled from a mixture of two Gaussians with peaks equidistant from $0.5$. These peaks were chosen separately for each subject.\footnote{Full experimental details are in Appendix~\ref{app:experiment-adaptation}. The experimental design was approved by the Ethical Review Board of the Faculty of Mathematics and Computer Science,
Saarland University.} Twenty-six subjects completed the task. Within the Bayesian framework, we extended the set of priors to include a mixture-of-two-Gaussians prior (BimodalP), and added a matched prior-encoding model (FreeP+PriorMatchedE) to instantiate the efficient-coding account.

\begin{figure}
\begin{center}
\includegraphics[width=1.0\linewidth]{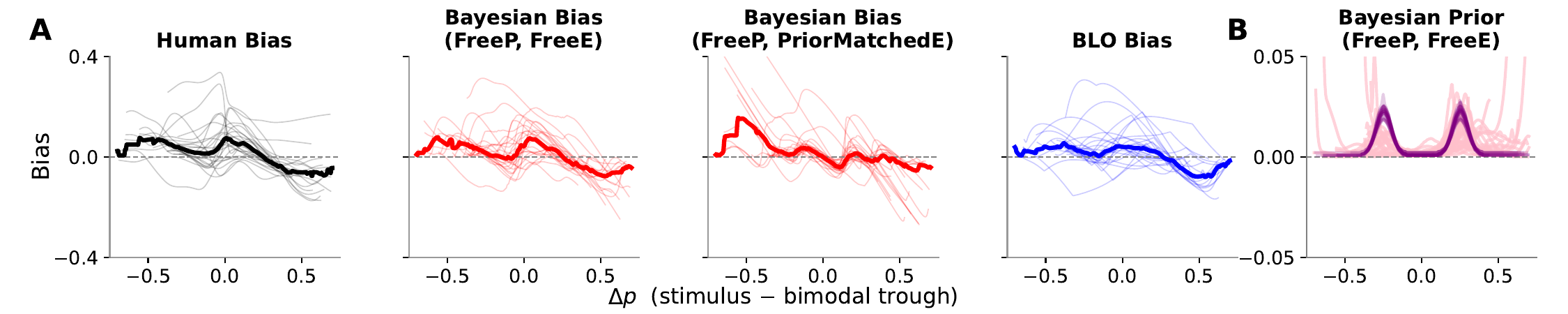}
\end{center}
\caption{
The perceived probability of human observers adapts to bimodal stimulus statistics in a dot-counting task. (A) Median bias across subjects, aligned so the trough between stimulus peaks is at $0$ ($\Delta p = $ stimulus $p$ minus trough). Human bias shows two cross-over points centered on the stimulus modes. FreeP+FreeE reproduces this pattern; FreeP+PriorMatchedE predicts the opposite sign of bias near the modes (repulsion rather than attraction); BLO produces a single S-shape and cannot reproduce the two-mode structure. (B) Priors recovered by FreeP+FreeE (pink: subjects; purple: ground truth stimulus distribution).
}
\label{fig:bimodal-bias-dcv}
\end{figure}

Human bias under the bimodal distribution shows two cross-over points centered on the stimulus modes, with attraction toward each mode (Figure~\ref{fig:bimodal-bias-dcv}A, leftmost). FreeP+FreeE reproduces this pattern (second panel). The prior recovered by FreeP+FreeE is bimodal with peaks aligned to the stimulus modes, closely matching the ground-truth distribution (Figure~\ref{fig:bimodal-bias-dcv}B). This direct recovery of the prior from data agrees with Prediction 2. Notably, the recovered encoding in this experiment remains U-shaped (see Appendix Figure~\ref{fig:jrf-FI-comparison-bimodal}), matching the shape recovered in Sections~\ref{sec:jrf}–\ref{sec:dmr} despite the different stimulus distribution. The two components dissociate empirically: the prior adapts while the encoding does not.

In contrast, BLO with single-anchor prior produces a single attraction point and cannot generate the two-mode structure (Figure~\ref{fig:bimodal-bias-dcv}A, rightmost). FreeP+PriorMatchedE matches the encoding to a bimodal prior, concentrating resources at the modes; through the likelihood-repulsion term, this biases estimates \emph{away} from the modes. The predicted bias near each mode therefore has the opposite sign from the human data (Figure~\ref{fig:bimodal-bias-dcv}A, third panel). Quantitative fit comparisons across all variants are reported in Appendix Figure~\ref{fig:jrf-daicc-full-both-analytical-categorical-bimodal}, with FreeP+FreeE outperforming other accounts.

These results confirm Prediction 2 and complete the empirical picture. Across this section the encoding component was recovered as U-shaped, and the prior component is recovered as flexible and tracking the stimulus distribution. The two components of the framework are supported by separate empirical tasks. Probability distortion thus emerges from optimal decoding of a U-shaped encoding combined with a flexible Bayesian prior.

\section{Discussion and Conclusion}
\label{sec:dicussion-and-conclusion}

We have developed and empirically validated a Bayesian encoding–decoding framework of probability distortion. In this framework, the U-shape encoding is a derived prediction rather than a built-in assumption. The bias decomposition derived in this framework unifies existing accounts on a common landscape.
A natural question raised by these results concerns the relationship between prior and encoding. Efficient-coding accounts predict that they should match, but our analyses find them to be empirically distinct in both perceptual and economic tasks. One possibility, consistent with related findings in perception \citep{fritsche2020bayesian}, is that priors and encodings adapt on different timescales. The bimodal adaptation experiment is consistent with this view, since the prior shifted within a single session while the recovered encoding remained U-shaped throughout.
More broadly, the work connects probability distortion to a growing line of work on imprecise but structured mental representations \citep[e.g.][]{woodford2012prospect, khaw2021cognitive, frydman2022efficient, barretto2023individual, zhang2020, frydman2023source}. By showing that a single encoding–decoding mechanism accounts for distortion across both perceptual judgment and economic decision under risk, the framework provides a candidate for what is shared across these domains, and a tool, agnostic recovery of encoding and prior from behavior, for probing it elsewhere.

\bibliography{refs}
\bibliographystyle{abbrvnat}

\newpage
\appendix
\section*{Appendix}
\addcontentsline{toc}{section}{Appendix}

\renewcommand{\contentsname}{Appendix Contents}
\startcontents[appendix]
\printcontents[appendix]{}{1}{}

\section{FAQs}

\begin{enumerate}
\item \textit{What is the contribution beyond \cite{hahn2024unifying}, \cite{Hahn2025Identifiability} or existing Bayesian accounts?}

\cite{hahn2024unifying}, \cite{Hahn2025Identifiability} provide general-purpose tools for Bayesian encoding–decoding models on unbounded or circular domains. We build on this infrastructure to address a specific open question in behavioral science: where does the probability weighting function come from?

Our contribution has three parts: First, probability lives on [0,1] with two hard boundaries. In Theorem~\ref{thm:bayesian-bias}, we derive a bias decomposition that accounts for both boundaries simultaneously, with boundary-modulated coefficients governing how regression, repulsion, and attraction trade off near each endpoint. Second, we establish that the inverse-S weighting pattern arises from a U-shaped allocation of encoding precision, recovered nonparametrically across perceptual estimation, lottery pricing, and risky choice, and stable under changes in stimulus distribution. Third, we show that regression-based, log-odds, efficient-coding, and BLO accounts emerge as special cases of a single framework, and the decomposition yields testable predictions that distinguish them.

\item \textit{FreeP+FreeE is more flexible than other models. Isn't it unsurprising that it outperforms other models?}

Flexibility is a fair concern, but the relevant question is not whether FreeP+FreeE fits better, but what shape it recovers and whether that shape generalizes.

First, the comparison uses two metrics that control for capacity. Held-out negative log-likelihood penalizes overfitting directly, and a conservative AICc that counts every grid point as a free parameter (e.g., 403 for JRF) penalizes nominal flexibility. FreeP+FreeE wins on both. Second, and more importantly, the recovered encoding is itself the finding. A flexible model can in principle recover any shape, yet across perceptual estimation, lottery pricing, and risky choice it consistently recovers a U-shape, and the same shape persists when the stimulus distribution becomes bimodal. The convergence on a specific encoding geometry, not the fit advantage, is what the framework predicts and what parametric alternatives such as uniform encoding cannot reproduce.

\item \textit{The adaptation experiment ($N$=26) is substantially smaller than the JRF ($N$=86) and pricing ($N$=75) datasets. Are the conclusions robust at this sample size?}

To verify that conclusions are not driven by individual subjects, we performed a subject-level bootstrap analysis (10,000 iterations, resampling with replacement) on held-out NLL. FreeP+FreeE outperformed the efficient-coding model (PriorMatchedE) in 98.1\% of bootstrap samples (95\% CI of NLL difference: $[233,7551]$) and outperformed BLO in 100\% of samples (95\% CI: $[7518,13399]$). See Appendix~\ref{app:adaptation-bootstrap} for full results.



\item \textit{Is the AICc penalty fair, given that regularization reduces the effective parameter count below 200?}

Strictly speaking, no. The AICc we report counts every grid point as a free parameter (200 for the encoding, 200 for the prior, plus noise parameters), which overstates the true degrees of freedom: smoothness regularization couples neighboring grid points, so the effective count is substantially lower. Computing the exact effective degrees of freedom under regularization is nontrivial, so we use the grid-size bound for transparency.

A tighter penalty would favor FreeP+FreeE further, and the ranking under the current bound already places it first.

\item \textit{AICc is an in-sample criterion. What about out-of-sample evaluation?}

We also report held-out NLL across all comparisons: see Appendix Figure~\ref{fig:jrf-dnll-full-both-analytical-categorical} (JRF), Appendix Figure~\ref{fig:dmr-dnll-full-both-analytical-categorical} (DMR pricing), Appendix Figure~\ref{fig:dmr-choice-dnll-full-both-analytical-categorical} (DMR choice), and Appendix Figure~\ref{fig:jrf-dnll-full-both-analytical-categorical-bimodal} (adaptation). The two metrics agree: FreeP+FreeE achieves the best held-out fit as well.




\item \textit{Existing log-odds and efficient-coding accounts already commit to a $U$-shaped FI; what is new here?}

These accounts posit a U-shape as a theoretical assumption: log-odds accounts derive it from a fixed $\log\frac{p}{1-p}$ transform, and efficient-coding accounts derive it from matching the encoding to a U-shaped prior. We recover it as an empirical output, nonparametrically and without any functional-form assumption on $F$.

Placing these accounts as special cases of a common framework lets us ask which commitment the data actually require, and the recovered U-shape differs from each in informative ways. Pure log-odds encoding implies $\sqrt{\mathcal{J}(p)} \to \infty$ at the boundaries, which is biologically implausible; BLO addresses this by truncating log-odds, but the truncation introduces discontinuities in the encoding. Efficient coding ties the U-shape to the prior, whereas the recovered encoding remains U-shaped even when the prior becomes bimodal under adaptation. By contrast, the recovered shape has finite endpoint precision and remains smooth throughout. The data favor a free encoding combined with a free prior over any of the constrained alternatives as well.

\item \textit{Why would the brain use a U-shaped encoding?}


In Appendix~\ref{sec:app:log-odds}, we describe a possible motivation of U-shaped encodings grounded in the well-documented encoding of scalar quantities following Weber's law.



\end{enumerate}

\section{Theoretical Derivations and Proofs}\label{app:whole-theory-section-app}

\subsection{Theory on Fisher Information (FI), Bias and Mean Square Error}\label{app:theory}

\begin{figure}[htbp]
\begin{center}
\includegraphics[width=0.8\linewidth]{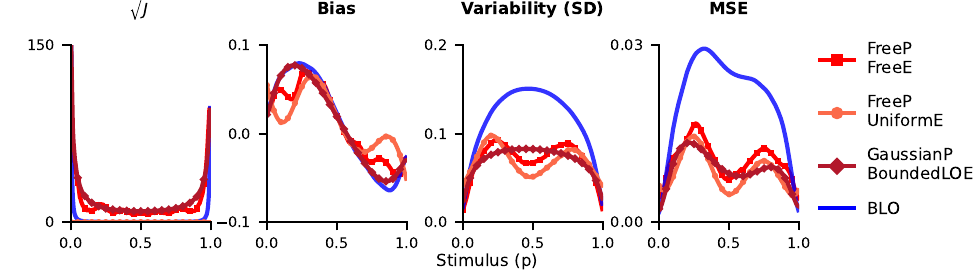}
\end{center}
\caption{FI, Bias, Variance, MSE of the three models: Bayesian, BLO, Bayesian with uniform encoding, for parameters fitted in Section~\ref{sec:jrf}. BLO has an FI peaking at 0 and 1, as does the Bayesian model with log-odds encoding. All model variants predict a positive bias for small $p$, and a negative bias for large $p$, in line with the classical probability weighting function (Figure~\ref{fig:intro}A). BLO and models with Gaussian priors predict a simple inverse U-shaped variability; other Bayesian priors can accommodate other variabilities. } 
\label{fig:jrf-FI-bias-var-mse}
\end{figure}

\subsubsection{Properties of Bayesian Model}\label{app:proofs-bayesian}
\begin{thm}[Repeated from Theorem~\ref{thm:bayesian-bias}]
There are functions $A_{1,\sigma}, A_{2,\sigma}, A_{3,\sigma} : [0,1] \rightarrow \mathbb{R}$ that are bounded, positive, and satisfy
    \begin{equation}\label{eq:A-bound-app}
         A_{\dots, \sigma}(p) \leq \phi\left(\frac{\sigma}{\min\{|F(p)-F(0)|, |F(p)-F(1)|\}}\right)
    \end{equation}
    for some nondecreasing function $\phi : \mathbb{R} \rightarrow \mathbb{R}$ with $\lim_{x\downarrow 0}\phi(x) = 0$,
    across $p \in (0,1); \sigma > 0$.
Then, at any $p \in (0,1)$, the Bayesian model has the bias, as $\sigma \rightarrow 0$: 
    \begin{align*}\label{eq:bayesian-bias-app}
\operatorname{Bias}(p) =      \underbrace{A_{1,\sigma}(p)  \cdot \frac{\operatorname{sign}(0.5-p)}{\sqrt{\mathcal{J}(p)}}}_{\text{Regression from Boundary}} +        \underbrace{(1-A_{2,\sigma}(p)) \cdot \frac{\operatorname{d}}{\operatorname{d}p} \left(\frac{1}{\mathcal{J}(p)}\right)}_{\text{Likelihood Repulsion}} \\
+ 
        \underbrace{(1-A_{3,\sigma}(p)) \frac{1}{\mathcal{J}(p)} \cdot \frac{\operatorname{d}}{\operatorname{d}p} \left(\log P_{{prior}}(p)\right)}_{\text{Prior Attraction}} + \mathcal{O}(\sigma^4)
    \end{align*}

    
\end{thm}

\begin{proof}
We first consider the case where $F(0) = -\infty$, $F(1) = \infty$, in which case (\ref{eq:A-bound-app}) forces $A_{\dots,\sigma}(p) \equiv 0$.
The above expression with $A_{\dots,\sigma}(p) \equiv 0$ simplifies to:
\begin{align*}
               \underbrace{\frac{\operatorname{d}}{\operatorname{d}p} \left(\frac{1}{\mathcal{J}(p)}\right)}_{\text{Likelihood Repulsion}} 
+ 
        \underbrace{\frac{1}{\mathcal{J}(p)} \cdot \frac{\operatorname{d}}{\operatorname{d}p} \left(\log P_{{prior}}(p)\right)}_{\text{Prior Attraction}} + \mathcal{O}(\sigma^4)
    \end{align*}
    Indeed, this is exactly the expression from Theorem 1 in \cite{hahn2024unifying}; the claim thus follows from that theorem.

We next consider the case  where $F(0), F(1)$ are finite; we obtain this on the basis of Theorem 3 in \cite{hahn2024unifying}. That theorem examines the bias in the setting where the stimulus space has \emph{one} boundary (e.g., $(-\infty, \theta_{Max}]$). 
For this case, it proves the decomposition (we consider the special case of $p=2$, as we assume the posterior mean estimator):
\begin{equation}
        \underbrace{-\frac{C_{1,p,D,F,\sigma}}{\sqrt{\mathcal{J}(p)}}}_{\text{Regression from Boundary}} +
        \underbrace{C_{3,p,D} \left(\frac{1}{\mathcal{J}(p)}\right)^{\prime}}_{\text{Likelihood Repulsion}} + 
        \underbrace{C_{2,p,D} \frac{1}{\mathcal{J}(p)}\left(\log P_{{prior}}(p)\right)^{\prime}}_{\text{Prior Attraction}} + \mathcal{O}(\sigma^4)
    \end{equation}
where $D(p) := \frac{F(\theta_{Max})-F(p)}{\sigma}$, and where $C_{\dots}$ are given as follows:
\begin{enumerate}
    \item $C_{1,p,D,F,\sigma}$ depends on the probability $p$, the quantity $D(p)$, the encoding function $F$, and the sensory noise magnitude $\sigma$
    \item $C_{2,p,D}$ depends on the probability $p$ and the quantity $D(p)$
    \item $C_{3,p,D}$ depends on the probability $p$ and the quantity $D(p)$
    \item $C_{1,p,D,F,\sigma}$, $C_{2,p,D}$, $C_{3,p,D}$ are positive everywhere
    \item $C_{1,p,D,F,\sigma} = \Theta(1)$ when $\sigma \rightarrow 0$
    \item $C_{1,\dots} \rightarrow_{D\rightarrow\infty} 0$, $C_{2/3,\dots}\rightarrow_{D\rightarrow \infty}1$.
\end{enumerate}
In fact, inspecting the proof of the theorem shows that $C_{3,p,D}, C_{2,p,D} \in (0,1)$ for all values of $p,D$, and that $C_{1,p,D(p),F,\sigma}$ is continuous in $p$ and $\sigma$ for $\sigma >0$.


Our aim is to use this result to establish the bias in the setting where the stimulus space is $[0,1]$, i.e., has \emph{two} boundaries.

Towards this end, we partition $(0,1)$ into $(0,1/2)$ and $(1/2,1)$.

First, $p \in (1/2, 1)$, the above decomposition is valid with $\theta_{Max} = 1$. The reason is because the contribution of the lower boundary at $0$ to the bias can be absorbed into the $\mathcal{O}(\sigma^4)$ remainder, because only exponentially small probability mass of the likelihood of the encoding $m$ falls into regions outside a small local environment of $p$ \citep[][SI Appendix, Section S.3.1.1]{hahn2024unifying}.

Second, for $p \in (0,1/2)$, by transforming the above result via $p \mapsto -p$, we obtain for the stimulus space $[\theta_{Min}, \infty)$ the bias:
\begin{equation}
        \underbrace{\frac{C_{1,p,D,F,\sigma}}{\sqrt{\mathcal{J}(p)}}}_{\text{Regression from Boundary}} +
        \underbrace{C_{3,p,D} \left(\frac{1}{\mathcal{J}(p)}\right)^{\prime}}_{\text{Likelihood Repulsion}} + 
        \underbrace{C_{2,p,D} \frac{1}{\mathcal{J}(p)}\left(\log P_{{prior}}(p)\right)^{\prime}}_{\text{Prior Attraction}} + \mathcal{O}(\sigma^4)
    \end{equation}
    where $D(p) := \frac{F(p)-F(\theta_{Min})}{\sigma}$, and where $C_{\dots}$ satisfy exactly the same conditions as given above.

    Now write $\hat{C}_{...}$ for the coefficients for $p \in (1/2,1)$ and $\tilde{C}_{...}$ for the coefficients for $p \in (0,1/2)$; the same with $\hat{D}(p)$ and $\tilde{D}(p)$.

    Now define
    \begin{align*}
        A_{1,\sigma}(p) &:= \hat{C}_{1,p,D(p),F,\sigma} \cdot 1_{p \geq 1/2} + \tilde{C}_{1,p,D(p),F,\sigma} \cdot 1_{p < 1/2} \\
                A_{2,\sigma}(p) :=& 1 - \left(\hat{C}_{3,p,D(p)} \cdot 1_{p \geq 1/2} + \tilde{C}_{3,p,D(p)} \cdot 1_{p < 1/2}\right) \\
        A_{3,\sigma}(p) :=& 1 - \left(\hat{C}_{2,p,D(p)} \cdot 1_{p \geq 1/2} + \tilde{C}_{2,p,D(p)} \cdot 1_{p < 1/2}\right) \\
    \end{align*}
First, inserting these definitions shows that the two decompositions do indeed produce the one claimed by the theorem.

We now need to show the remaining conclusions.
        By the assumptions above, $A_{1,\sigma}$ is positive and also continuous in both $\sigma$ and $p$. 
        Next, by the facts that $C_{1,p,D,F,\sigma} = \Theta(1)$ when $\sigma \rightarrow 0$ and  $C_{1,\dots} \rightarrow_{D\rightarrow\infty} 0$, we find that $A_{1,\sigma}(p)$ is bounded across $\sigma$ and $p$ whenever $\sigma$ is upper-bounded by some $\sigma_{Max}$. 
        As $C_{2,\dots}, C_{3,\dots} \in [0,1]$, we find that $A_{2,\sigma}$ and $A_{3,\sigma}$ are positive and bounded.
We have thus shown that $A_{1,\sigma}$, $A_{2,\sigma}$, $A_{3,\sigma}$ are positive and bounded.

    Finally, as mentioned above, $\hat{C}_{1,p,\hat{D}(p), F,\sigma} \rightarrow 0$ as $\hat{D}(p) \rightarrow \infty$; equivalently, 
\begin{equation}
    \hat{C}_{1,p,\hat{D}(p), F,\sigma} = \mathcal{O}\left(\phi_1\left(\frac{1}{\hat{D}(p)}\right)\right)
\end{equation}
for some nondecreasing function $\phi_1$ with $\lim_{s\downarrow 0}\phi_1(s) = 0$.
The same holds for $\tilde{C}$, $\tilde{D}$ with a function $\phi_2$.
Thus,
\begin{align*}
    A_{1,\sigma}(p) =& \mathcal{O}\left(\phi_1\left(\frac{1}{\hat{D}(p)}\right)\right) + \mathcal{O}\left(\phi_2\left(\frac{1}{\tilde{D}(p)}\right)\right) \\
    = & \mathcal{O}\left(\max \left\{\phi_1(\frac{1}{\hat{D}(p)}), \phi_2(\frac{1}{\tilde{D}(p)})\right\}\right)
\end{align*}
which is equivalent to the claim that \begin{equation}
         A_{1, \sigma}(p) \leq \phi\left(\frac{\sigma}{\min\{|F(p)-F(0)|, |F(p)-F(1)|\}}\right)
    \end{equation}
    across $p \in (0,1); \sigma > 0$, for some nondecreasing  function $\phi$ with $\lim_{s\downarrow 0}\phi(s) = 0$.
An analogous argument applies to $A_{2,\sigma}$ and $A_{3,\sigma}$.

\color{black}
\end{proof}

    \begin{corollary}\label{cor:log-odds-encoding}
For unbounded log-odds encoding:
    \begin{equation}\label{eq:bayesian-bias-log-odds}
        \mathbb{E}[\hat{p}|p]-p=
   \underbrace{\sigma^2 p^2 (1-p)^2 \left(\log P_{{prior}}(p)\right)^{\prime}}_{\text{Attraction}} +
 \underbrace{ 2 \sigma^2 p (1-p)(1-2p)}_{\text{Likelihood Repulsion}} + \mathcal{O}(\sigma^4)
    \end{equation}
\end{corollary}

\begin{proof}
    When $\beta=0$, (\ref{eq:f-bounded-log-odds}) simplifies to:
    \begin{equation}
        \sqrt{\mathcal{J}(p)} = \frac{1}{\sigma \cdot p\cdot (1-p)}
    \end{equation}
    Furthermore, the coefficients $A_{\dots}$ vanish for this encoding, because the encoding map $F(p) = \int_{1/2}^p \sqrt{\mathcal{J}(q)} dq$ is the log-odds transformation, satisfying $F(0) = -\infty$, $F(1) = \infty$.
    Now
    \begin{equation}
        \frac{1}{\mathcal{J}(p)} = \sigma^2 \cdot p^2\cdot (1-p)^2
    \end{equation}
    with derivative
    \begin{equation}
        \left(\frac{1}{\mathcal{J}(p)}\right)' = \sigma^2 \cdot 2 \cdot p \cdot (1-p) \cdot (1-2p)
    \end{equation}
    Plugging these into Theorem~\ref{thm:bayesian-bias} yields the result.
    
\end{proof}

\begin{thm}\label{thm:bayesian-var}
At each $p \in (0,1)$, the Bayesian model has the response variability:
    \begin{equation}
        \frac{1}{\mathcal{J}(p)}  + \frac{2\sigma^2 }{F'^2(p)} \frac{d}{dp} \underbrace{\left[\mathbb{E}[\hat{p} |p] - \mathbb{E}[F^{-1}(m)|p]\right]}_{\text{Bias introduced by decoding}} + \frac{\sigma^4 F''^2(p)}{2F'^6(p)} + O(\sigma^6)
    \end{equation}
    as $\sigma \rightarrow 0$.
\end{thm}

\begin{proof}
The term $\left[\mathbb{E}[\hat{p} |p] - \mathbb{E}[F^{-1}(m)|p]\right]$ corresponds to the quantity referred to as Decoding Bias in \cite{hahn2024unifying, Hahn2025Identifiability}.
The equation is then obtained from Lemma S24 in \cite{Hahn2025Identifiability}. Note that, for the quantity $C_{dec,M}$ used in that lemma, we have $\left[\mathbb{E}[\hat{p} |p] - \mathbb{E}[m|p]\right] = \sigma^2 C_{dec,M} + \mathcal{O}(\sigma^4)$, completing the proof.
\end{proof}

\begin{thm}[Optimality of Decoding]\label{thm:Bayesian-optimality}
The MSE of the decoded estimate in the Bayesian Model is
\begin{equation}
    \frac{1}{\mathcal{J}(p)} + O(\sigma^4)
\end{equation}
\end{thm}
\begin{proof}
    Immediate, from $MSE = Var + Bias^2$, and noting that $Bias = \mathcal{O}(\sigma^2)$.
\end{proof}

\subsubsection{Properties of BLO Model}\label{app:proofs-blo}
\begin{thm}[Repeated from Theorem~\ref{thm:blo-bayesian}]
    The BLO model in the limit of untruncated log-odds ($\Delta_+ \rightarrow \infty, \Delta_- \rightarrow -\infty$), and a Bayesian model with unbounded log-odds encoding ($\beta = 0$) and a specific unimodal prior $P_{prior}(p)$ (depending on $\Lambda_0$) have the same bias up to difference $\mathcal{O}(\kappa^2 +  \sigma^4)$.
\end{thm}

\begin{proof}
We derive this as a corollary of Theorem~\ref{thm:blo-bias}.
In the limit $\Delta_+ \rightarrow \infty, \Delta_- \rightarrow -\infty$, $\Phi(p)$ equals $p$.
We now obtain the result by matching (\ref{eq:bayesian-bias-log-odds}) and (\ref{eq:blo-bias}). Specifically, (\ref{eq:blo-bias}) then assumes the form:
\begin{align}
       \underbrace{\kappa \cdot p (1-p) V_p [\Lambda_0-\Lambda(p)]}_{\equiv Attraction} + \underbrace{\frac{\sigma^2}{2} p\,(1 - p)\,(1 - 2 p)}_{\equiv Repulsion}  + O(\kappa^2) + O(\sigma^4)
\end{align}
We match this with the Bayesian bias:
    \begin{equation}
   \underbrace{\sigma^2 p^2 (1-p)^2 \left(\log P_{{prior}}(p)\right)^{\prime}}_{\text{Attraction}} +
 \underbrace{ 2 \sigma^2 p (1-p)(1-2p)}_{\text{Likelihood Repulsion}} + \mathcal{O}(\sigma^4)
    \end{equation}
    Under the identification $\kappa = \sigma_{Bayesian}^2$, $\sigma^2_{BLO} = 4\sigma^2_{Bayesian}$, we find
    \begin{equation}
        \left(\log P_{{prior}}(p)\right)^{\prime} = \frac{
       V_p [\Lambda_0-\Lambda(p)]}{p(1-p)} \propto 
       \Lambda_0-\Lambda(p)
    \end{equation}
    Since $\Lambda(p)$ is monotonically increasing in $p$, this quantity is monotonically decreasing.
    Hence, $P_{{prior}}(p)$ has a single peak and describes a unimodal prior.
\end{proof}

\begin{thm}\label{thm:blo-FI}
    The BLO model has resource allocation $\sqrt{\mathcal{J}(p)}$ equal to: 
    \begin{equation}
         \frac{1}{\sigma} \cdot\left[\left(- \kappa \frac{ 1 - 2p}{(1+\kappa p ( 1-p))^2} \right) \cdot (\Lambda(p) - \Lambda_0) +  \frac{\omega_{p} \cdot \Psi}{(\Delta_+-\Delta_-)/2} \begin{cases} 0 & if~~\lambda(p) \not\in [\Delta_-, \Delta_+] \\  \lambda\prime(p) & else \end{cases} \right] \\
    \end{equation}
\end{thm}
\begin{proof}[Proof of Theorem~\ref{thm:blo-FI}]
    Given the standard expression for the Fisher information of a Gaussian with parameter-dependent mean and constant variance,     $\sqrt{\mathcal{J}(p)}$ equals
    \begin{align*}
        & \frac{1}{\sigma}\cdot\frac{d\hat{\Lambda}_{\omega}(p)}{dp}  \\
        = & \frac{1}{\sigma} \cdot \frac{d }{dp} \left[\omega_p \cdot\Lambda(p) + (1 - \omega_p)\cdot \Lambda_0\right]\\
        = & \frac{1}{\sigma} \cdot \frac{d }{dp} \left[\omega_p \cdot\Lambda(p) - \omega_p\cdot \Lambda_0\right] \\
        = & \frac{1}{\sigma} \cdot\left[ \left(\frac{d}{dp} \omega_p\right) \cdot\Lambda(p) + \omega_p \left(\cdot \frac{d}{dp} \Lambda(p)\right)- \Lambda_0 \cdot \left(\frac{d}{dp} \omega_p\right) \right]\\
        = & \frac{1}{\sigma} \cdot\left[ \left(\frac{d}{dp} \omega_p\right) \cdot (\Lambda(p) - \Lambda_0) + \omega_p \left(\cdot \frac{d}{dp} \Lambda(p)\right)\right]\\
        = & \frac{1}{\sigma} \cdot \left[ \left(\frac{d}{dp} \frac{1}{1+\kappa p (1-p)} \right)\cdot (\Lambda(p) - \Lambda_0) +  \omega_p \cdot  \left(\frac{\Psi}{(\Delta_+-\Delta_-)/2} \cdot\frac{d}{dp} \Gamma(\lambda(p)) \right) \right] \\
        = & \frac{1}{\sigma} \cdot \left[ \left(- \kappa \frac{ 1 - 2p}{(1+\kappa p ( 1-p))^2} \right) \cdot (\Lambda(p) - \Lambda_0) +    \frac{\omega_p \cdot\Psi}{(\Delta_+-\Delta_-)/2} \begin{cases} 0 & if ~\lambda(p) \not\in [\Delta_-, \Delta_+] \\  \lambda\prime(p) & else \end{cases} \right] 
    \end{align*}
\end{proof}

\begin{thm}\label{thm:blo-bias}
    At any $p \in (0,1)$, the BLO model has the bias:
    \begin{align}\label{eq:blo-bias}
      & \Phi(p)  - p  + \underbrace{\kappa \cdot \Phi(p) (1-\Phi(p)) V_p [\Lambda_0-\Lambda(p)]}_{\equiv \text{Attraction}} + \underbrace{\frac{\sigma^2}{2} \Phi(p)\,(1 - \Phi(p))\,(1 - 2 \Phi(p))}_{\equiv \text{Repulsion}} + O(\kappa^2) + O(\sigma^4)
\end{align}
where $\Phi : (0,1) \rightarrow (0,1) : \Phi(p) := \lambda^{-1}(\hat\Lambda(p))$. 
\end{thm}

\begin{proof}[Proof of Theorem~\ref{thm:blo-bias}]
Consider the estimate as a function of $p$:
    \begin{equation}
    \hat{\pi}(p) = \lambda^{-1}(\hat{\Lambda}_{\omega}(p) + \epsilon_{\lambda})
\end{equation}
The bias is its expectation over $\epsilon_\lambda$ minus the true value:
\begin{equation}
    \mathbb{E}\left[\lambda^{-1}(\hat{\Lambda}_{\omega}(p) + \epsilon_{\lambda})\right] - p
\end{equation}
where $\epsilon_\lambda \sim N(0,\sigma^2)$.
To understand it, we perform a Taylor expansion around $\sigma=0$, $\kappa=0$.
That is, we start by computing
\begin{equation}
    \frac{\partial^2}{(\partial \epsilon_\lambda)^2} \left[\lambda^{-1}(\hat{\Lambda}_{\omega}(p) + \epsilon_{\lambda})\right] = \lambda^{-1}(\Lambda(p))\,(1 - \lambda^{-1}(\Lambda(p)))\,(1 - 2 \lambda^{-1}(\Lambda(p)))
\end{equation}
at $\epsilon_\lambda=0$, $\kappa=0$, and
\begin{equation}
    \frac{\partial}{\partial \kappa} \left[\lambda^{-1}(\hat{\Lambda}_{\omega}(p) + \epsilon_{\lambda})\right] = \lambda^{-1}(\Lambda(p))\,(1-\lambda^{-1}(\Lambda(p)))\,V_p\,[\Lambda_0-\Lambda(p)]\ \\
\end{equation}
at $\kappa=0$, $\epsilon_\lambda=0$. 

Then
\begin{align*}
    \mathbb{E}\left[\lambda^{-1}(\hat{\Lambda}_{\omega}(p) + \epsilon_{\lambda})\right] - p = & \kappa \cdot \frac{\partial}{\partial \kappa} \left[\lambda^{-1}(\hat{\Lambda}_{\omega}(p) + \epsilon_{\lambda})\right] \\
    & + \frac{\sigma^2}{2}\frac{\partial^2}{(\partial \epsilon_\lambda)^2} \left[\lambda^{-1}(\hat{\Lambda}_{\omega}(p) + \epsilon_{\lambda})\right] \\
    & + \lambda^{-1}(\Lambda(p)) \\
    & - p \\
    & + O(\kappa^2) + O(\sigma^4) 
\end{align*}
Filling in the above expressions, we get
\begin{align*}
   \mathbb{E}\left[\lambda^{-1}(\Lambda(p) + \epsilon_{\lambda})\right] - p = & \kappa \cdot \lambda^{-1}(\Lambda(p)) (1-\lambda^{-1}(\Lambda(p))) V_p [\Lambda_0-\Lambda(p)]\ \\
    & + \frac{\sigma^2}{2} \lambda^{-1}(\Lambda(p))\,(1 - \lambda^{-1}(\Lambda(p)))\,(1 - 2 \lambda^{-1}(\Lambda(p))) \\
    & + \lambda^{-1}(\Lambda(p)) \\
    & - p \\
    & + O(\kappa^2) + O(\sigma^4)
\end{align*}
\end{proof}

\begin{thm}\label{thm:blo-var}
    The BLO model has the response variability: 
    \begin{align*}  \lambda^{-1}(\hat{\Lambda}_{\omega}(p))^2 (1-\lambda^{-1}(\hat{\Lambda}_{\omega}(p))^{2} \,\sigma^{2} 
&+ \tfrac{1}{2}  \lambda^{-1}(\hat{\Lambda}_{\omega}(p))^2 (1- \lambda^{-1}(\hat{\Lambda}_{\omega}(p)))^2 (1-2 \lambda^{-1}(\hat{\Lambda}_{\omega}(p)))^2 \,\sigma^{4} \\
&+ O(\sigma^{6}).
\end{align*}
\end{thm}

\begin{proof}
Consider the estimate as a function of $p$:
    \begin{equation}
    \hat{\pi}(p) = \lambda^{-1}(\hat{\Lambda}_{\omega}(p) + \epsilon_{\lambda})
\end{equation}
Conditioning on $p$, the variance over $\epsilon_\lambda$ is using a Taylor expansion:
\begin{align*}
    \operatorname{Var}\left[\lambda^{-1}(\hat{\Lambda}_{\omega}(p) + \epsilon_{\lambda})\right] =& 
\big(f'(\mu)\big)^{2} \,\sigma^{2}
+ \tfrac{1}{2}\big(f''(\mu)\big)^{2} \,\sigma^{4}
+ O(\sigma^{6}).
\end{align*}
where
\begin{equation}
    f(x) := \lambda^{-1}(\hat{\Lambda}_{\omega}(p) + x)
\end{equation}
Plugging this in, we obtain 
\begin{align*}
    \operatorname{Var}\left[\lambda^{-1}(\hat{\Lambda}_{\omega}(p) + \epsilon_{\lambda})\right] =& 
 \lambda^{-1}(\hat{\Lambda}_{\omega}(p))^2 (1-\lambda^{-1}(\hat{\Lambda}_{\omega}(p))^{2} \,\sigma^{2} \\
&+ \tfrac{1}{2}  \lambda^{-1}(\hat{\Lambda}_{\omega}(p))^2 (1- \lambda^{-1}(\hat{\Lambda}_{\omega}(p)))^2 (1-2 \lambda^{-1}(\hat{\Lambda}_{\omega}(p)))^2 \,\sigma^{4} \\
&+ O(\sigma^{6}).
\end{align*}
\end{proof}




\subsection{Bayesian Models with Log-Odds Encoding}\label{sec:app:log-odds} 
The Bayesian model can accommodate a general mapping $F$ and we will be able to infer it from data.
Here, we justify the popular choice of a log-odds mapping (as assumed by  \cite{khaw2021cognitive,Khaw2022StateDependent,zhang2020}), and how this leads to a reinterpretation of BLO as an approximation of the Bayesian model. 
Suppose that observers encode positive counts ($\kappa_+$) and negative counts ($\kappa_-$):
\begin{equation}
    m = \left(\begin{matrix} F_{num}(\kappa_+) + \epsilon_1 \\ F_{num}(\kappa_-) + \epsilon_2 \end{matrix} \right)\ \ \ \ \ \  \ \ \ \ \ \epsilon_1, \epsilon_2 \sim \mathcal{N}(0,\sigma^2)
\end{equation}
where  $p = \frac{\kappa_+}{\kappa_++\kappa_-}$.
Based on research on magnitude perception, we take $F_{num}$ to be consistent with Weber's law by assuming the form \citep[e.g.][]{petzschner2011iterative,dehaene2003neural}:
\begin{align*}
    F_{num}(\kappa) =& \log(\kappa+\alpha) & \text{hence\ \ \ \ \ \ \ \ \ \ } m =& \left(\begin{matrix} \log(pN+\alpha) + \epsilon_1 \\ \log((1-p)N+\alpha)+ \epsilon_2 \end{matrix} \right)
\end{align*}
where $\alpha > 0$ prevents an infinite FI at zero \citep{petzschner2011iterative}, and $N = \kappa_++\kappa_-$.
What is an optimal 1D encoding $m_{1D} \in \mathbb{R}$ of the rate $p = \frac{\kappa_+}{\kappa_++\kappa_-}$? We focus on linear encodings with coefficients $w_1, w_2$:
\begin{equation}
 m_{1D} :=   w_1  (\log(pN+\alpha) + \epsilon_1) + w_2 (\log((1-p)N+\alpha)+ \epsilon_2)
\end{equation}
This uniquely encodes $p$ only if $w_1w_2 < 0$ because of monotonicity.
Symmetry thus suggests $w_1=-w_2$, equivalent to a log-odds encoding smoothed at the boundaries ($\beta>0$):
\begin{equation}\label{eq:f-bounded-log-odds}
    F(p) := \log\frac{p+\beta}{(1-p)+\beta} \ \ \ \ \ \ \sqrt{\mathcal{J}(p)} = \frac{1}{\sigma} \cdot \left( \frac{1}{p+\beta} + \frac{1}{(1-p)+\beta} \right)
\end{equation}

where, for fixed $N$, we absorbed $N$ into $\beta = \alpha/N$. 
Note that this $\mathcal{J}(p)$ is U-shaped as in Figure~\ref{fig:conceptual}A.
This form for $F$ allows us to interpret BLO as an approximation to the Bayesian model.
First, in the limit where $\beta \rightarrow 0$ (i.e., unbounded log-odds encoding), the Bayesian bias comes out to (see Corollary~\ref{eq:bayesian-bias-log-odds} in Appendix~\ref{app:proofs-bayesian}):
    \begin{equation}\label{eq:bayesian-bias-log-odds}
        \mathbb{E}[\hat{p}|p]-p=
   \underbrace{\sigma^2 p^2 (1-p)^2 \left(\log P_{{prior}}(p)\right)^{\prime}}_{\text{Attraction}} +
 \underbrace{2 \sigma^2 p (1-p)(1-2p)}_{\text{Repulsion}} + \mathcal{O}(\sigma^4)
    \end{equation}
    The first term depends on the prior distribution.
    As expected, the second term describes an S-shaped bias (Figure~\ref{fig:conceptual}B). 
Indeed, when noise parameters are small, the BLO model matches the bias of this Bayesian model with a specific unimodal prior (Proof in Appendix~\ref{app:proofs-blo}):
\begin{thm}\label{thm:blo-bayesian}
    The BLO model in the limit of untruncated log-odds ($\Delta_+ \rightarrow \infty, \Delta_- \rightarrow -\infty$), and a Bayesian model with unbounded log-odds encoding ($\beta = 0$) and a specific unimodal prior $P_{prior}(p)$ (depending on $\Lambda_0$) have the same bias up to difference $\mathcal{O}(\kappa^2 +  \sigma^4)$. 
\end{thm}
This result suggests that one can reinterpret BLO as an approximation of a Bayesian model with log-odds encoding and a unimodal prior, albeit with divergences in the response distribution due to its sub-optimal decoding process.
Our empirical results detailed below support this view, showing that the Bayesian model consistently provides better fit to behavior compared to the BLO model.

\subsection{FI for Efficient Code in \cite{frydman2023source}}\label{sec:frydman-jin-fi}

\citet{frydman2023source} take the encoding to be given by the code derived in \cite{heng2020efficient}, which is defined via:
\begin{equation}
    \theta(p) = \sin^2\left(\frac{\pi}{2} F(p)\right)
\end{equation}
where $F(p)$ is the cumulative distribution function of the prior, whose density we denote $f(p)$.
We now note, for $a(p) := \frac{\pi}{2} F(p)$:
\begin{equation}
    \theta'(p) = 2 \sin a \cos a \cdot a'(p) = \sin(2a) \cdot a'(p)
\end{equation}
and hence $a'(p) = \frac{\pi}{2} f(p)$.
Now for a code given by $m \sim \operatorname{Binomial}(n, \theta(p))$, the Fisher Information is
\begin{align*}
    \mathcal{J}(p) =&  \frac{n [\theta'(p)]^2}{\theta(p) (1-\theta(p))} \\
    =&  \frac{n \sin^2(2a) [a'(p)]^2}{\sin^2(a) (1-\sin^2(a))} \\
    =&  \frac{n \sin^2(2a) [a'(p)]^2}{\sin^2(a) \cos^2(a)} \\
    =& \frac{n \sin^2(2a) [a'(p)]^2}{\frac{1}{4} \sin^2(2a)} \\
    =& 4n [a'(p)]^2 \\
    =&  4n(\frac{\pi}{2} f(p))^2 \\
    \propto & f(p)^2
\end{align*}
We thus obtain $\sqrt{\mathcal{J}(p)} \propto p_{prior}(p)$.


\section{Experimental Methods and Supplementary Results}\label{app:whole-experimental-details-app}
\subsection{Theoretical Mapping of Model Variants}

\begin{table}[H]
\caption{Correspondence between existing theoretical accounts and the specific model variants used in our empirical comparison.}
\label{sample-table}
\begin{center}
\begin{tabular}{ll}
\multicolumn{1}{c}{\bf Concept}  &\multicolumn{1}{c}{\bf Corresponding Model in Our Comparison}
\\ \hline \\
Fixed probability distortion function&LILO, Prelec \\
Log-odds based accounts&BLO\\
&Bayesian model with bounded log-odds encoding\\
Regression based accounts   &Bayesian model with uniform encoding \\
Efficient-coding &Bayesian model with prior-matched encoding\\
General Bayesian account &Bayesian model with freely fitted prior and encoding\\
\label{tab:model-mapping}
\end{tabular}
\end{center}
\end{table}

\subsection{Details on Task and Datasets}\label{app:details-task-datasets}

\subsubsection{Judgment of Relative Frequency (JRF) Task in  Section~\ref{sec:jrf}}

For this task, we analyzed the publicly available dataset from \cite{zhang2020} and used \cite{ZM2012} . On each trial, subjects were briefly shown an array of black and white dots and reported their estimate of the relative frequency of one color by clicking on a horizontal scale.

We used two datasets that differ in the granularity of the stimulus proportions:

\begin{itemize}
     \item \textbf{\cite{zhang2020} Dataset (JD, including JDA and JDB):}
     \begin{itemize}
         \item \textbf{subjects:} A total of 75 subjects were divided into two groups: 51 subjects who completed 660 trials (JDA) and 24 who completed 330 trials (JDB).
         \item \textbf{Stimuli:} Coarse-grained. The objective relative frequency was drawn from the 11 discrete probability levels: {0.01, 0.05, 0.1, 0.25, 0.4, 0.5, 0.6, 0.75, 0.9, 0.95, 0.99}.
     \end{itemize}
     \item \textbf{\cite{ZM2012} Dataset (ZM12):}
     \begin{itemize}
         \item \textbf{subjects: } 11 subjects who completed 800 trials.
         \item \textbf{Stimuli:} Fine-grained. Stimulus proportions consisted of 99 levels, ranging from 0.01 to 0.99.
     \end{itemize}
\end{itemize}
   
In both dataset, the total number of dots on any given trial was one of five values: 200, 300, 400, 500, or 600. For all datasets, each trial included the true stimulus proportion, the subject's estimate, total number of dots and other information, such as reaction time, that we didn't use for model fitting.We did not apply any data preprocessing.

\subsubsection{Pricing Task in Decision-making Under Risk (DMR) in Section~\ref{sec:dmr}}\label{app:details-dmr-pricing}
For this task, we analyzed the publicly available dataset from \cite{zhang2020}. This task used the same procedure and design as \cite{gonzalez1999shape}. On each trial, subjects were presented with a two-outcome monetary gamble (e.g., a 50\% chance to win \$100 or \$0 otherwise) and a table of sure amounts. subjects made a series of choices between the gamble and the sure amounts, a process that used a sequential bisection method to narrow the range and determine their Certainty Equivalent (CE) for the gamble. 
\begin{itemize}
    \item \textbf{subjects:} The dataset comprises responses from the same 75 subjects who participated in the JD dataset in JRF task described above. 51 subjects from JDA performed 330 trials each and 24 subjects from JDB performed 165 trials each.
    \item \textbf{Stimuli:} The experiment used 15 distinct pairs of \textbf{non-negative} outcomes (e.g., \$25 vs. \$0; \$100 vs. \$50; \$800 vs. \$0). These were crossed with the same 11 probability levels from the JRF task for the higher outcome, resulting in 165 unique gambles.
    \item \textbf{Data:} Each recorded data point included the gamble's two outcomes, the probability of the higher outcome, the subject's final determined CE and other information that we didn't use for model fitting.
\end{itemize}

We did not apply any data preprocessing.

\subsubsection{Choice Task in Decision-making Under Risk (DMR) in Section~\ref{sec:dmr}}\label{app:details-dmr-choice}
For this task, we analyzed the publicly available dataset from the Choice Prediction Competition 2015 (CPC15). This study was designed to test and compare models on their ability to predict choices between gambles that elicit classic decision-making anomalies. On each trial, subjects were presented with two or more distinct monetary gambles and made a one-shot choice indicating their preference.

The full CPC15 dataset contains a wide range of gamble types, including ambiguous gambles (Amb=1), gambles with correlated outcomes (Corr $\neq$ 0), and multi-outcome lotteries (LotNum $>$ 1). For the purpose of our analysis, which requires two-outcome gambles with fully specified probabilities and independent payoffs, we applied the following filters:

\begin{itemize}
\item Amb = 0: excluded ambiguous gambles, i.e., those with probabilities not explicitly described to subjects.
\item Corr = 0: excluded gambles with correlated payoffs across options.
\item LotNum = 1: restricted to gambles with exactly two outcomes in each option.
\end{itemize}

This gives us a subset of the CPC15 dataset with following information:

\begin{itemize}
    \item \textbf{subjects: }The subset contains responses from 153 subjects from the competition's estimation and test sets.
    \item \textbf{Stimuli: }The gambles covered 14 different behavioral phenomena, including the Allais and Ellsberg paradoxes. Unlike the pricing task, these gambles included both \textbf{gains and losses}, and the probabilities were drawn from a larger set of levels. 
    \item \textbf{Data: } Each data point recorded the two gambles presented, the subject's choice, and information we didn't use for model fitting.
\end{itemize}

\subsection{Details for Adaptation Experiment in Section~\ref{sec:empirical-validation-adaptation}}\label{app:experiment-adaptation}

This experiment used the JRF dot-counting paradigm, but critically, the distribution of stimulus proportions was manipulated to be 
bimodal to test the model's ability to adapt its prior. Data was collected on an online platform following the procedure in \cite{zhang2020}.

The bimodal dataset comprises responses from \textbf{26 subjects} across several designs to assess adaptation:
\begin{itemize}
    \item 5 subjects performed 740 trials following from a bimodal distribution.
    \item 13 subjects performed an initial 740 trials following a bimodal distribution, followed by uniform trials to assess after-effects. In the 13 subjects, 7 subjects performed 840 trials(740 bimodal trials followed by 100 uniform trials), and 6 subjects performed 1136 trials(740 bimodal trials followed by 396 uniform trials).
    \item 8 subjects performed 942 bimodal trials interspersed with 198 uniform trials.
\end{itemize}

For all subjects in our collected datasets, in addition to the true number of black and white dots, the color designated for estimation, the subject's estimated proportion, and the reaction time, we also varies and logged the display time for each trial. 

We removed estimated proportions of 0 and 100.

\subsubsection{Recruitment and Task instructions}
\label{app:experiment-adaptation-details}

\textbf{Recruitment and compensation.} Participants were recruited online via Prolific (\url{prolific.co}), with eligibility limited to adults (18+) with normal or corrected-to-normal vision. Recruitment complied with Prolific's policies. The experiment was expected to take less than 120 minutes; payment was prorated by expected completion time to exceed the German minimum wage. Pay and time estimates were adjusted during recruitment if observed completion time deviated from expectation. Subject identities were stored as anonymized numerical identifiers (subject 1, subject 2, ...).

\textbf{Task instructions.} Participants were shown the following clarification at the start of the experiment:

\begin{quote}
``Thank you for participating in this study. In this experiment, you will be estimating the relative frequency of visual elements (black and white dots). You must be at least 18 years of age and legally able to consent to your participation in order to take part in this study.

You will be presented with a display of black and white dots. The display will disappear after a short duration. After viewing the display, you will be asked to estimate the proportion of the specified color of dots by selecting a value on a horizontal bar ranging from 0\% to 100\%.

The entire experiment is expected to take less than 120 minutes. If you are ready to begin, please turn off any distractions such as music or television, and click below to begin. Thank you!''
\end{quote}

\textbf{Informed consent.} Participants provided informed consent before beginning the study:

\begin{quote}
``Your participation in this experiment is voluntary. You have the right to withdraw your consent or discontinue participation at any time without penalty or loss of benefits to which you are otherwise entitled. All data resulting from this study will be stored in anonymized form. The results of the experiment may, in anonymized form, be used and published for scientific purposes. By clicking on `Continue', you declare that you agree to voluntarily participate in this experiment.''
\end{quote}

\textbf{Optional post-task questionnaire.} At the end of the experiment, participants were asked optional questions about (1) task comprehension, (2) study experience, (3) appropriate payment, and (4) free-text feedback. No personally identifiable information was collected.

\textbf{Ethics approval.} The experimental design was approved by the [ANONYMIZED] Ethical Review Board.

\subsection{Fitting models to Datasets}\label{app:fitting-models}
Across all datasets, we ran all Bayesian models using the implementation from \cite{hahn2024unifying, Hahn2025Identifiability}, using a grid of 200 points and regularization strength 1.0.

For all the Bayesian models, as well as the reimplemented \cite{zhang2020}'s models, the optimization was performed using a gradient-based approach(Adam or SignSGD).

All fits run on a single consumer-grade GPU (e.g., NVIDIA Titan X) and are also feasible on CPU. Per-subject fit time varies by dataset and model variant, ranging from minutes to hours. Total compute across all paradigms (JRF, lottery pricing, risky choice, bimodal adaptation) and model variants is on the order of GPU-days. Fits are independent across subjects and trivially parallelizable across compute nodes. Reproduction does not require specialized compute infrastructure.

\subsubsection{Details on Model Fit Metrics}\label{app:details-model-fit-metrics}

To compare the performance of all model variants, we use two primary evaluation metrics: summed Heldout $\Delta$NLL and summed $\Delta$AICc.

\paragraph{Summed $\Delta$AICc}\label{app:metric-aicc}

This is the metric used in \cite{zhang2020}. It assesses overall model quality by balancing fit and simplicity:
\begin{enumerate}
    \item For each subject, a model is fit on all trials, yielding the final NLL and the number of free parameters $k$.
    \item We compute the corrected Akaike Information Criterion:
    \begin{equation}
        \text{AICc} = 2\,\text{NLL} + 2k + \frac{2k(k+1)}{n - k - 1}
    \end{equation}
    where $n$ is the number of trials per subject.
    \item For each subject, $\Delta\text{AICc}$ is a given model's AICc minus the lowest AICc across the model set. The Summed $\Delta\text{AICc}$ is the sum of these values across all subjects.
\end{enumerate}

\paragraph{Counting parameters for nonparametric Bayesian models.}
For the nonparametric Bayesian models (FreeE, FreeP), defining the effective number of parameters is non-trivial: regularization shrinks the effective degrees of freedom below the raw grid size. To provide a lower bound on model performance, we report $\Delta\text{AICc}$ using the \textbf{upper bound on parameter count}, i.e., counting every grid point as a free parameter.

For the DMR Pricing task, the trial-by-trial implied weights $\hat{\pi}_{\text{implied},t}$ from Stage 1 (Section~\ref{app:details-bayesian-dmrpricing-models}) are treated as derived targets rather than model parameters, consistent with their role as a non-parametric re-representation of the CE data used for Stage 2 fitting \citep{zhang2020}.

The three noise parameters appearing across tasks correspond to: sensory noise variance, motor variance, and a mixture logit for the guessing rate. Parameter counts for each model variant are listed in Tables~\ref{tab:aicc-params-jrf}--\ref{tab:aicc-params-choice}.

\begin{table}[h]
  \caption{Upper-bound parameter counts for nonparametric Bayesian models on the JRF and Adaptation tasks. Both prior and encoding distributions are represented on 200-point grids.}
  \label{tab:aicc-params-jrf}
  \centering
  \begin{tabular}{lcccc}
    \toprule
    Model & $k$ & Encoding & Prior & Noise \\
    \midrule
    FreeP, FreeE          & 403 & 200          & 200 & 3 \\
    FreeP, UniformE       & 203 & 0            & 200 & 3 \\
    FreeP, BoundedLOE     & 205 & 2 (bounds)   & 200 & 3 \\
    FreeP, PriorMatchedE  & 203 & 0 (tied to prior) & 200 & 3 \\
    \bottomrule
  \end{tabular}
\end{table}

\begin{table}[h]
  \caption{Upper-bound parameter counts for nonparametric Bayesian models on the DMR Pricing task. Stage 2 fits encoding and prior on implied weights with 3 noise parameters; Stage 3 fits the utility exponent $\alpha$ and 2 CE-level noise parameters on the original CE data.}
  \label{tab:aicc-params-pricing}
  \centering
  \begin{tabular}{lccccc}
    \toprule
    Model & $k$ & Enc. & Prior & Stage2 noise & Stage3 ($\alpha$ + CE noise) \\
    \midrule
    FreeP, FreeE          & 406 & 200            & 200 & 3 & $1 + 2$ \\
    FreeP, UniformE       & 206 & 0              & 200 & 3 & $1 + 2$ \\
    FreeP, BoundedLOE     & 208 & 2 (bounds)     & 200 & 3 & $1 + 2$ \\
    FreeP, PriorMatchedE  & 206 & 0 (tied to prior) & 200 & 3 & $1 + 2$ \\
    \bottomrule
  \end{tabular}
\end{table}

\begin{table}[H]
  \caption{Upper-bound parameter counts for nonparametric Bayesian models on the DMR Choice task. Separate 200-point prior grids are fit for the gain and loss domains. The value function contributes 3 parameters and the softmax choice rule contributes 1 temperature parameter. For PriorMatchedE, encoding is tied to the prior separately in each domain, requiring one matching parameter per domain (hence 2 encoding-noise parameters).}
  \label{tab:aicc-params-choice}
  \centering
  \begin{tabular}{lccccccc}
    \toprule
    Model & $k$ & Enc. & Prior(gain) & Prior(loss) & noise & Value fn & Temp \\
    \midrule
    FreeP, FreeE          & 605 & 200          & 200 & 200 & 1 & 3 & 1 \\
    FreeP, UniformE       & 405 & 0            & 200 & 200 & 1 & 3 & 1 \\
    FreeP, BoundedLOE     & 407 & 2   & 200 & 200 & 1 & 3 & 1 \\
    FreeP, PriorMatchedE  & 406 & 0      & 200 & 200 & 2 & 3 & 1 \\
    \bottomrule
  \end{tabular}
\end{table}

\paragraph{Summed Heldout $\Delta$NLL}
This metric measures a model's generalization performance, without penalizing for model complexity. The procedure of using this metric is as follows:
\begin{enumerate}
    \item For each subject, the data is partitioned into a training set (9 out of 10 folds) and a held-out test set (the remaining 1 fold). A model is trained only on the training set.
    \item The trained model is measured by calculating its Negative Log-Likelihood (NLL) on the held-out test set. A lower NLL indicates better predictions.
    \item To compare models for that subject, we find the model with the lowest NLL (the best model). The $\Delta$NLL for any other model is its NLL minus the best model's NLL. The Summed Held-out $\Delta$NLL is the total of these individual $\Delta$NLL scores across all subjects.
\end{enumerate}

\subsubsection{Details on Categorical and Analytical Fitting}~\label{app:analytical-categorical-fitting}
When modeling motor noise, there are two main ways to compute the likelihood of an observed value given the model-predicted $\hat{\theta}_m$. One way is using treat the response value as continuous, we refer to this way as analytical; the other way involves discretizing bins, which is referred to as categorical. In JRF and adaption data, we applied both versions in modelling motor noise in the proportion responses, and in DMR pricing task, we applied the same idea to the Certainty Equivalent (CE) responses.
\paragraph{JRF and Adaptation Task}
\textbf{Analytical Version.}
We treat responses as continuous and assume they are drawn from a Gaussian centered on the Bayesian estimate $\hat{p}(m)$ with variance $\sigma_{\text{motor}}^2$. Because responses are bounded by the grid $[r_{\min}, r_{\max}]$, the Gaussian is truncated and normalized using the corresponding CDF values. This gives the exact continuous likelihood, though it can be numerically unstable when the motor variance is very small. Finally, we mix this motor likelihood with a uniform component to account for guessing:

\begin{align*}
Z &= \Phi\left(\tfrac{p_{\max}-\hat{p}(m)}{\sigma_{\text{motor}}}\right) - \Phi\left(\tfrac{p_{\min}-\hat{p}(m)}{\sigma_{\text{motor}}}\right), \\
P(p_{\text{obs}}\mid m) &= \frac{1}{Z}\frac{1}{\sqrt{2\pi\sigma^2_{\text{motor}}}}
\exp\Big(-\tfrac{(p_{\text{obs}}-\hat{p}(m))^2}{2\sigma^2_{\text{motor}}}\Big), \\
P_{\text{mix}}(p_{\text{obs}}\mid m) &= (1-u)P(p_{\text{obs}}\mid m) + u.
\end{align*}

\textbf{Categorical Version.}
Here we discretize the response axis into bins $\{c_j\}_{j=1}^J$, compute a categorical distribution over bins for each $\hat{p}(m)$, and assign each observed response to its nearest bin. With a fine grid and the bin-width correction, this converges to the analytical solution, but remains stable at very small motor variance:

\begin{align*}
\log P(p_j \mid m) &= \text{logsoftmax}\left(-\tfrac{(p_j - \hat{p}(m))^2}{2\sigma^2_{\text{motor}}}\right), \\
j^{(p_{\text{obs}})} &= \arg\min_j |p_{\text{obs}} - p_j|, \\
\log \tilde{P}(p_{\text{obs}}\mid m) &= \log P(p_{j}\mid m) - \log \Delta c, \\
P_{\text{mix}}(p_{\text{obs}}\mid m) &= (1-u)\tilde{P}(p_{\text{obs}}\mid m) + u \cdot \tfrac{1}{J}.
\end{align*}

In this case, the number of bins is the same the number of the grid size we discretize the input stimuli. We used 200 for the JRF and Adapation datasets.

\paragraph{DMR Pricing Task}
\textbf{Analytical Version.} We treat the CE report as a continuous variable and assume it is drawn from a Gaussian centered on the model prediction $\mu_m$ with variance $\sigma_{\text{motor}}^2$. The likelihood of an observed CE is given directly by this Gaussian density. This provides the exact continuous likelihood, though it can become numerically unstable when $\sigma_{\text{motor}}^2$ is very small. As before, we mix this motor likelihood with a continuous uniform distribution to account for guessing:

\begin{align*}
P(\text{CE}_{\text{obs}}\mid m) &= \exp\left(-\tfrac{1}{2}\left[\tfrac{(\text{CE}_{\text{obs}}-\mu_m)^2}{\sigma^2_{\text{motor}}} + \log\left(2\pi\sigma^2_{\text{motor}}\right)\right]\right) \\
P_{\text{mix}}(\text{CE}_{\text{obs}} \mid m) &= (1-u)\cdot P(\text{CE}_{\text{obs}}\mid m) + u \cdot \frac{1}{\text{CE}_{\max}-\text{CE}_{\min}}.
\end{align*}

\textbf{Categorical Version.} Here we discretize the response axis into bins$ \{c_j\}_{j=1}^J$, model a categorical distribution over bins for each m, and then select the probability of the observed bin. With a fine grid and the -$\log\Delta c$ correction, the categorical method converges to the analytical one:

\begin{align*}
\log P(\text{CE}_j \mid m) &= \text{logsoftmax}\left(-\tfrac{(\hat{p}_m - \text{CE}_j)^2}{2\sigma^2_{\text{motor}}}\right) \\
j^{(\text{CE}_{\text{obs}})} &= \arg\min_j |\text{CE}_{\text{obs}} - \text{CE}_j| \\
\log \tilde{P}(\text{CE}_{\text{obs}} \mid m) &= \log P(\text{CE}_{j} \mid m) - \log \Delta c \\
P_{\text{mix}}(\text{CE}_{\text{obs}} \mid m) &= (1-u)\tilde{P}(\text{CE}_{\text{obs}}\mid m) + u \cdot \tfrac{1}{J}.
\end{align*}

In \cite{zhang2020}'s dataset, $\text{CE}_{\max}=800$ and $\text{CE}_{\min}=0$. We use 1000 grid size for the categorical version.

\subsection{Details on Bayesian Model for JRF Task}\label{app:details-bayesian-jrf}

\subsubsection{Bayesian Model Variants}\label{app:details-bayesian-jrf-models}

We tested several variants of our Bayesian framework by combining different priors and encodings; some are discussed in Section~\ref{sec:jrf}:
\begin{itemize}
    \item \textbf{Priors:}
    \begin{itemize}
        \item \textbf{Uniform Prior (UniformP):} This variant assumes a uniform distribution of prior over the range of possible stimuli(i.e., across all grid points). There is no learnable parameter.
        \item \textbf{Gaussian Prior (GaussianP): }This variant assumes a Gaussian distribution over the range of possible grid points. The two fitted parameters are Gaussian mean and Gaussian standard deviation.
        \item \textbf{Freely Fitted Prior (FreeP):} In this variant, all values of the prior distribution across 200 grid points are treated as trainable parameters. While this allows the model maximal flexibility to fit the data, the large number of trainable parameters can make the calculation of summed $\Delta$AICc challenging. There are 200 freely fitted parameters.
    \end{itemize}
\end{itemize}
\textbf{Priors}
\begin{itemize}

\item \textbf{Encodings}
\begin{itemize}
    \item \textbf{Uniform Encoding (UniformE):} This encoding assumes a uniform distribution of encoding over the range of stimuli. There is no fitted parameter.
    \item \textbf{Fixed Unbounded Log-Odds Encoding (UnboundedLOE):} The encoding is proportional to the log-odds of the stimuli(grid values). This is consistent with the log-odds assumption underlying Zhang's unbounded log-odds models. There is no fitted parameter.
    \item \textbf{Bounded Log-Odds Encoding (BoundedLOE):} This encoding is given by $F'(p) = \frac{1}{(p + \beta_1)(1 - p + \beta_2)}$\footnote{Note that this is equivalent to (\ref{eq:f-bounded-log-odds}) in the setting where different $\beta$'s, written here as $\beta_1, \beta_2$, are allowed for positive and negative  counts, up to an (irrelevant, as it doesn't depend on $p$) proportionality constant: 
    \begin{equation}
        \frac{1}{\sigma} \left(\frac{1}{p+\beta_1} + \frac{1}{(1-p)+\beta_2}\right) \propto \frac{1+\beta_2 +\beta_1}{(p+\beta_1)(1-p+\beta_2)} \propto \frac{1}{(p + \beta_1)(1 - p + \beta_2)}
    \end{equation}}, where $x$ is the value of grid, and $\beta_1$ and $\beta_2$ are small, positive, learnable parameters between 0 and 1, which we bound using the sigmoid function. This form is motivated by Zhang's BLO models and our discussion in Section~\ref{sec:log-odds}. There are two fitted parameters: $\beta_1$ and $\beta_2$.
    \item \textbf{Prior-matched Encoding(PriorMatchedE): }This encoding assumes that the resources is proportional to the prior distribution, i.e., the encoding density is identical to the prior. There are no additional fitted parameters beyond those of the prior.   
    \item \textbf{Freely Fitted Encoding (FreeE):} Similar to freely fitted prior. There are 200 freely fitted parameters.
\end{itemize}
\end{itemize}

Parameters for each Bayesian model variant, including prior parameters, encoding parameters, sensory noise variance, motor variance and mixture logit are optimized against the subject-level data using the same gradient based method described for BLO models on the same task.

\subsubsection{Performance Comparison of All Model variants}\label{app:details-bayesian-jrf-performance}
In this part, we show the performance of model variants on both evaluation metrics(Summed Heldout $\Delta$NLL and $\Delta$AICc) and with both analytical and categorical fitting methods in Figure\ref{fig:jrf-dnll-full-both-analytical-categorical} and ~\ref{fig:jrf-daicc-full-both-analytical-categorical}\footnote{Note that we report the Summed $\Delta$AICc metric for both parametric and nonparametric model variants. For  nonparametric models, we use the upper bound as described in Appendix~\ref{app:details-model-fit-metrics}.}. The Bayesian model variants with red color in the figures are detailed in Appendix Section~\ref{app:details-bayesian-jrf-models}. The model variants from \cite{zhang2020} with blue color in the figures are detailed in Appendix Section~\ref{app:details-blo-jrf-models}.

Overall, our model performance better than \cite{zhang2020}'s model variants across two metrics and two fitting approaches. Within Bayesian models, models using a bounded log-odds encoding outperform those with an unbounded encoding, and a Gaussian prior is superior to a uniform prior. Within \cite{zhang2020}'s model variants, bounded model performs better than unbounded model, and using $V(p)$(explained in Eq~\ref{eq:vp}) is better than assuming a constant value $V$.

\begin{figure}[H]
    \centering
    \begin{subfigure}{0.48\linewidth}
        \centering
        \includegraphics[width=\linewidth]{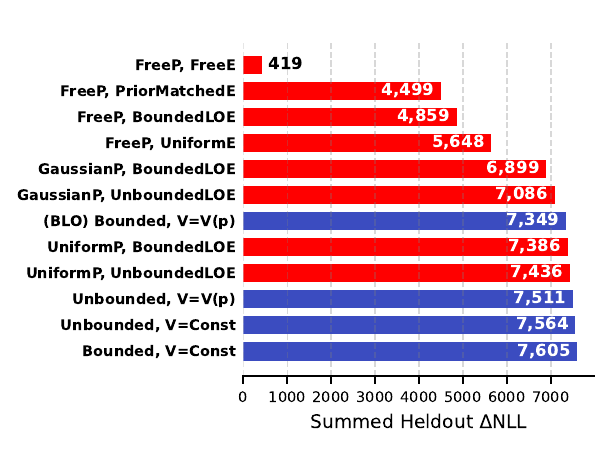}
        \caption{Analytical}
        \label{fig:jrf-dnll-analytical}
    \end{subfigure}
    \hfill
    \begin{subfigure}{0.48\linewidth}
        \centering
        \includegraphics[width=\linewidth]{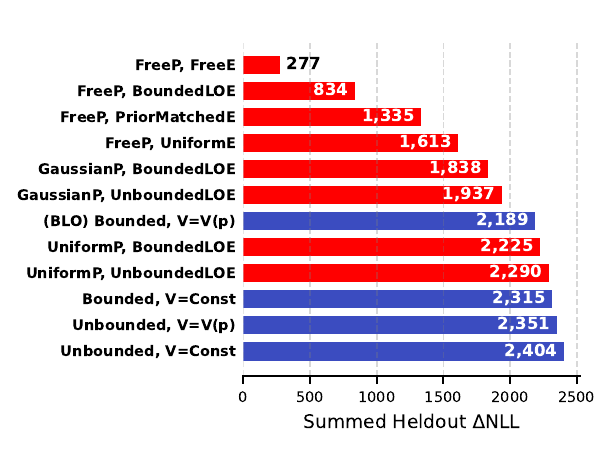}
        \caption{Categorical}
        \label{fig:jrf-dnll-categorical}
    \end{subfigure}
    \caption{\textbf{JRF Task:} Performance of the models, measured by the Summed Heldout $\Delta$NLL metric. See Appendix~\ref{app:details-bayesian-jrf-models} for the shorthands for Bayesian models (red). We refer to \cite{zhang2020} for the shorthands for their model variants (blue). Analytical and categorical fitting methods are explained in Appendix~\ref{app:analytical-categorical-fitting}. The results shown in the main text correspond to Analytical. This result is mentioned in main text section~\ref{sec:jrf}}
    \label{fig:jrf-dnll-full-both-analytical-categorical}
\end{figure}

\begin{figure}[H]
    \centering
    \begin{subfigure}{0.48\linewidth}
        \centering
        \includegraphics[width=\linewidth]{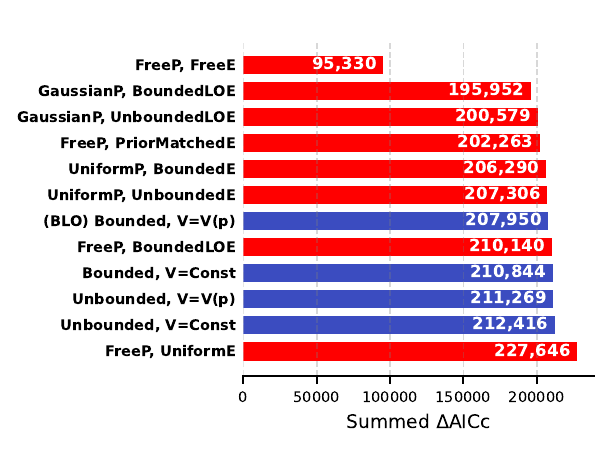}
        \caption{Analytical}
        \label{fig:jrf-daicc-analytical}
    \end{subfigure}
    \hfill
    \begin{subfigure}{0.48\linewidth}
        \centering
        \includegraphics[width=\linewidth]{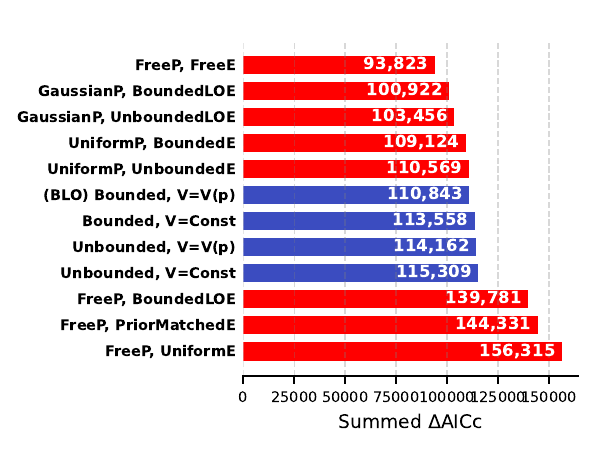}
        \caption{Categorical}
        \label{fig:jrf-daicc-categorical}
    \end{subfigure}
    \caption{\textbf{JRF Task:} Performance of the models, measured by the Summed $\Delta$AICc metric.  See Appendix~\ref{app:details-bayesian-jrf-models} for the shorthands for Bayesian models (red). We refer to \cite{zhang2020} for the shorthands for their model variants (blue). Analytical and categorical fitting methods are explained in Appendix~\ref{app:analytical-categorical-fitting}. The results shown in the main text correspond to Analytical.}
    \label{fig:jrf-daicc-full-both-analytical-categorical}
\end{figure}

\subsubsection{Analysis of Fitted Prior and Resources}
Figure~\ref{fig:jrf-fi-bayesian-perSubject} plots the fitted resources for our Bayesian model (red) and the BLO model (blue) for 86 subjects in the JRF task. For nearly all subjects, the resources from both models are U-shaped, with peaks near the probability endpoints of $p=0$ and $p=1$.

A closer examination reveals a difference. As shown in Figure~\ref{fig:jrf-fi-blo-perSubject}, which restricts the y-axis for clarity, the resources of the BLO model exhibit two points of discontinuity. This finding is formally predicted by and consistent with our Theorem~\ref{thm:blo-FI}.

\begin{figure}[H]
    \centering
        \includegraphics[width=0.75\linewidth]{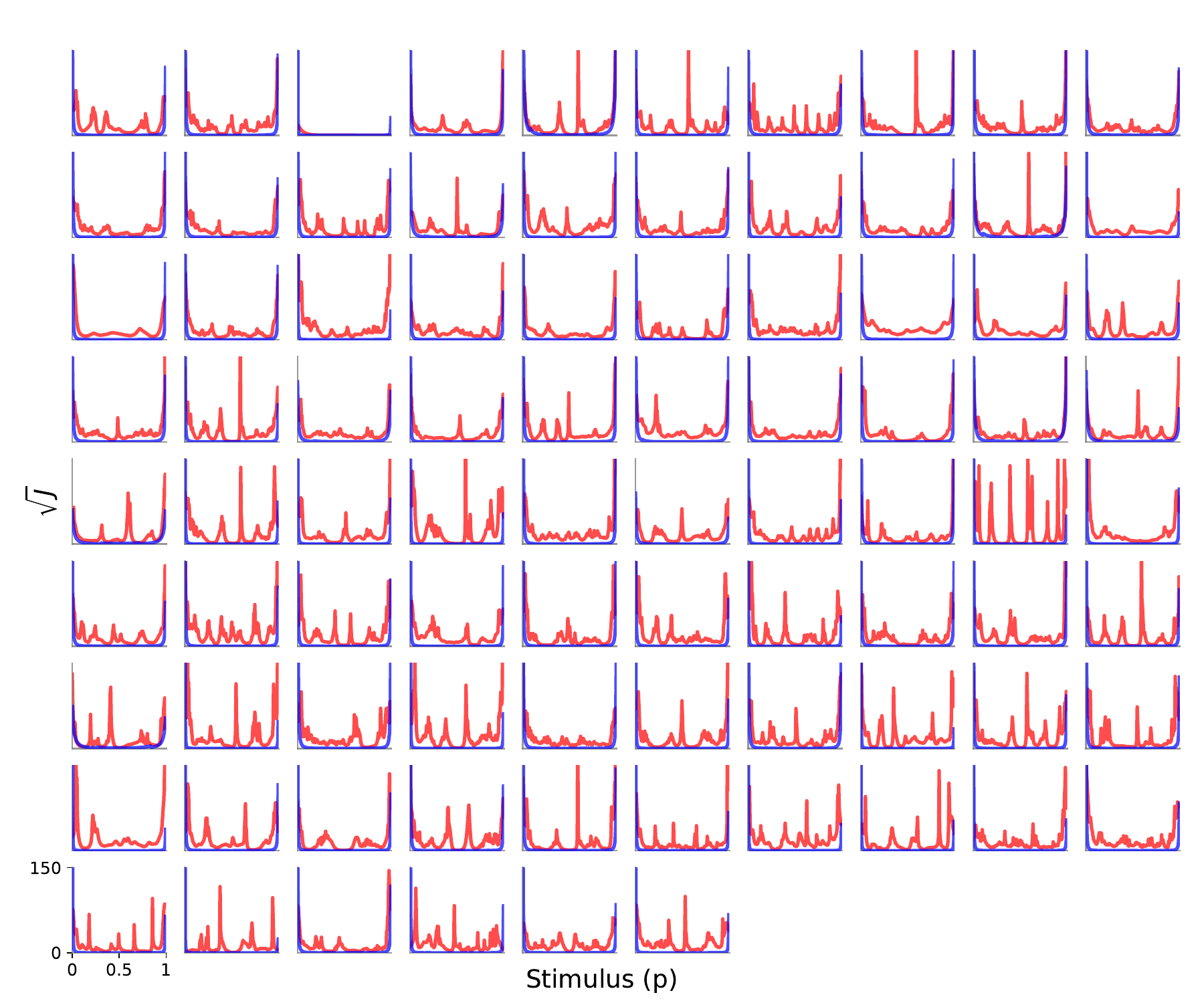}
        \caption{\textbf{JRF Task: } Per-subject resources ($\sqrt{J(p)}$) for the Bayesian model with freely fitted prior(in red) and encoding and the BLO model(in blue). Subjects S1 to S51 are from dataset JDA, S52 to S75 from dataset JDB (both from \cite{zhang2020}), and S76 to S86 from dataset ZM12 (\cite{ZM2012}).}
        \label{fig:jrf-fi-bayesian-perSubject}
\end{figure}

\begin{figure}[H]
    \centering
    \includegraphics[width=0.75\linewidth]{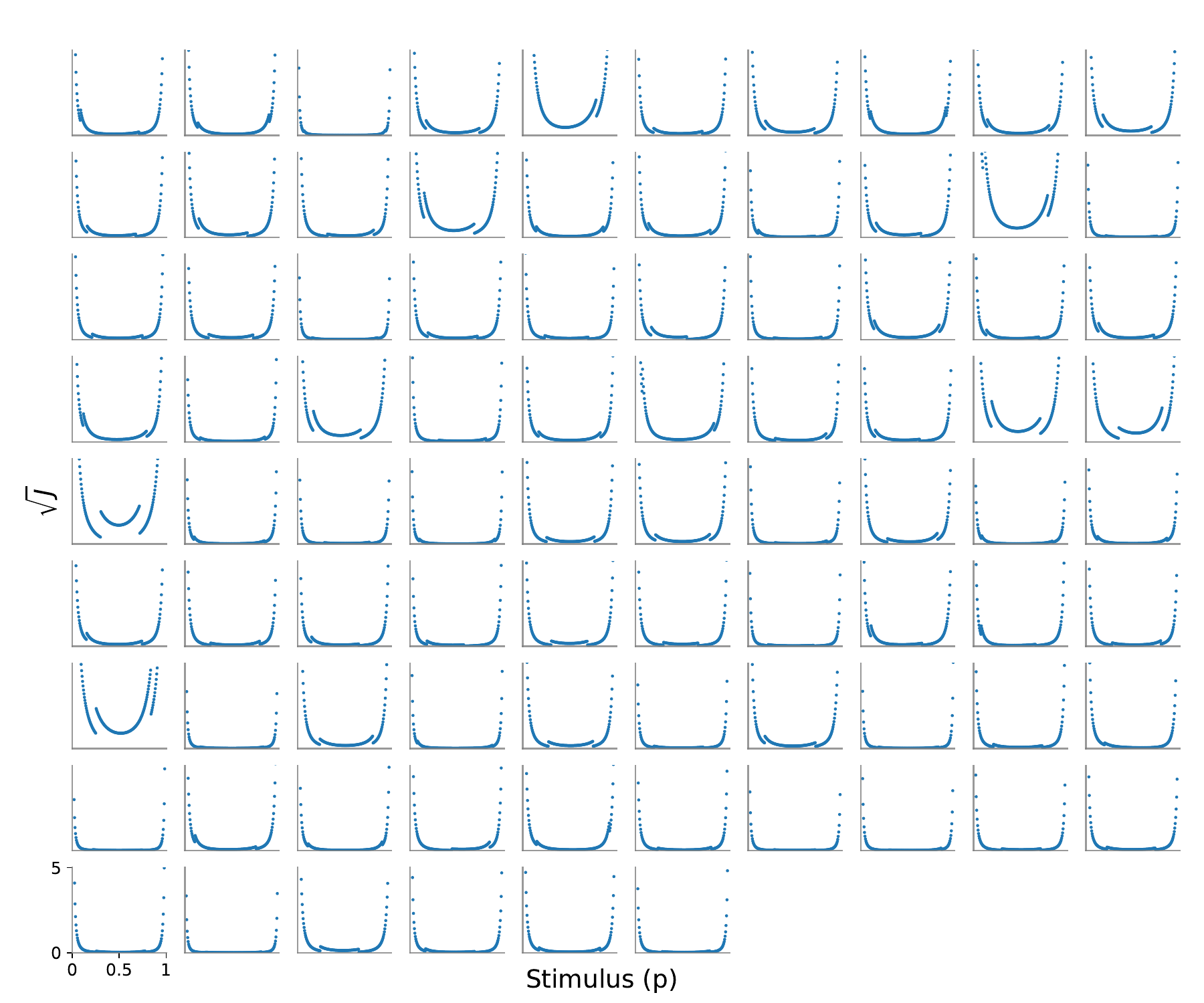}
    \caption{\textbf{JRF Task: }Discontinuity observed in the resources ($\sqrt{J(p)}$) of the BLO model. The proof of this property is in Theorem~\ref{thm:blo-FI}. Subjects S1 to S51 are from dataset JDA, S52 to S75 from dataset JDB (both from \cite{zhang2020}), and S76 to S86 from dataset ZM12 (\cite{ZM2012}).}
    \label{fig:jrf-fi-blo-perSubject}
\end{figure}

Figure~\ref{fig:jrf-prior-comparison} compares the group-level fitted priors for the Bayesian model variants. The parametric Gaussian prior is shown to capture the main features of the non-parametric, freely fitted prior. The freely fitted prior with matched encoding shows not only peaks at 0.5, but also at 0 and 1, which is a property of the Resources.

Figure~\ref{fig:jrf-fi-comparison} compares the group-level resources for all Bayesian model variants. The resources are U-shaped for the parametric models (the Bayesian log-odds variants), and the freely fitted resources also recover this U-shape.

\begin{figure}[H]
\begin{center}
\includegraphics[width=0.6\linewidth]{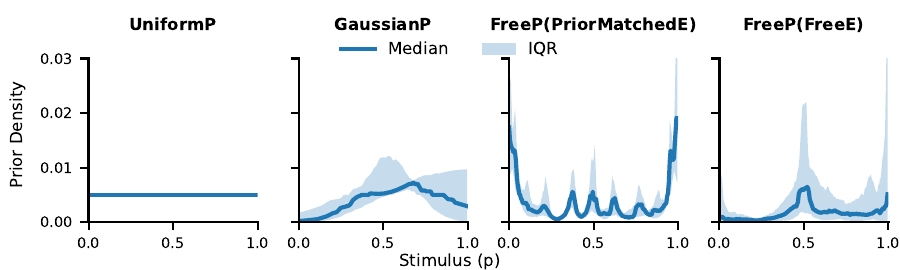}
\end{center}
\caption{
\textbf{JRF Task: }Group-level priors for the four prior components evaluated in this task. The solid line shows the group median, and the shaded area indicates the interquartile range (IQR). Plotting details are provided in Appendix~\ref{app:plotting}.}
\label{fig:jrf-prior-comparison}
\end{figure}

\begin{figure}[H]
\begin{center}
\includegraphics[width=0.8\linewidth]{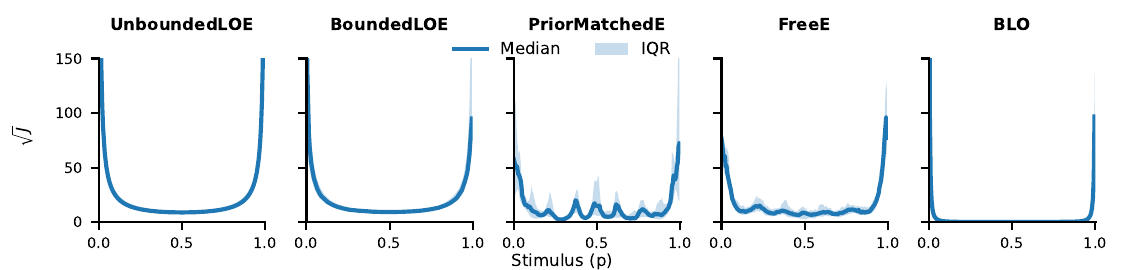}
\end{center}
\caption{\textbf{JRF Task: }Group-level resources ($\sqrt{J(p)}$) for the four encoding components evaluated in this task, along with the resources of the BLO model. The solid line shows the group median, and the shaded area indicates the interquartile range (IQR). Plotting details are provided in Appendix~\ref{app:plotting}. This result is mentioned in main text section~\ref{sec:jrf}.
}
\label{fig:jrf-fi-comparison}
\end{figure}

\subsubsection{Analysis of Bias and Variance}
Figures \ref{fig:jrf-bias-perSubject} and \ref{fig:jrf-var-perSubject} show the per-subject bias and variance, respectively, providing a more detailed view of the group-level results presented in Figure~\ref{fig:jrf-bias-var}. The methods used to calculate these quantities are detailed in Appendix~\ref{app:methods-bias-var}.

The per-subject plots confirm the main findings. For bias, both the Bayesian model and the BLO model capture the bis pattern of the non-parametric estimates. The key divergence appears in the response variability. Figure~\ref{fig:jrf-var-perSubject} clearly shows that while our Bayesian model captures subject-level variability, the BLO model consistently fails to model the dip at $p=0.5$.

\begin{figure}[H]
\begin{center}
\includegraphics[width=0.75\linewidth]{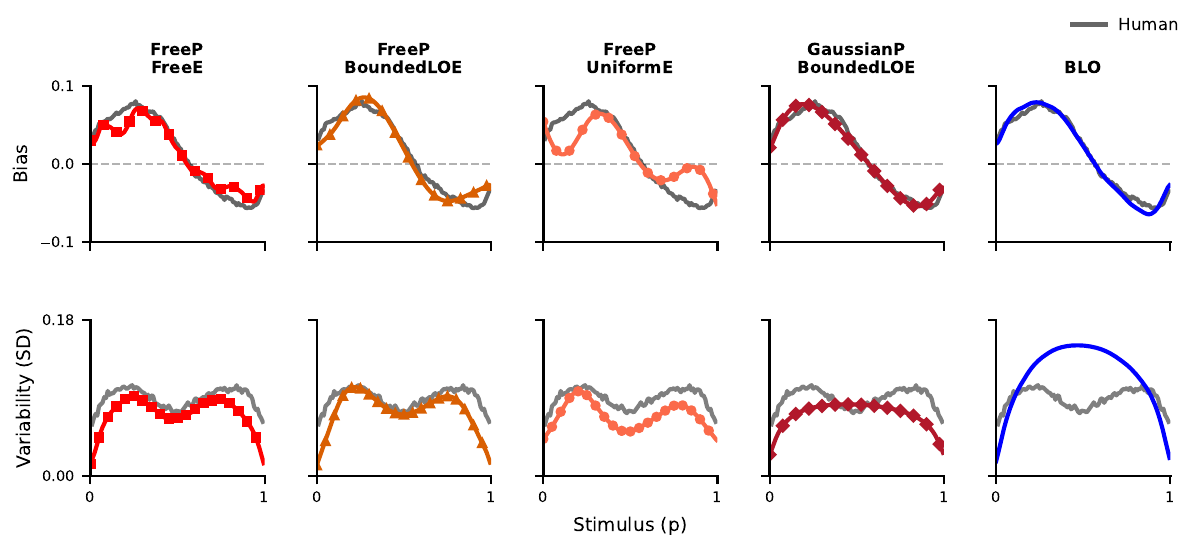}
\end{center}
\caption{\textbf{JRF Task: } Model fits to response bias (top) and variability (bottom). Gray curves show human data. The S-shaped bias is only explained by non-uniform encodings (BoundedLOE, FreeE); the variability is only captured by the nonparametric prior (FreeP).} 
\label{fig:jrf-bias-var}
\end{figure}

\begin{figure}[H]
\begin{center}
\includegraphics[width=0.75\linewidth]{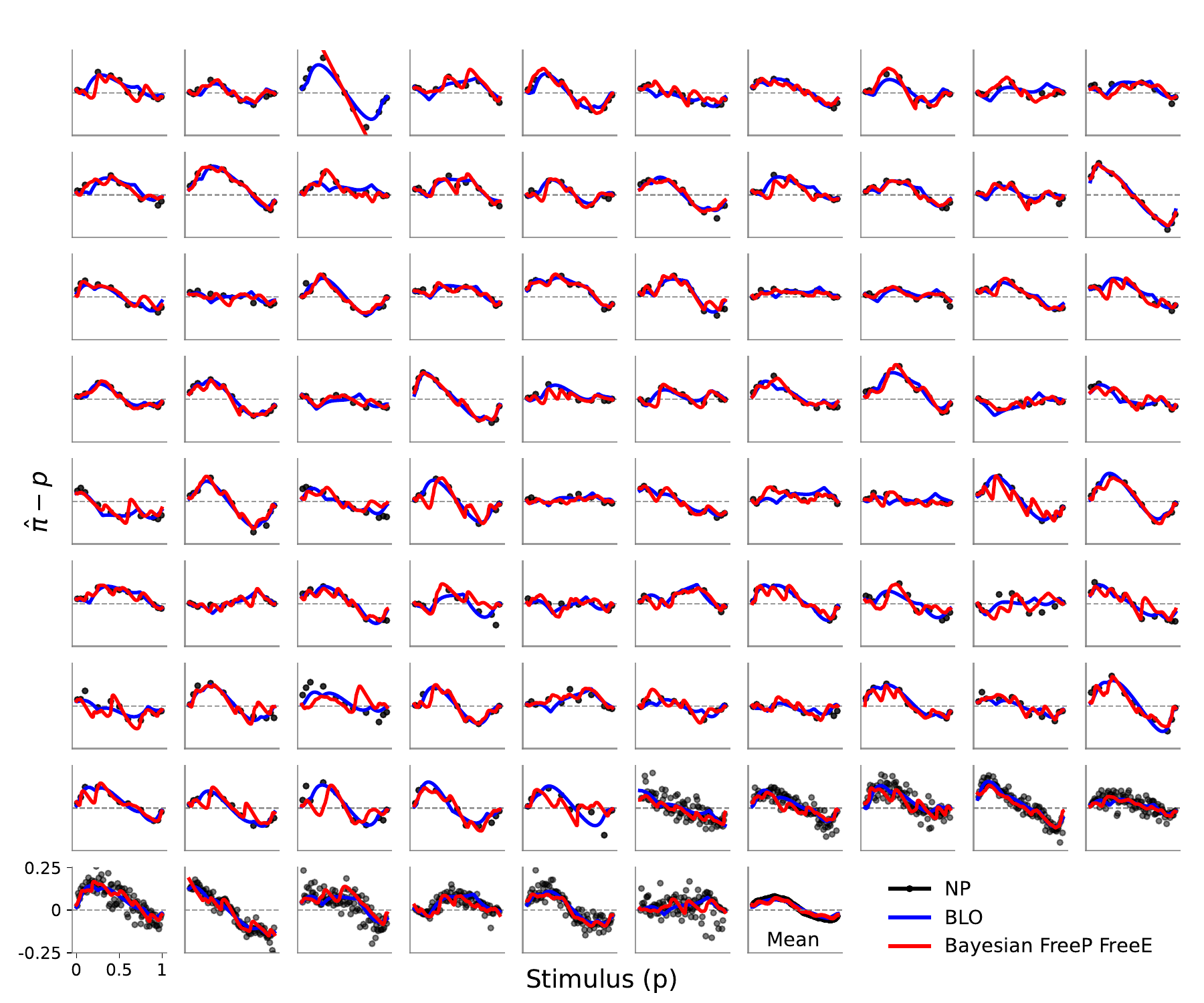}
\end{center}
\caption{\textbf{JRF Task: }Per-subject bias of the non-parametric estimates from the data,  the BLO model and the Bayesian model with free prior and free encoding. Subjects S1 to S51 are from dataset JDA, S52 to S75 from dataset JDB (both from \cite{zhang2020}), and S76 to S86 from dataset ZM12 (\cite{ZM2012}).} 
\label{fig:jrf-bias-perSubject}
\end{figure}

\begin{figure}[H]
\begin{center}
\includegraphics[width=0.75\linewidth]{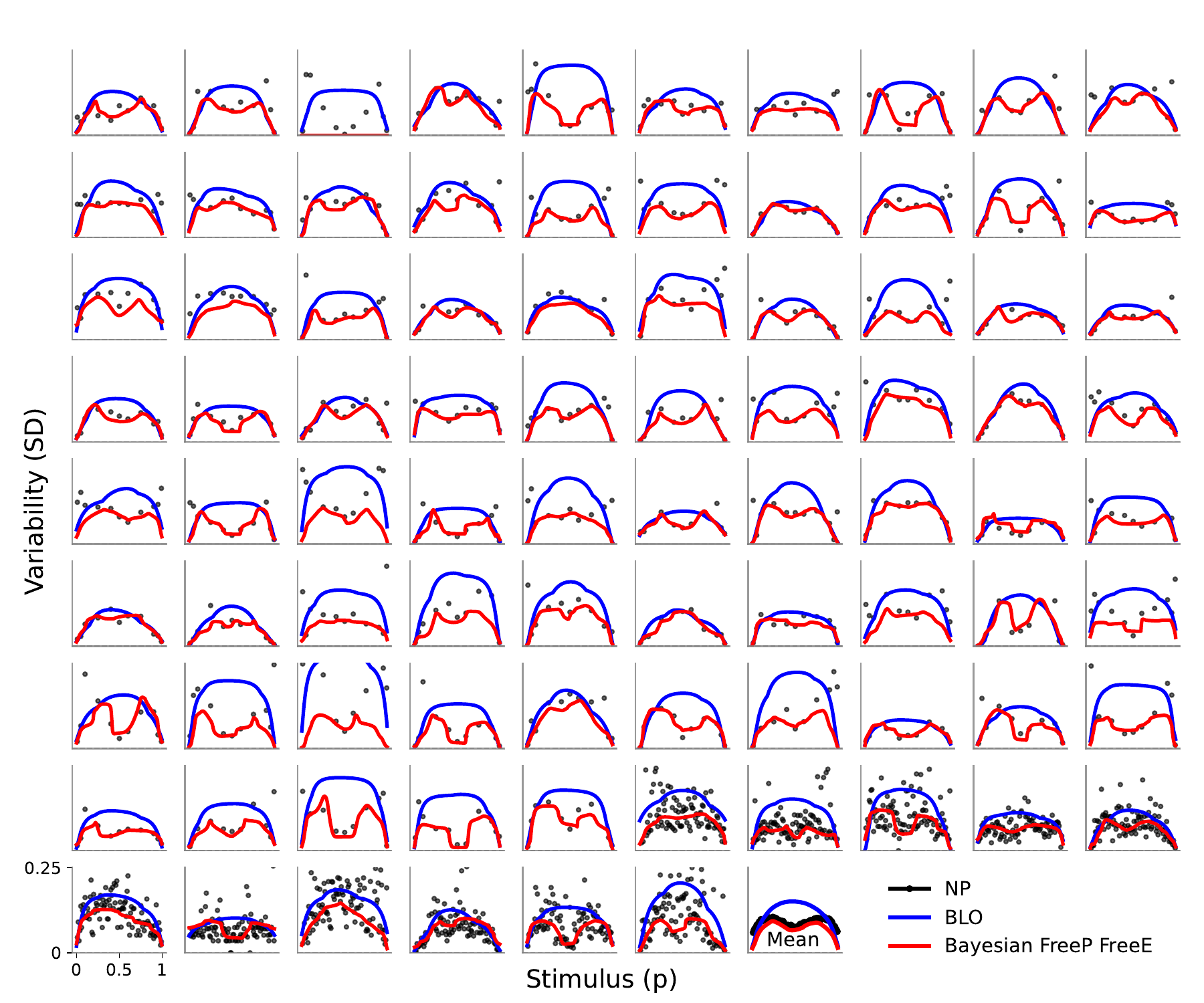}
\end{center}
\caption{\textbf{JRF Task: }Per-subject variability of the non-parametric estimates from the data, the BLO model and the Bayesian model with free prior and free encoding. Subjects S1 to S51 are from dataset JDA, S52 to S75 from dataset JDB (both from \cite{zhang2020}), and S76 to S86 from dataset ZM12 (\cite{ZM2012}).
} 
\label{fig:jrf-var-perSubject}
\end{figure}

Moreover, Figure~\ref{fig:jrf-bias-decomposition} decompose the bias of four Bayesian model variants into attraction, repulsion and boundary regression. Both attractive and repulsive components point away from 0 and 1. Both are needed to account for the overall distortion, as both prior and encoding need to be nonuniform to achieve good model fit.

\subsection{Details on Bayesian Model for DMR Pricing Task}\label{app:details-bayesian-dmr-pricing}
\subsubsection{Bayesian Model Variants and Fitting Procdure}\label{app:details-bayesian-dmrpricing-models}

We tested the similar Bayesian model variants for the pricing task as we did for the JRF task (detailed in Appendix~\ref{app:details-bayesian-jrf-models}). However, fitting these models to the Certainty Equivalent (CE) data required a specific two-stage procedure. The goal was to first convert each observed CE into an ``implied" subjective probability weight, and then fit our Bayesian models to these weights.

\paragraph{Stage 1: Deriving trial-by-trial implied estimates.}
The goal of this initial stage was to convert each raw CE response into a non-parametric, trial-level estimate of the subjective probability weight,  $\hat{\pi}_{\text{implied}, t}(p)$. For each trial $t$, we fits a free variable. We also fit $\alpha$ for applying the CPT utility function. To account for additional variability in CE, we included an extra noise term $\epsilon_\text{CE}$ and optimized parameters by minimizing the loss between predicted and observed CE. This method is similar to that of \cite{zhang2020}, but our implementation works on a trial-by-trial basis rather than on 11 discrete probability levels.

\paragraph{Stage 2: Fitting Bayesian estimator.}  
The set of ($p_t$, $\hat{\pi}_{\text{implied}, t}$) pairs derived from Stage 1 was then used as the target data to fit the parameters of our Bayesian model. These parameters in turn determine the set of optimal point estimates, $\hat{p}(m)$ (the decoded stimulus value for each possible internal measurement m), by maximizing the likelihood of the implied weights.

\paragraph{Stage 3: Final Likelihood Maximization on Original CE Data.}
To ensure a fair and direct comparison with the BLO model, the final model evaluation was based on the likelihood of the original CE data. In this final stage, the Bayesian encoding parameters (and therefore the set of $\theta(m)$ values) were held fixed from the results of Stage 2. We then performed a final optimization to find the subject's remaining parameters—the utility exponent $\alpha$ and the CE noise variance $\sigma_\text{CE}$—that maximized the log-likelihood of their observed CE responses. The resulting maximum log-likelihood value was then used to calculate the Held-out $\Delta$NLL and $\Delta$AICc scores.

\paragraph{Model variants.}  
Because stage 2 closely resembles the JRF task, we applied the same set of model variants used there (Appendix~\ref{app:details-bayesian-jrf-models}), with the exception of the Fixed Unbounded Log-Odds Encoding. We excluded this variant to focus on encoding schemes that showed better performance in this task.


\subsubsection{Performance Comparison of All Model Variants}\label{app:performance-dmr-pricing}
Figure~\ref{fig:dmr-dnll-full-both-analytical-categorical} and Figure~\ref{fig:dmr-daicc-full-both-analytical-categorical} presents the performance of model variants on both evaluation metrics. 

For the likelihood of the observed CE data, we chose to present results from a categorical likelihood function in the main text. While a fully analytical (continuous Gaussian) likelihood is possible—and in our tests, this analytical version of our model achieves a lower summed $\Delta$heldout loss than Zhang's models—we opted for the categorical approach. We argue that the DMR task, which involves a comparative judgment, is inherently more categorical in nature than the continuous estimation required in the JRF task. To ensure a fair and direct comparison, we re-evaluated Zhang's original models using this identical categorical likelihood. 

Across both evaluation methods, the Bayesian models (red bars) consistently outperformed Zhang et al.’s variants (blue bars). In particular, when measured by summed $\Delta$NLL, the Bayesian models achieved substantially smaller losses, indicating a much better account of the observed CE distributions. The advantage is also evident under the summed $\Delta$AICc, where Bayesian models again dominate.

\begin{figure}[H]
    \centering
    \begin{subfigure}{0.48\linewidth}
        \centering
        \includegraphics[width=\linewidth]{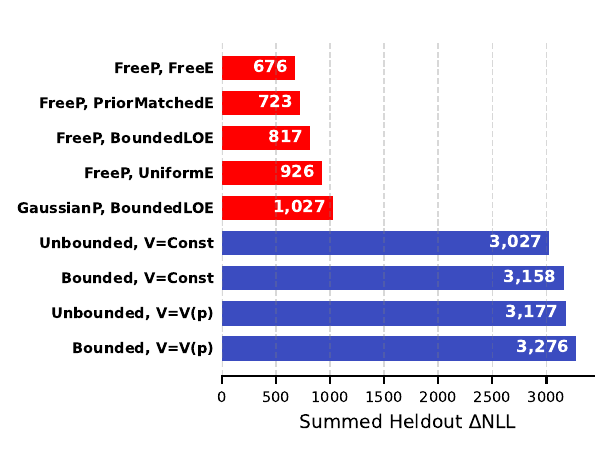}
        \caption{Analytical}
        \label{fig:dmr-dnll-analytical}
    \end{subfigure}
    \hfill
    \begin{subfigure}{0.48\linewidth}
        \centering
        \includegraphics[width=\linewidth]{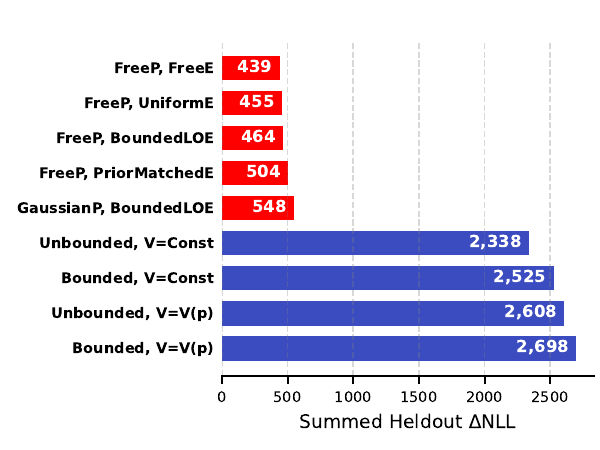}
        \caption{Categorical}
        \label{fig:dmr-dnll-categorical}
    \end{subfigure}
    \caption{\textbf{DMR Pricing Task: } Performance of the models, measured by the Heldout $\Delta$NLL metric.  See Appendix~\ref{app:details-bayesian-jrf-models} for the shorthands for Bayesian models (red). We refer to \cite{zhang2020} for the shorthands for their model variants (blue). Analytical and categorical fitting methods are explained in Appendix~\ref{app:analytical-categorical-fitting}. The results shown in the main text correspond to Categorical.}
    \label{fig:dmr-dnll-full-both-analytical-categorical}
\end{figure}

\begin{figure}[H]
    \centering
    \begin{subfigure}{0.48\linewidth}
        \centering
        \includegraphics[width=\linewidth]{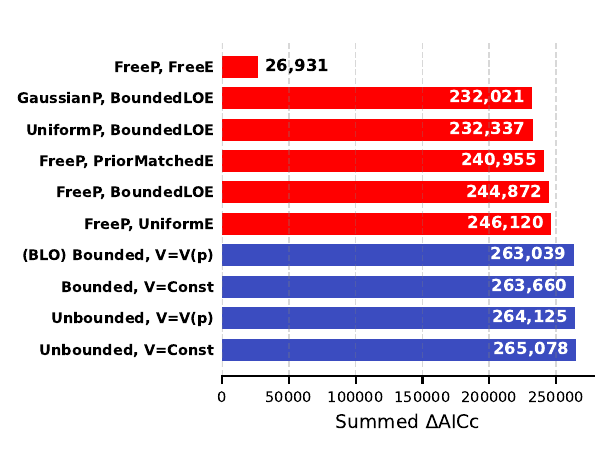}
        \caption{Analytical}
        \label{fig:dmr-daicc-analytical}
    \end{subfigure}
    \hfill
    \begin{subfigure}{0.48\linewidth}
        \centering
        \includegraphics[width=\linewidth]{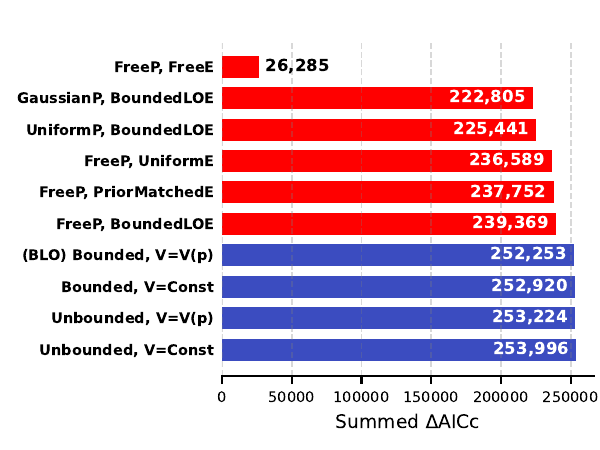}
        \caption{Categorical}
        \label{fig:dmr-daicc-categorical}
    \end{subfigure}
    \caption{\textbf{DMR Pricing Task: }Performance of the models, measured by the $\Delta$AICc metric. See Appendix~\ref{app:details-bayesian-jrf-models} for the shorthands for Bayesian models (red). We refer to \cite{zhang2020} for the shorthands for their model variants (blue). Analytical and categorical fitting methods are explained in Appendix~\ref{app:analytical-categorical-fitting}. The results shown in the main text correspond to Categorical.}
    \label{fig:dmr-daicc-full-both-analytical-categorical}
\end{figure}

\subsubsection{Analysis of Fitted Prior and Resources}
The fitted resources in Figure~\ref{fig:dmr-enc-comparison} closely resemble those obtained in the JRF task.  The BLO model doesn't have meaningful resources because it doesn't fit a noise in the encoding phase.

For the prior, the freely fitted version exhibits a shape similar to the Gaussian prior, which accounts for the strong performance of Bayesian models with Gaussian priors reported in Section~\ref{app:performance-dmr-pricing}.

\begin{figure}[H]
\begin{center}
\includegraphics[width=0.6\linewidth]{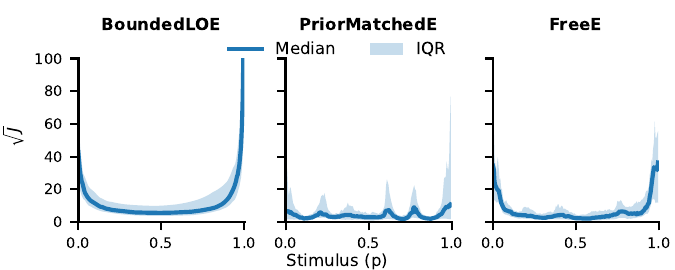}
\end{center}
\caption{\textbf{DMR Pricing Task: }Group-level resources ($\sqrt{J(p)}$) for the three encoding components evaluated in this task. The solid line shows the group median, and the shaded area indicates the interquartile range (IQR). Plotting details are provided in Appendix~\ref{app:plotting}.}
\label{fig:dmr-enc-comparison}
\end{figure}

\begin{figure}[H]
\begin{center}
\includegraphics[width=0.7\linewidth]{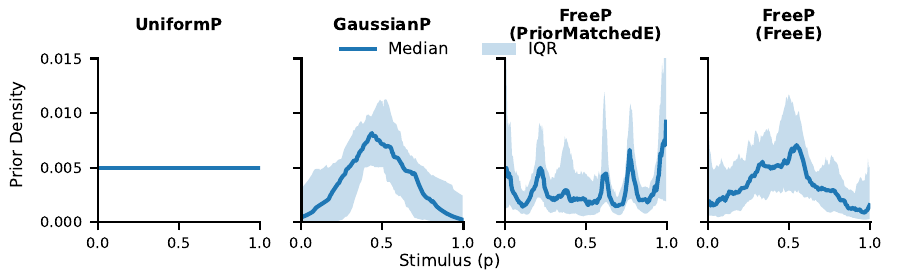}
\end{center}
\caption{\textbf{DMR Pricing Task: }Group-level priors for the four prior components evaluated in this task. The solid line shows the group median, and the shaded area indicates the interquartile range (IQR). Plotting details are provided in Appendix~\ref{app:plotting}. Free fitting is compatible with a unimodal prior. }
\label{fig:dmr-prior-comparison}
\end{figure}

\subsubsection{Analysis of Bias and Variance}

\begin{figure}[H]
\begin{center}
\includegraphics[width=0.75\linewidth]{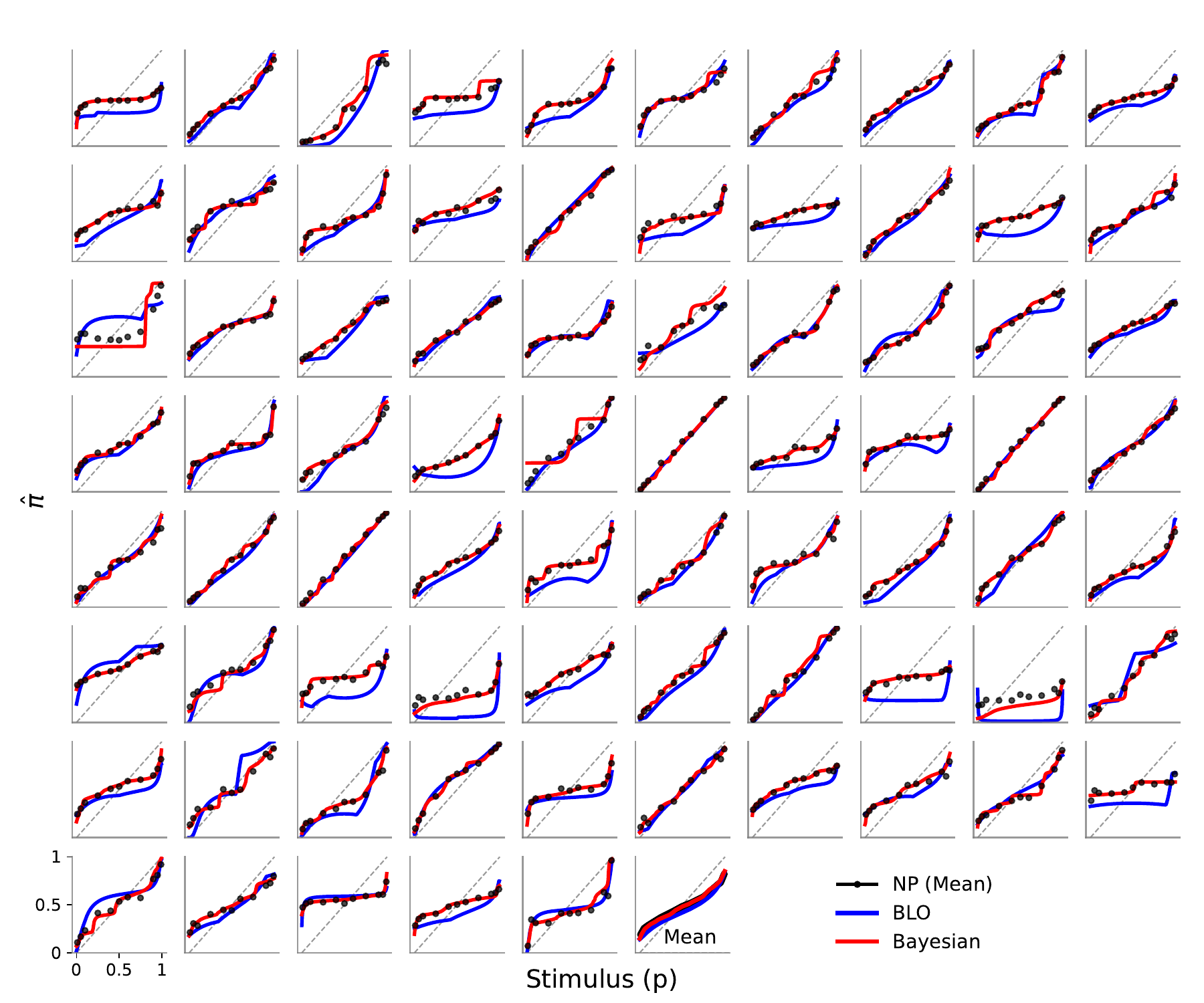}
\end{center}
\caption{\textbf{DMR Pricing Task: }Per-subject bias of the non-parametric estimates from the data, the BLO model and the Bayesian model with free prior and free encoding on the DMR pricing task. Subjects S1 to S51 are from dataset JDA, and S52 to S75 are from dataset JDB (both from \cite{zhang2020}).
}
 \label{fig:dmr-bias}
\end{figure}
In Figure~\ref{fig:dmr-bias}, we present the per-subject bias of the probability estimate, $\hat{\pi}$, for both the Bayesian and BLO models. As the figures show, both models capture the general pattern of bias for most subjects.
A direct comparison of the variance of  $\hat{\pi}$ is not shown because the two models treat this quantity fundamentally differently. The BLO model's estimate $\hat\pi_\text{BLO}$ is a deterministic point value and thus has zero variance.  In contrast, our Bayesian estimate, $\hat\pi_\text{Bayesian}$ , is the mean of a full posterior distribution and has inherent sensory noise variance. 
The methods used to calculate the non-parametric, BLO's, and Bayesian model's bias are detailed in Appendix~\ref{app:methods-bias-var}.

\subsection{Details on Bayesian Model for DMR Choice Task}\label{app:details-bayesian-dmr-choice}
\subsubsection{Full Utility Function, Logit Choice Rule, and LILO function}
\paragraph{Full Utility Function}
\begin{equation}\label{eq:full-cpt}
    \text{Utility} = \begin{cases}
w^+(p)v(X) + (1 - w^+(p))v(Y) & \text{if } X \ge Y \ge 0 \\
(1 - w^-(1 - p))v(X) + w^-(1 - p)v(Y) & \text{if } X > Y \text{ and } X, Y < 0 \\
w^+(p)v(X) + w^-(1 - p)v(Y) & \text{if } X > 0 > Y
\end{cases}
\end{equation}

\paragraph{LILO function}
($\gamma,\beta$ are free parameters):
\begin{equation}\label{eq:lilo}
\widehat{p} = \lambda^{-1}(\gamma \lambda(p) + \beta), \quad \lambda(p) := \log \frac{p}{1-p}.
\end{equation}

 This equation is mentioned in main text section~\ref{sec:dmr}

\paragraph{Logit Choice Rule}
\begin{equation}
P(\text{Choose B}) = \text{sigmoid}(\tau \cdot(\text{Utility}_B - \text{Utility}_A))
\end{equation}

The parameter $\tau$ ($\tau > 0$,) controls the choice sensitivity and is a free parameter to fit.

\subsubsection{Bayesian Model Variants and Fitting Procedure}\label{app:details-bayesian-dmr-choice-models}
We tested several model variants:

\begin{itemize}
    \item \textbf{Priors:}
    \begin{itemize}
        \item \textbf{Freely Fitted Prior:}  
        As we mainly focus on validating  the resources shape with this task, we only evaluate the freely fitted prior, which is the same as described in Appendix~\ref{app:details-bayesian-jrf-models}. 
    \end{itemize}
    
    \item \textbf{Encodings:}  
    \begin{itemize}
        \item \textbf{Bounded Log-odds Encoding, Prior-matched Encoding, Freely Fitted Encoding:} Same as in Appendix~\ref{app:details-bayesian-jrf-models}.
    \end{itemize}
\end{itemize}

Because the dataset includes both gains and losses, we fit separate probability weighting functions for them. Specifically, we estimate separate priors for gains and for losses, while assuming a shared encoding across both. The resulting probability weighting functions are then computed as the expectation of $\hat{p}$ under the encoding distribution.

\subsubsection{Performance Comparison of All Model Variants}
Figure~\ref{fig:dmr-choice-dnll-full-both-analytical-categorical} shows the performance of different models on the DMR choice task. Among all tested models, the Bayesian variant with freely fitted prior and encoding (FreeP, FreeE) achieve the best fit, clearly outperforming both other Bayesian variants and classical parametric weighting functions. Models with parametric encoding perform worse, and parametric models such as LILO and Prelec also show substantially higher $\Delta$AICc. The Bayesian model with uniform encodoing performs the worst, showing that this encoding cannot capture the probability distortion in this task.

\begin{figure}[H]
    \centering
    \includegraphics[width=0.48\linewidth]{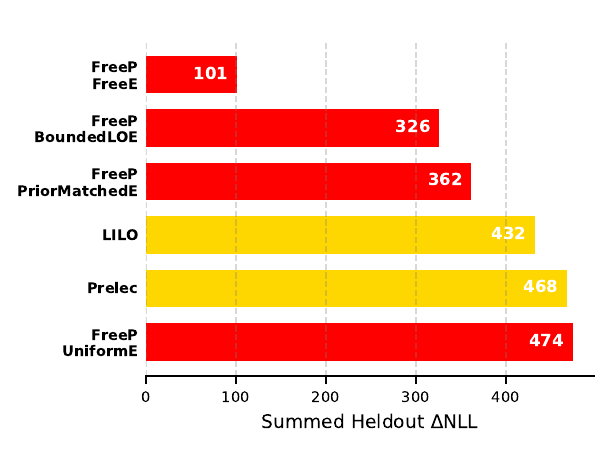}
    \caption{\textbf{DMR Choice Task:} Performance of the models, measured by the Heldout $\Delta$NLL metric. See Appendix~\ref{app:details-bayesian-jrf-models} for the shorthands for Bayesian models (red).}
    \label{fig:dmr-choice-dnll-full-both-analytical-categorical}
\end{figure}

\subsubsection{Analysis of Fitted Prior and Resources}
We show the group-level encoding resources (Figure~\ref{fig:dmr-enc-comparison-choice}) and priors (Figure~\ref{fig:dmr-prior-comparison-choice}). Across models, the fitted resources consistently U-shaped, with peaks at 0 and 1. Interestingly, the freely fitted priors for both gains and losses are also U-shaped.

\begin{figure}[H]
\begin{center}
\includegraphics[width=0.7\linewidth]{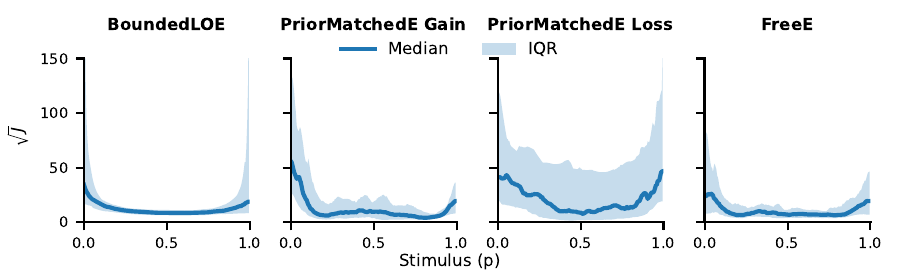}
\end{center}
\caption{\textbf{DMR Choice Task: }Group-level resources ($\sqrt{J(p)}$) for the four encoding components evaluated in this task. The solid line shows the group median, and the shaded area indicates the interquartile range (IQR). Plotting details are provided in Appendix~\ref{app:plotting}.}
\label{fig:dmr-enc-comparison-choice}
\end{figure}

\begin{figure}[H]
\begin{center}
\includegraphics[width=0.4\linewidth]{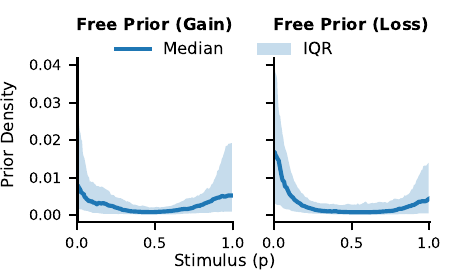}
\end{center}
\caption{\textbf{DMR Choice Task: }Group-level prior of the Bayesian model with freely fitted priors and encoding. The solid line shows the group median, and the shaded area indicates the interquartile range (IQR). Plotting details are provided in Appendix~\ref{app:plotting}.}
\label{fig:dmr-prior-comparison-choice}
\end{figure}

\subsubsection{Analysis of Bias}
Figure~\ref{fig:dmr-choice-perSubject-bias} shows the per-subject bias for the Bayesian freely fitted model, the LILO model, and the Prelec model. Across subjects, the three models tend to capture similar patterns of bias, although the detailed shapes at the individual level likely reflect a considerable degree of overfitting.

\begin{figure}[H]
\begin{center}
\includegraphics[width=0.7\linewidth]{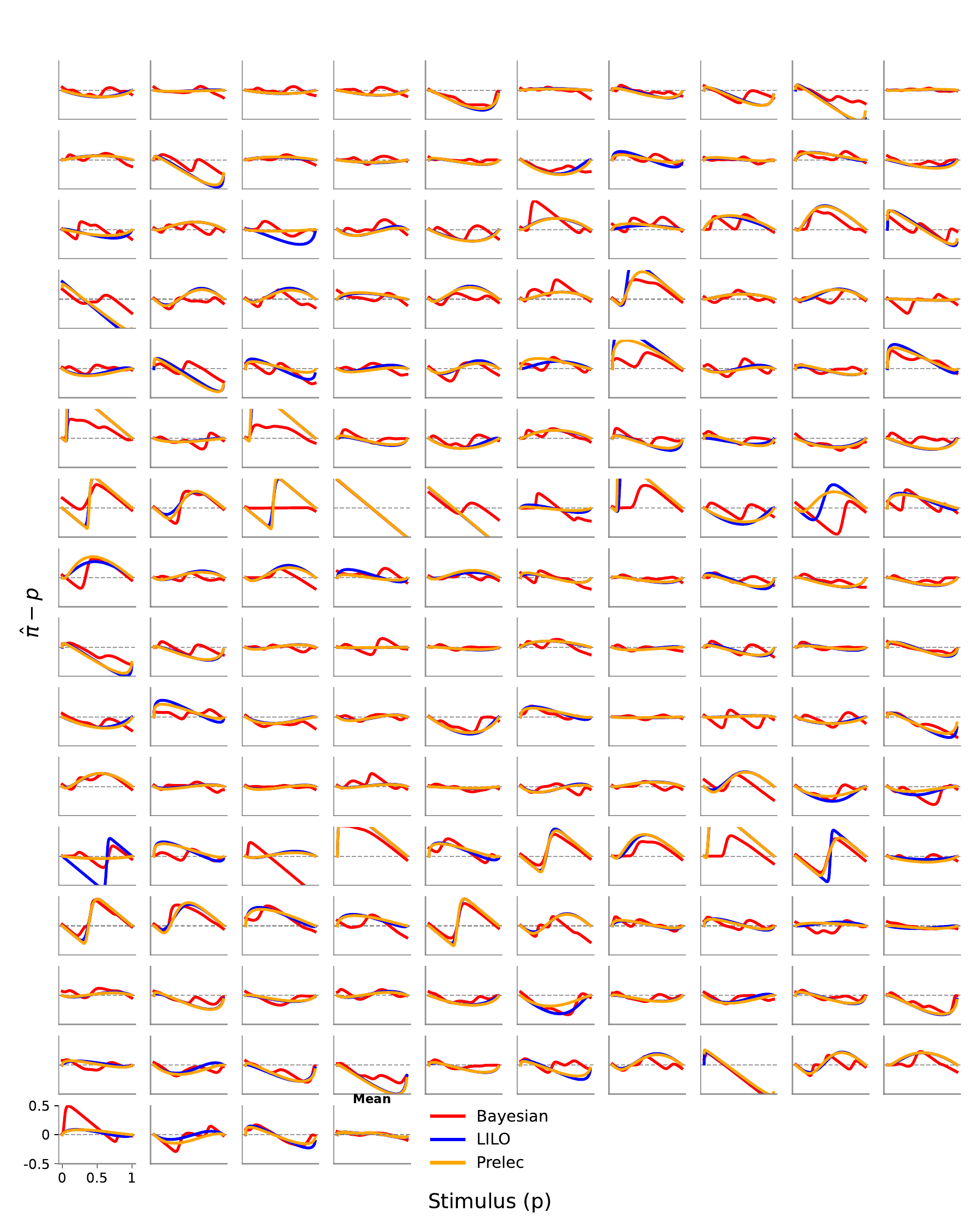}
\end{center}
\caption{\textbf{DMR Choice Task: }Per-subject bias of the LILO model, the Prelec model and the Bayesian model with free prior and free encoding on the DMR pricing task. Subjects S1 to 153 are from dataset CPC15 \citep{erev2017anomalies}.}
\label{fig:dmr-choice-perSubject-bias}
\end{figure}

\subsection{Details on Bayesian Model for Adaptation Bimodal Data}
\subsubsection{Bayesian Model Variants}\label{app:details-bayesian-jrf-bimodal-models}
We tested several model variants. Most aspects were the same as the variants used for the JRF task, with the exception that here we evaluated a bimodal prior instead of a uniform prior.  

\begin{itemize}
    \item \textbf{Priors:}
    \begin{itemize}
        \item \textbf{Bimodal Prior:}  
        This variant assumes a bimodal distribution over the range of possible grid points. The free parameters are the two means and their corresponding standard deviations.
        
        \item \textbf{Gaussian Prior, Freely Fitted Prior:}  
        These variants are the same as described in Appendix~\ref{app:details-bayesian-jrf-models}. 
    \end{itemize}
    
    \item \textbf{Encodings:}  
    \begin{itemize}
        \item \textbf{Unbounded Log-odds Encoding, Bounded Log-odds Encoding, Prior-matched Encoding,  Freely Fitted Encoding:} Same as in Appendix~\ref{app:details-bayesian-jrf-models}.
    \end{itemize}
\end{itemize}

\subsubsection{Performance Comparison of All Model variants}\label{app:performance-adaptation}

For the Adaptation Task, model comparison again shows a clear advantage for Bayesian variants over BLO. When measured by both held-out $\Delta$NLL (Figure~\ref{fig:jrf-dnll-full-both-analytical-categorical-bimodal}) and $\Delta$AICc (Figure~\ref{fig:jrf-daicc-full-both-analytical-categorical-bimodal}), Bayesian models consistently achieve substantially lower scores, indicating a better quantitative account of the observed responses.

Among the Bayesian models, those with bimodal priors provide a closer fit than their Gaussian prior counterparts, reflecting the bimodal structure of the stimulus distribution. The freely fitted prior yields the best performance overall, suggesting that allowing the prior to flexibly adapt to the empirical distribution of stimuli gives the most accurate description of subjects’ behavior. We also find that model with matched prior and encoding perform particularly well.

One thing worth noting is that the model ranking under the categorical likelihood differs from that under the analytical likelihood. FreeP+FreeE achieves the best fit under the analytical likelihood but the worst under the categorical likelihood; in the JRF and DMR tasks, the two metrics agree (Appendix Figures~\ref{fig:jrf-dnll-full-both-analytical-categorical} and~\ref{fig:dmr-dnll-full-both-analytical-categorical}).

We do not have a complete account of this discrepancy. One possibility is that it arises from the interaction between FreeP+FreeE's flexibility and the bimodal stimulus distribution: with sharply peaked stimulus modes, FreeP+FreeE can drive motor noise $\sigma_{\text{motor}}$ to small values that fit response clusters tightly under the continuous analytical likelihood, but this advantage is diminished under the discretized categorical likelihood.

We emphasize that this metric sensitivity does not affect the qualitative findings of the bimodal experiment: across all variants and both metrics, the recovered prior tracks the bimodal stimulus distribution (Figure~\ref{fig:bimodal-bias-dcv}B), and the recovered encoding remains U-shaped (Figure~\ref{fig:jrf-FI-comparison-bimodal}). The bimodal experiment's role in the paper is to demonstrate prior-encoding dissociation, which holds regardless of which likelihood metric is used for model ranking.

\begin{figure}[H]
    \centering
    \begin{subfigure}{0.48\linewidth}
        \centering
        \includegraphics[width=\linewidth]{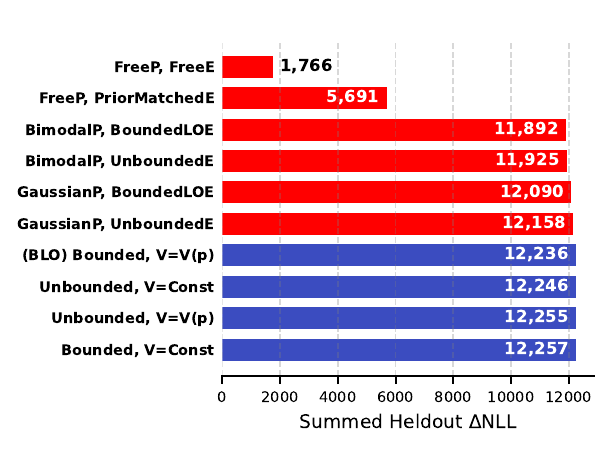}
        \caption{Analytical}
        \label{fig:jrf-dnll-analytical-bimodal}
    \end{subfigure}
    \hfill
    \begin{subfigure}{0.48\linewidth}
        \centering
        \includegraphics[width=\linewidth]{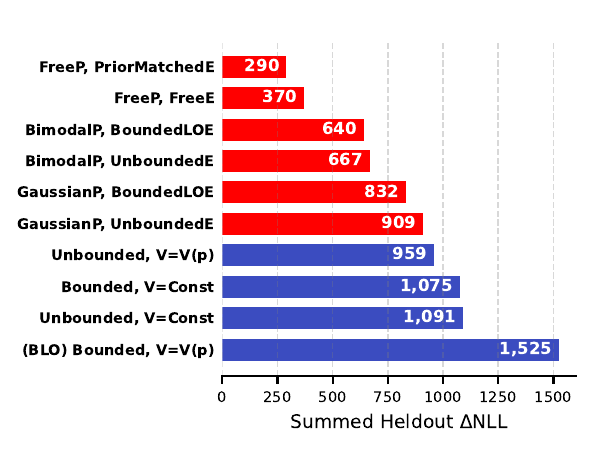}
        \caption{Categorical}
        \label{fig:jrf-dnll-categorical-bimodal}
    \end{subfigure}
    \caption{\textbf{Adaptation Task:} Performance of the models, measured by the Heldout $\Delta$NLL metric. Stimuli follows bimodal distribution. See Appendix~\ref{app:details-bayesian-jrf-models} for the shorthands for Bayesian models (red). We refer to \cite{zhang2020} for the shorthands for their model variants (blue). Analytical and categorical fitting methods are explained in Appendix~\ref{app:analytical-categorical-fitting}. The results shown in the main text correspond to Analytical. Note that in the bimodal task, FreeP+FreeE achieves the best fit under analytical likelihood but the worst under categorical (see Appendix~\ref{app:performance-adaptation} for discussion)}
    \label{fig:jrf-dnll-full-both-analytical-categorical-bimodal}
\end{figure}

\begin{figure}[H]
    \centering
    \begin{subfigure}{0.48\linewidth}
        \centering
        \includegraphics[width=\linewidth]{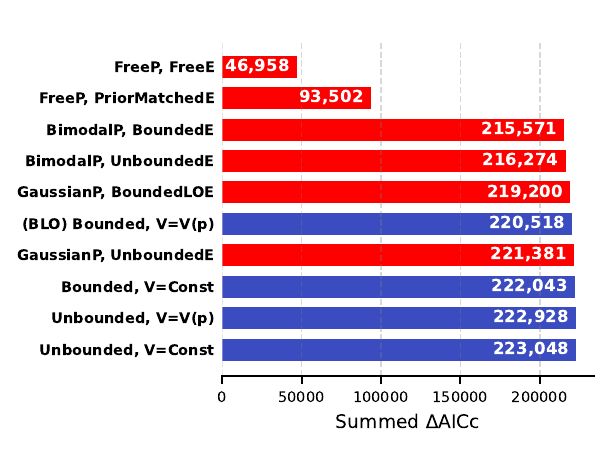}
        \caption{Analytical}
        \label{fig:jrf-daicc-analytical-bimodal}
    \end{subfigure}
    \hfill
    \begin{subfigure}{0.48\linewidth}
        \centering
        \includegraphics[width=\linewidth]{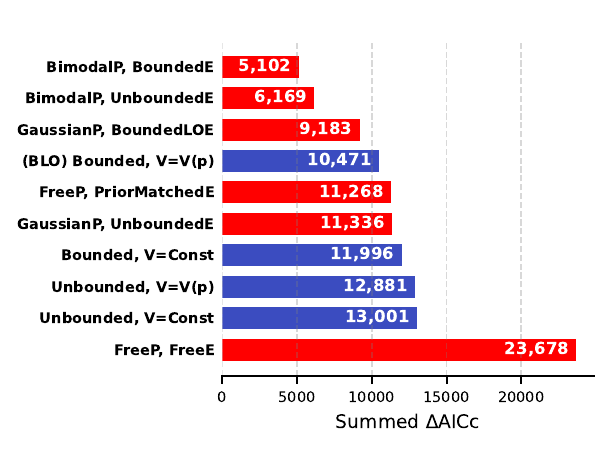}
        \caption{Categorical}
        \label{fig:jrf-daicc-categorical-bimodal}
    \end{subfigure}
    \caption{\textbf{Adaptation Task:} Performance of the models, measured by the Heldout $\Delta$AICc metric. Stimuli follows bimodal distribution. See Appendix~\ref{app:details-bayesian-jrf-models} for the shorthands for Bayesian models (red). We refer to \cite{zhang2020} for the shorthands for their model variants (blue). Analytical and categorical fitting methods are explained in Appendix~\ref{app:analytical-categorical-fitting}. The results shown in the main text correspond to Analytical. Note that in the bimodal task, FreeP+FreeE achieves the best fit under analytical likelihood but the worst under categorical (see Appendix~\ref{app:performance-adaptation} for discussion)
       }
    \label{fig:jrf-daicc-full-both-analytical-categorical-bimodal}
\end{figure}

\subsubsection{Robustness of Model Comparison via Bootstrapping}
\label{app:adaptation-bootstrap}

We performed a subject-level bootstrap analysis on the aggregate model comparison results (Figure~\ref{fig:bimodal-bias-dcv}D). Specifically, we resampled the 26 subjects with replacement 10,000 times. For each bootstrap sample, we calculated the Summed Held-out $\Delta$NLL difference between our freely fitted Bayesian model (FreeP, FreeE) and each competing model.

Table~\ref{tab:bootstrap} summarizes the results. We report the Bootstrap Support, the percentage of resampled datasets in which the FreeP, FreeE model achieved a lower Summed NLL than the competitor. We also report the 95\% Confidence Interval (CI) of the difference in NLL (competing model minus FreeP, FreeE).

The freely fitted Bayesian model significantly outperformed all competing accounts, including the Efficient Coding model (PriorMatchedE) and the BLO model, achieving a bootstrap support greater than $98\%$. The 95\% confidence intervals for the NLL differences  exclude zero for all comparisons. This analysis confirms the superiority of our Bayesian account in this task.

\begin{table}[h]
\caption{\textbf{Adaptation Task:} Bootstrap analysis of model comparison in the (10,000 iterations). The freely fitted Bayesian model (FreeP, FreeE) is the reference. Positive $\Delta$NLL values mean the reference model fits better. CI is Confidence Interval.}
\label{tab:bootstrap-results}
\begin{center}
\begin{tabular}{l c c}
\multicolumn{1}{c}{\bf Competitor Model Variants}  &\multicolumn{1}{c}{\bf Bootstrap Support} &\multicolumn{1}{c}{\bf 95\% CI of Difference ($\Delta$NLL)} 
\\ \hline \\
FreeP, PriorMatchedE & 98.1\% & $[233.43, 7550.74]$  \\
BimodalP, BoundedLOE & 100.0\% & $[7177.88, 13032.97]$ \\
BimodalP, UnboundedE & 100.0\% & $[7216.18, 13054.51]$ \\
GaussianP, BoundedLOE & 100.0\% & $[7380.49, 13224.09]$\\
(BLO) Bounded, $V=V(p)$ & 100.0\% & $[7518.55, 13399.81]$ \\
GaussianP, UnboundedE & 100.0\% & $[7466.34, 13271.41]$\\
Bounded, $V=\text{Const}$ & 100.0\% & $[7532.97, 13426.86]$ \\
Unbounded, $V=V(p)$ & 100.0\% & $[7518.60, 13437.80]$\\
Unbounded, $V=\text{Const}$ & 100.0\% & $[7510.38, 13429.20]$  \\
\label{tab:bootstrap}
\end{tabular}
\end{center}
\end{table}

\subsubsection{Analysis of Fitted Prior and Resources}
Figure~\ref{fig:jrf-FI-comparison-bimodal} shows that the group-level resources  largely retain the characteristic U-shape across encoding variants, consistent with the JRF task results. This indicates that encoding efficiency remains highest near the extremes of the probability scale. Because each subject was exposed to a different bimodal stimulus distribution, we do not plot a group-level prior. Instead, Figure~\ref{fig:jrf-priors-bimodal-perSubject} shows the freely fitted prior for each subject, which for most subjects successfully adapts to the underlying bimodal distribution.


\begin{figure}[H]
\begin{center}
\includegraphics[width=0.8\linewidth]{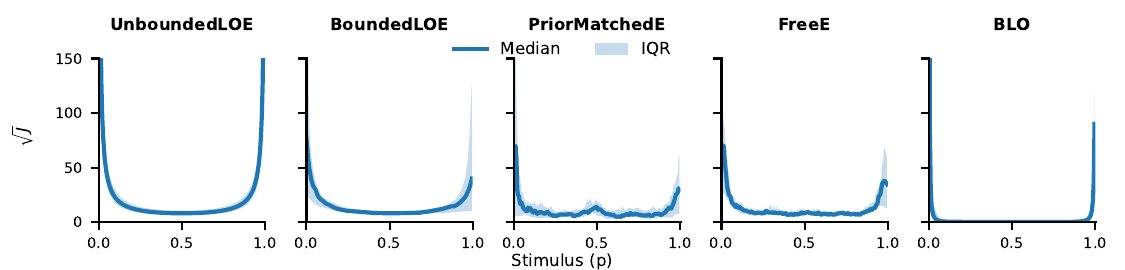}
\end{center}
\caption{
\textbf{Adaptation Task:} Group-level resources ($\sqrt{J(p)}$) for the three encoding components evaluated in this task. The solid line shows the group median, and the shaded area indicates the interquartile range (IQR). Plotting details are provided in Appendix~\ref{app:plotting}.}
\label{fig:jrf-FI-comparison-bimodal}
\end{figure}

\begin{figure}[H]
\begin{center}
\includegraphics[width=0.8\linewidth]{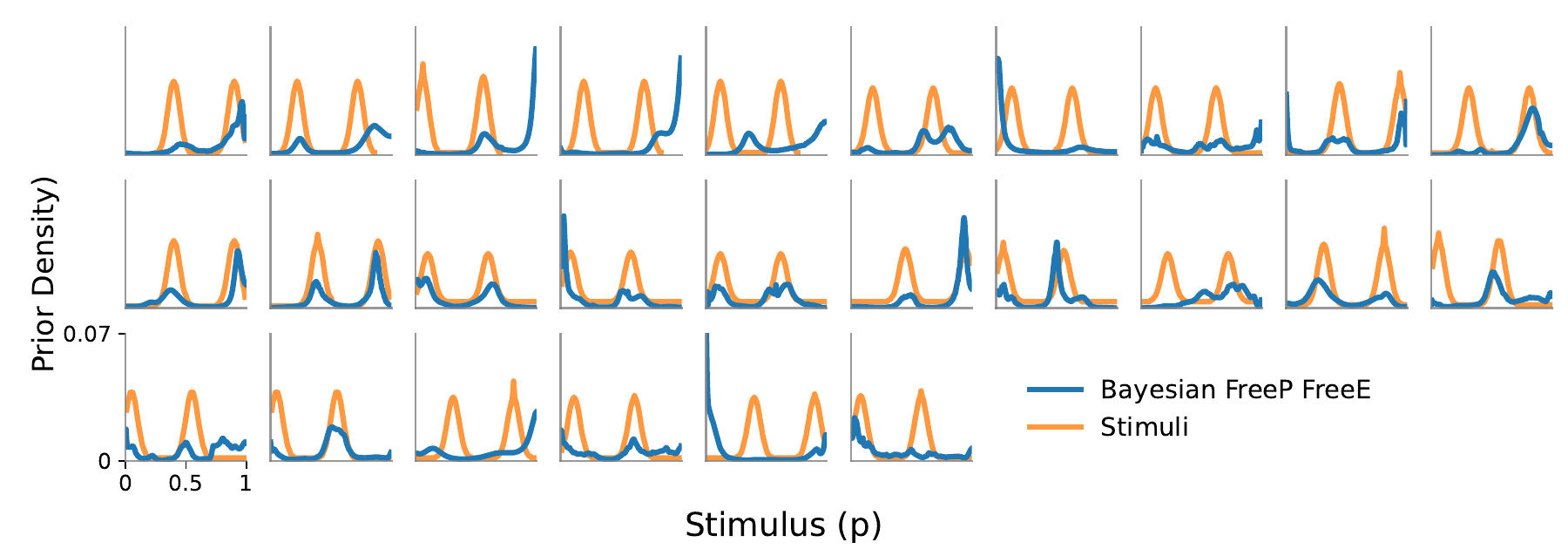}
\end{center}
\caption{\textbf{Adaptation Task:} Fitted priors for individual subjects. The freely fitted prior (blue) successfully adapts to the bimodal stimulus distribution (orange).}
\label{fig:jrf-priors-bimodal-perSubject}
\end{figure}

\subsubsection{Analysis of Bias and Variance}
Figure~\ref{fig:jrf-bias-bimodal-perSubject} shows the bias of each subject. The BLO model captures broad trends in bias but lacks the flexibility to account for the more complex, stimulus-dependent bias patterns observed in human data. In contrast, the Bayesian model with freely fitted prior and encoding provides a closer fit to subject-level biases, particularly in regions where deviations from linearity are more pronounced.

In the variability figure(Figure~\ref{fig:jrf-var-bimodal-perSubject})The Bayesian model also provides a superior account of variability compared to BLO. While BLO could capture the overall variability magnitude, the Bayesian model models both the overall magnitude and the shape of variability more accurately, and aligns more closely with the human data.

\begin{figure}[H]
\begin{center}
\includegraphics[width=0.8\linewidth]{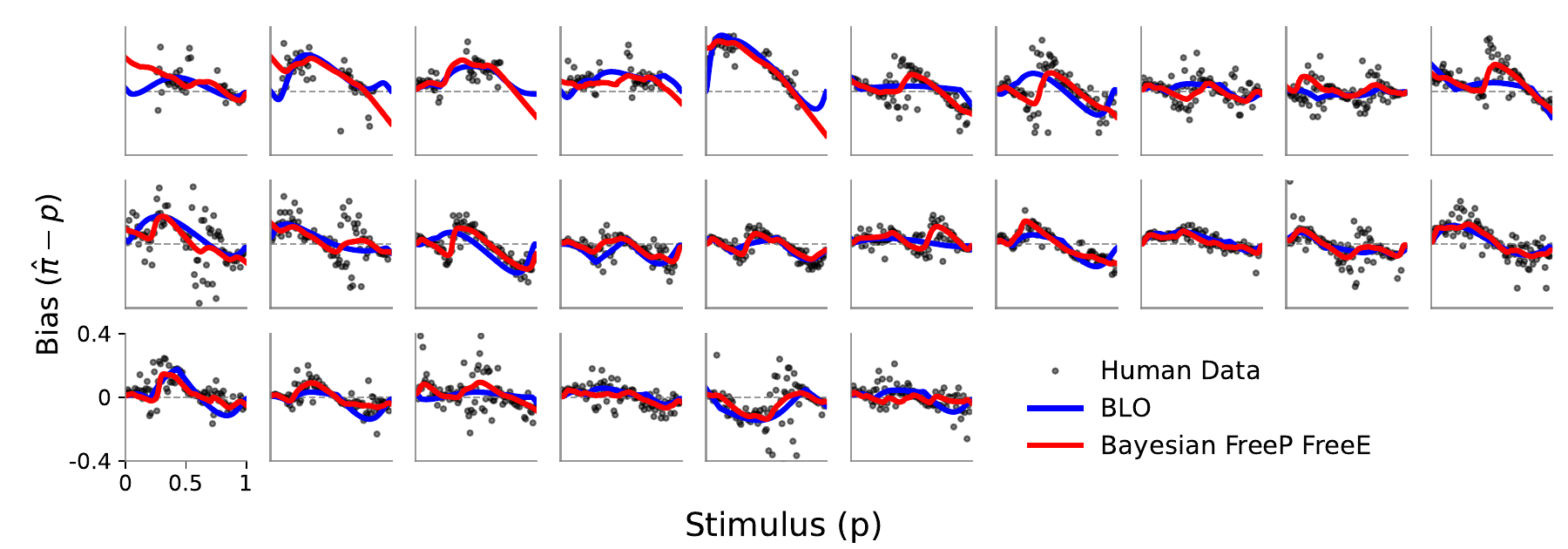}
\end{center}
\caption{\textbf{Adaptation Task:} Per-subject bias of the non-parametric estimates from the data, the BLO model and the Bayesian model with free prior and free encoding on the Adaptation task. Subjects S1 to S26 are from our collected dataset.}
\label{fig:jrf-bias-bimodal-perSubject}
\end{figure}

\begin{figure}[H]
\begin{center}
\includegraphics[width=0.8\linewidth]{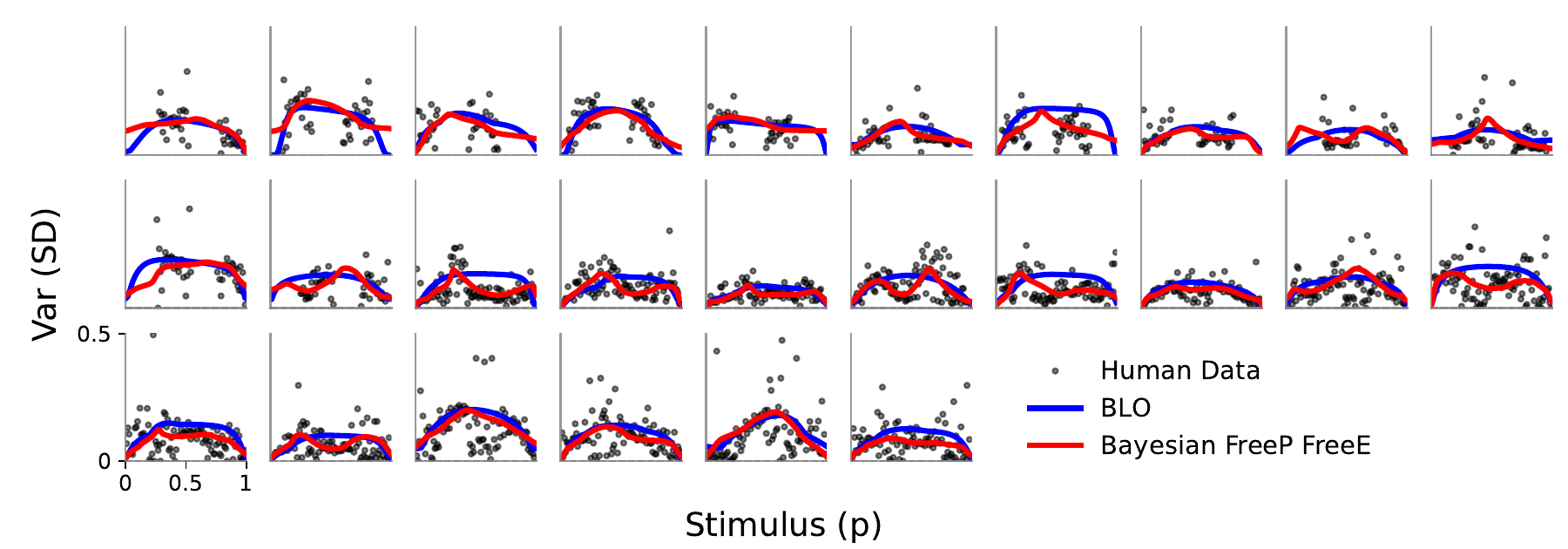}
\end{center}
\caption{\textbf{Adaptation Task:} Per-subject variability of the non-parametric estimates from the data, the BLO model and the Bayesian model with free prior and free encoding on the Adaptation task. Subjects S1 to S26 are from our collected dataset.}
\label{fig:jrf-var-bimodal-perSubject}
\end{figure}

\subsection{Details on BLO Model for JRF Task and Adapatation Data}\label{app:details-blo-jrf}


\subsubsection{Application of BLO Model to JRF Task}

While initially developed for decision under risk (DMR) tasks, the BLO model can be adapted to account for perception in judgment of relative frequency (JRF) tasks. In the JRF context, the core principles of bounded log-odds encoding, truncation, and variance compensation remain, but the model additionally accounts for sensory noise introduced by the observer's sampling of stimuli.

In JRF tasks, it is assumed that people do not process every element in the display. Instead, they sample a subset of items from the whole display. This sampling introduces additional noise in the observer’s estimate. Consequently, the variance of the sample-based estimate, $V(\hat{p})$, is modeled with a finite population correction: If the total number of items is $N$ and the observer samples $n_s$ items, the variance of the sample-based estimate is given by:
\begin{equation}\label{eq:vp}
    V(\hat{p}) = \frac{p(1-p)}{n_s}\cdot \frac{N - n_s}{N - 1}
\end{equation},
where $N$ is the total number of dots and $n_s$ is a free parameter.
We note that sampling of the items is not implemented in the published implementation from \cite{zhang2020}, and we correspondingly do not attempt to model it in our reimplementation of the BLO model.

\subsubsection{Reimplementation of BLO and LLO Models Variants}\label{app:details-blo-jrf-models}
We reimplemented the key model variants proposed by \cite{zhang2020} for JRF. Specifically, we focused on four variants, Bounded Log-Odds(BLO) and Linear Log-Odds (LLO) models, each combined with either a constant perceptual variance (V=Const) or a proportion-dependent variance compensation (V=V(p)). 

Following \cite{zhang2020}, we fit these models to the data per subject by minimizing the negative log-likelihood of the observed responses given the model's predicitions. 
Our reimplementation achieved nearly identical estimated parameter values and model fits to those reported by \cite{zhang2020} for JRF.

For all these reimplemented models, the optimization was performed using a gradient-based approach(Adam or SignSGD) rather than the original Nelder-Mead optimizer(fminsearchbnd) used in their implmentation. We set parameter bounds for our optimization based on \cite{zhang2020}'s settings. the hyperbolic tangent (tanh) function was applied to parameters such as $\Delta^-$ and 
$\Delta^+$ , while the exponential function was used for strictly positive parameters like $\kappa$. Parameters not naturally constrained were optimized directly.

We successfully achieved nearly identical estimated parameter values and model fits to those reported by \cite{zhang2020} for the dot counting task, thereby validating our reimplementation.

\subsection{Details for BLO on DMR Pricing Task}\label{app:details-blo-dmr}

\subsubsection{Application of BLO Model to DMR Task}

The uncertainty or variance associated with this internal encoding is modeled as being proportional to the binomial variance.
\begin{equation}
    V(\hat{p}) \propto p(1 - p)
\end{equation}

Following \cite{zhang2020}, the estimate $\hat{\pi}(p)$ is integrated into Cumulative Prospect Theory  (CITE) to predict choice behavior;  the certainty equivalent (CE) for a two-outcome lottery $(x_1, p; x_2, 1 - p)$ is given by:
\begin{equation}
\label{eq:CE}
    \text{CE} = U^{-1}\left[U(x_1) \cdot \hat{\pi}(p) + U(x_2) \cdot (1 - \hat{\pi}(p))\right] + \varepsilon_{CE}
\end{equation}
Here, $U(\cdot) = x^\alpha$ is the utility function, and $\varepsilon_{CE}$ represents Gaussian noise on the CE scale.

\subsubsection{Reimplementation of BLO and LLO Models Variants}\label{app:details-blo-dmr-models}
We first reimplemented the four main parametric models proposed by \cite{zhang2020} for the DMR task. These included the BLO and LLO models, each combined with either a proportion-dependent variance compensation (V(p)) or a constant variance (Const V). As with the JRF task analysis, these models were fitted to the data for each subject individually. Our implementation utilized gradient-based optimizers (specifically, Adam or SignSGD) for parameter optimization.

We observed that the negative log-likelihood (NLL) values obtained from our reimplemented models were nearly identical to Zhang's reported fitted results for most variants. However, for the BLO + V(p) model, our reimplemented model showed, on average, a 7.45 higher loss (negative log-likelihood) compared to Zhang's original results; though this does not impact quantitative model comparison with the Bayesian model. 

\subsection{Methods for Generating Empirical Bias Decomposition}
\label{app:bias-decomposition-empirical}

The analytical decomposition in Theorem~\ref{thm:bayesian-bias} is valid 
to $\mathcal{O}(\sigma^4)$ and assumes a non-truncated stimulus space. 
For the bias decomposition figures (e.g., Figure~\ref{fig:jrf-bias-decomposition-resources-daicc}, Panel C), 
we instead report an \emph{empirical} decomposition that is exact at 
the fitted noise level and treats the truncated stimulus space $[0,1]$ 
explicitly. The three components map onto the analytical ones (Boundary 
Regression, Likelihood Repulsion, Prior Attraction) but are obtained 
from the fitted Bayesian observer through counterfactual simulations, 
without the small-$\sigma$ expansion. We illustrate the procedure using 
the JRF decomposition (Figure~\ref{fig:jrf-bias-decomposition-resources-daicc}); the same procedure 
applies to all other figures.

\paragraph{Setup.}
For each subject and each Bayesian variant, we use the maximum-likelihood 
parameters $(F, P_{\text{prior}}, \sigma)$ from the cross-validation fit. 
The forward model produces the bias $\hat{p}(p) - p$ on the 200-point 
grid over $[0,1]$.

\paragraph{Counterfactual simulations.}
To isolate each mechanism, we evaluate the fitted Bayesian observer in 
three configurations:
\begin{enumerate}
    \item \textbf{Full model on $[0,1]$.}
          The total bias $b_{\text{tot}}(p)$ from the fitted model on 
          the original grid.
    \item \textbf{Full model on expanded grid.} The grid is widened to 
          $[-1,2]$ with $F$ and $P_{\text{prior}}$ extrapolated by their 
          endpoint values outside $[0,1]$. The resulting bias 
          $b_{\text{int}}(p)$ removes boundary effects: the encoding 
          can now spread beyond the original support.
    \item \textbf{Uniform-prior model on expanded grid.} Same as (2) but 
          with $P_{\text{prior}}$ replaced by a uniform distribution 
          on $[-1,2]$. The resulting bias $b_{\text{rep}}(p)$ contains 
          only the contribution from non-uniform encoding.
\end{enumerate}

\paragraph{Components.}
The three components are
\begin{align}
    \text{Likelihood Repulsion}(p) &= b_{\text{rep}}(p), \\
    \text{Prior Attraction}(p) &= b_{\text{int}}(p) - b_{\text{rep}}(p), \\
    \text{Boundary Regression}(p) &= b_{\text{tot}}(p) - b_{\text{int}}(p),
\end{align}
so that 
$b_{\text{tot}}(p) = \text{Likelihood Repulsion} + \text{Prior Attraction} + \text{Boundary Regression}$ 
by construction.

\paragraph{Mapping to the analytical decomposition.}
Replacing the prior with a uniform distribution removes the 
$\frac{1}{\mathcal{J}(p)}\frac{\mathrm{d}}{\mathrm{d}p}\log P_{\text{prior}}(p)$ 
term in Theorem~\ref{thm:bayesian-bias}, and expanding the grid removes 
the $A_{1,\sigma}(p)\,\mathrm{sign}(0.5-p)/\sqrt{\mathcal{J}(p)}$ term.

The same procedure is applied to all other tasks and model variants. 
Detailed bias decomposition figures across tasks and variants are 
provided in Appendix~\ref{app:cross-task-bias-decomposition}.

\subsection{Cross-Task Analysis of Bias Decomposition}\label{app:cross-task-bias-decomposition}

As shown in Section~\ref{sec:framework-decomposition}, we decomposed the total estimation bias ($\mathbb{E}[\hat{p}|p] - p$) into three components: Prior Attraction, Likelihood Repulsion, and Boundary Regression. 

We computed these components by the methods describe in Appendix Section~\ref{app:bias-decomposition-empirical}.

Figure~\ref{fig:dmr-bias-decomposition} illustrates the bias decomposition for the DMR Pricing task. Similar to JRF, the UniformE model exhibits zero likelihood repulsion. It attempts to compensate for this absence by inducing a heavily distorted Prior Attraction and relying on Boundary Regression. While this compensation allows for a moderate quantitative fit in pricing (a task where we compare the expected value instead of probability), it remains mechanistically distinct from the repulsion in U-shaped models, ann we seek for more evidence in Figure~\ref{fig:cpc15-bias-decomposition}, the bias decomposition for the DMR Choice task. As in other tasks, the Uniform model shows zero likelihood repulsion. However, unlike in the Pricing task, the Prior Attraction here remains relatively flat and fails to compensate for the missing repulsion. This flatness likely occurs because, in a binary choice task, increasing prior precision (to reduce variance) without sufficient sensory discriminability at the boundaries would only compress values and reduce the model's ability to predict preference reversals. The model thus defaults to a smoother prior to avoid introducing excessive bias error.

Figure~\ref{fig:bimodal-bias-decomposition} shows how the Bayesian model captures the shift in behavior under bimodal stimulus statistics. In Prior Attraction, the GaussianP model produces a simple, smooth attraction curve. In contrast, both the FreeP and BimodalP models exhibit multi-peaked attraction patterns. These fluctuations correspond to attraction forces directed toward the two modes of the bimodal stimulus distribution.

\begin{figure}[H]
\begin{center}
\includegraphics[width=0.5\linewidth]{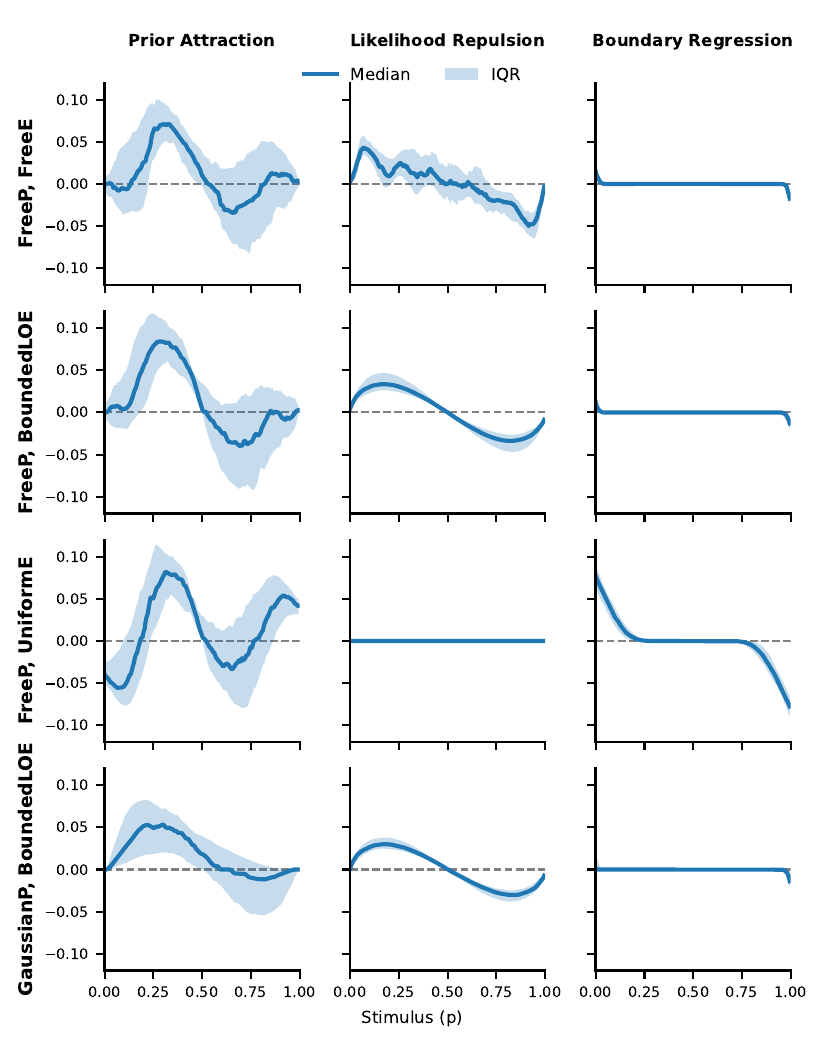} 
\end{center}
\caption{\textbf{JRF Task:} The comparison between bias decomposition of four Bayesian model variants. The solid line shows the group median, and the shaded area indicates the interquartile range (IQR). Plotting details are provided in Appendix~\ref{app:plotting}.
}
\label{fig:jrf-bias-decomposition}
\end{figure}

\begin{figure}[H]
\begin{center}
\includegraphics[width=0.5\linewidth]{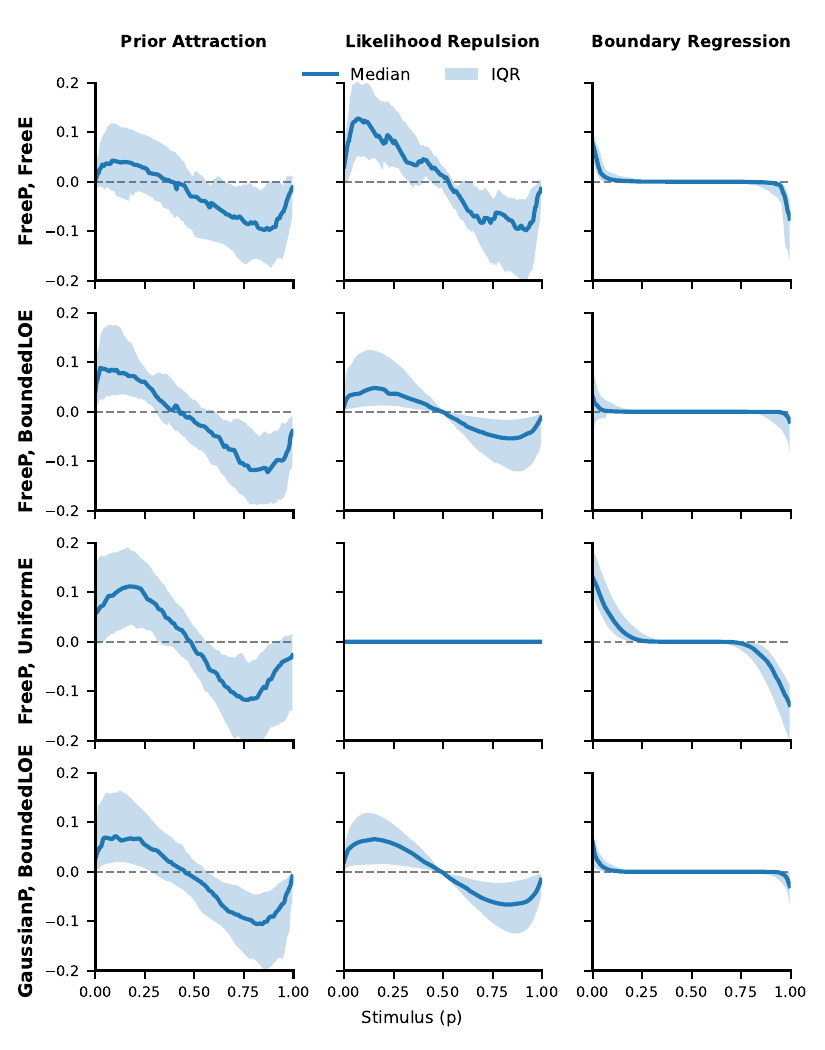}
\end{center}
\caption{\textbf{DMR Pricing Task:} Bias decomposition of four Bayesian model variants. The solid line shows the group median, and the shaded area indicates the interquartile range (IQR). Plotting details are provided in Appendix~\ref{app:plotting}.
}
\label{fig:dmr-bias-decomposition}
\end{figure}

\begin{figure}[H]
\begin{center}
\includegraphics[width=1.0\linewidth]{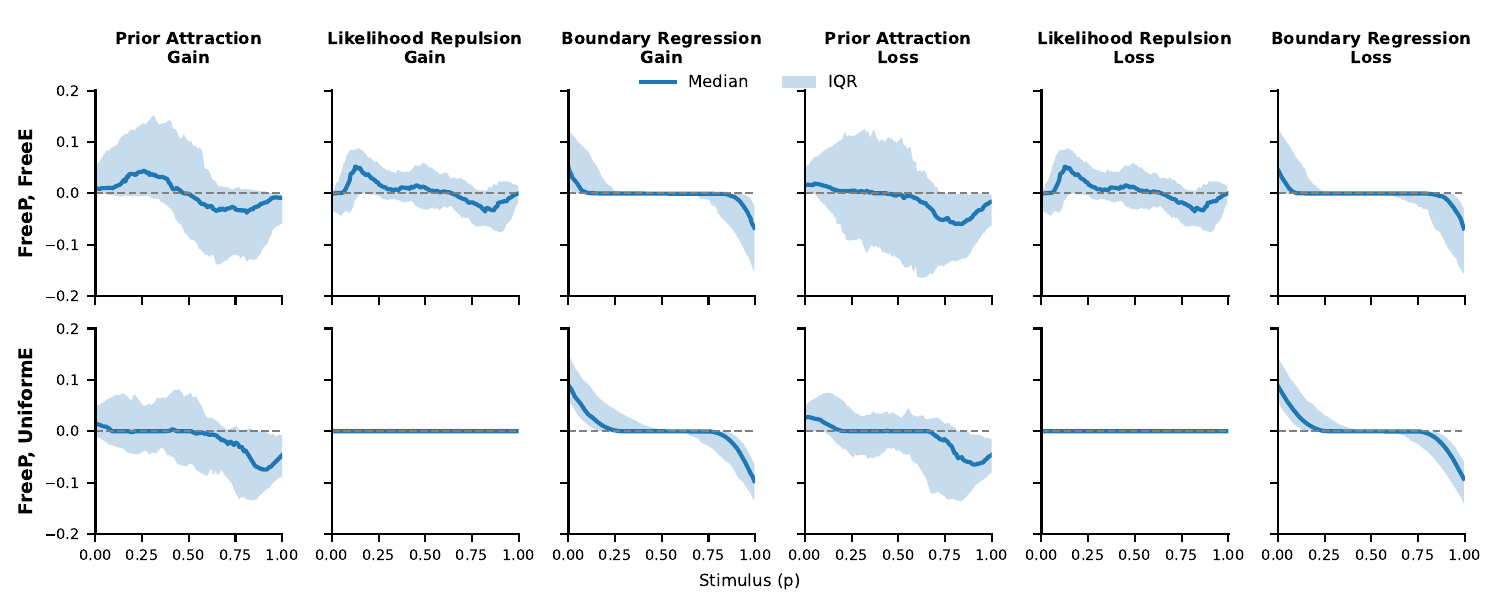}
\end{center}
\caption{\textbf{DMR Choice Task:} Bias decomposition of two Bayesian model variants. The solid line shows the group median, and the shaded area indicates the interquartile range (IQR). Plotting details are provided in Appendix~\ref{app:plotting}.
}
\label{fig:cpc15-bias-decomposition}
\end{figure}

\begin{figure}[H]
\begin{center}
\includegraphics[width=0.6\linewidth]{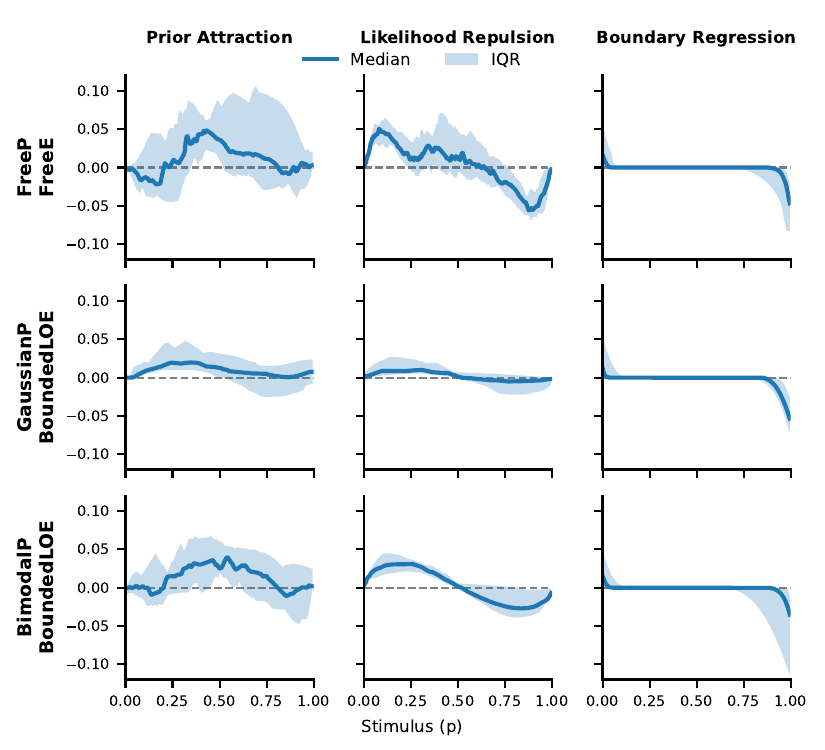}
\end{center}
\caption{\textbf{Adaptation Task:} Bias decomposition of three Bayesian model variants in Figure~\ref{fig:bimodal-bias-dcv}D. The solid line shows the group median, and the shaded area indicates the interquartile range (IQR). Plotting details are provided in Appendix~\ref{app:plotting}.}
\label{fig:bimodal-bias-decomposition}
\end{figure}

\subsection{Cross-Task Analysis of Resources and Bias}
The freely fitted priors and corresponding resources functions in three tasks are shown in Figure~\ref{fig:fi-across-3tasks}).  Priors differ noticeably between tasks, but resources consistently exhibits the U-shape, with highest sensitivity near the extremes.  Interestingly, in the DMR Choice tasks, the shape of the prior appears more closely aligned with the resources than in the other tasks. This prior–resources match is reminiscent of the account proposed by \cite{frydman2023source}, but the fact that it is not observed across all tasks suggests that such alignment is limited in generality.

\begin{figure}[H]
\begin{center}
\includegraphics[width=0.8\linewidth]{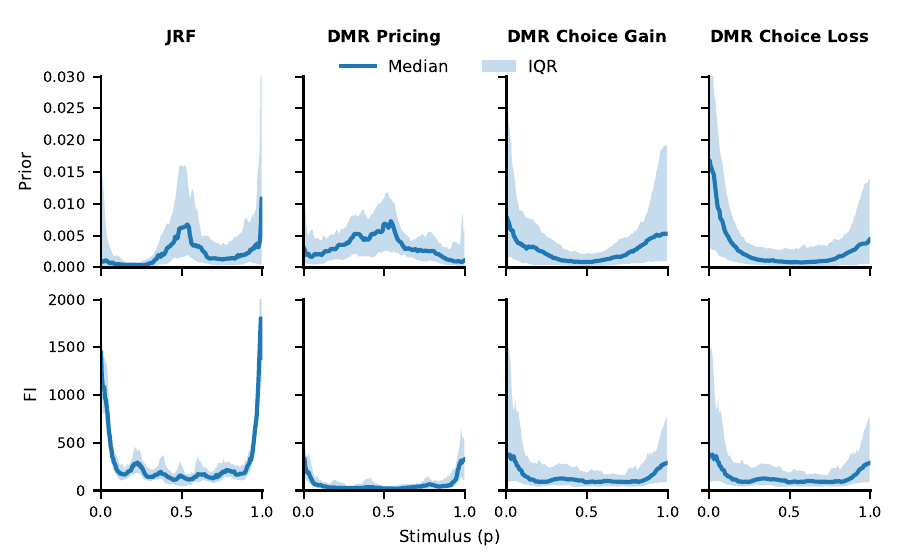}
\end{center}
\caption{\textbf{Across Tasks:}  Freely fitted priors (top) and resources $\sqrt{J}$ (bottom) across tasks.The solid line shows the group median, and the shaded area indicates the interquartile range (IQR). Plotting details are provided in Appendix~\ref{app:plotting}.  This figure is mentioned in main text section~\ref{sec:jrf}.}
\label{fig:fi-across-3tasks}
\end{figure}

\subsection{Details for Plotting Group-level Prior and Resources}\label{app:plotting}
For each group-level curve presented in the paper, we plot the median across subjects as the solid line and the interquartile range (IQR; 25th–75th percentile) as the shaded area. These quantities are computed pointwise on the stimulus grid, providing a summary of central tendency and between-subject variability.

\subsection{Methods for Calculating Bias and Variance}\label{app:methods-bias-var}

This section describes the procedures used to calculate bias and variance for the non-parametric data, the Bayesian model, and the BLO model for the bias and variance figures (for example, Figure~\ref{fig:jrf-bias-perSubject} and Figure~\ref{fig:jrf-var-perSubject}).

\subsubsection{Non-Parametric Estimation}\label{app:methods-bias-var-np}
The non-parametric estimate of relative frequency is defined as
\[
    \hat{\pi}_{NP}(p) = \frac{1}{m} \sum_{t=1}^{m} \hat{\pi}_{t}(p),
\]
where $\hat{\pi}_{t}(p)$ denotes the subject’s estimate on trial $t$, $t=1,2,\ldots,m$.  

\paragraph{Bias.}  
Bias is the difference between the mean estimate and the true probability:
\[
\text{Bias}_{NP}(p) = \hat{\pi}_{NP}(p) - p.
\]

\paragraph{Variance.}  
The variance is computed as the sample variance of $\hat{\pi}_{t}(p)$ across trials of the same $p$.

\subsubsection{BLO Model}\label{app:methods-bias-var-blo}
The BLO model produces deterministic predictions $\hat{\pi}_{BLO}(p,N)$ that depend on numerosity $N$. To obtain a single prediction per probability $p$, we average across the five numerosity conditions ($N = \{200,300,400,500,600\}$):
\[
    \mathbb{E}[\hat{p}]_\text{BLO} = \frac{1}{5} \sum_{N \in \{200,300,400,500,600\}} \hat{\pi}_\text{BLO}(p,N).
\]

\paragraph{Bias.}  
Bias is then defined as
\[
\text{Bias}_\text{BLO}(p) = \mathbb{E}[\hat{p}]_\text{BLO} - p.
\]
\paragraph{Variance.}  
Variance in BLO arises from Gaussian noise in log-odds space, $\epsilon_\lambda \sim \mathcal{N}(0,\sigma_{\lambda}^2)$.  
Mapping back into probability space requires a Jacobian transformation:
\[
    f(p) = \frac{1}{p(1-p)} \cdot \frac{1}{\sqrt{2\pi\sigma_\lambda^2}}
    \exp\!\left(-\frac{(\lambda(p) - \hat{\Lambda}_\omega)^2}{2\sigma_\lambda^2}\right).
\]
The probability distribution is normalized, and the variance is computed as
\[
\mathrm{Var}_\text{BLO}(p) = \mathbb{E}[P^2] - \big(\mathbb{E}[P]\big)^2,
\]
where expectations are taken with respect to $f(p)$.

\subsubsection{Bayesian Model}\label{app:methods-bias-var-bayes}

For the Bayesian model, the estimate given a measurement $m$ is $\hat{p}(m)$.  
The mean estimate for stimulus $p$ is
\[
\hat{p}(p) = \sum_{m} \hat{p}(m)\, P(m \mid p).
\]

\paragraph{Bias.}  
Bias is defined as
\[
\text{Bias}_\text{Bayesian}(p) = \hat{p}(p) - p.
\]

\paragraph{Variance.}  
Variance has two components:
\[
\mathrm{Var}_\text{Bayesian}(\hat{p}(m)\mid p) 
= \underbrace{\sum_m P(m \mid p)\,\big(\hat{p}(m) - \hat{p}(p)\big)^2}_{\text{sensory variance}}
+ \underbrace{\sigma^2_{\text{motor}}}_{\text{motor variance}}.
\]

\end{document}